\documentclass[a4paper,reqno, 11pt]{amsart}
\usepackage{graphicx,color}
\usepackage{hyperref}
\usepackage{footmisc}
\usepackage{amsmath}
\usepackage{tikzsymbols}
\usepackage{amsfonts, dsfont}
\usepackage{amssymb}
\usepackage{bbm}
\usepackage{epstopdf}
\usepackage{color}
\numberwithin{equation}{section}

\let\a=\alpha \let\b=\beta  \let\g=\gamma  \let\d=\delta \let\e=\varepsilon
\let\z=\zeta  \let\h=\eta   \let\th=\theta \let\k=\kappa \let\l=\lambda
\let\m=\mu    \let\n=\nu             \let\r=\rho
\let\s=\sigma \let\t=\tau   \let\f=\varphi 
   \let\o=\omega
 \let\D=\Delta  \let\L=\Lambda

\newcommand{\VV}{{\mathcal V}}

\newcommand{\xx}{{\bf x}}

\newcommand{\bt}{{\boldsymbol{\theta}}}

\def\PP{{\mathcal P}}
\newcommand{\WW}{{\mathcal W}}
\newcommand{\LL}{{\mathcal L}}
\newcommand{\RR}{{\mathcal R}}

\newtheorem{Theorem}{Theorem}
\newtheorem*{MT}{ Main Theorem}

\newtheorem{Remark}{Remark}

\newtheorem{Assumption}{Assumption}

\newtheorem{Proposition}{Proposition}

\def\nn{\nonumber}

\def\\{\hfill\break}
\def\={:=}
\let\io=\infty
\let\0=\noindent

\def\media#1{{\langle#1\rangle}}

\let\dpr=\partial
\def\sign{{\rm sign}}

\def\tende#1{\,\vtop{\ialign{##\crcr\rightarrowfill\crcr\noalign{\kern-1pt
    \nointerlineskip} \hskip3.pt${\scriptstyle #1}$\hskip3.pt\crcr}}\,}
\def\otto{\,{\kern-1.truept\leftarrow\kern-5.truept\to\kern-1.truept}\,}

\def\to{\rightarrow}

\def\qed{\hfill\raise1pt\hbox{\vrule height5pt width5pt depth0pt}}

\def\ul#1{{\underline#1}}

\def\V#1{{\bf#1}}
\def\nn{\nonumber}

\def\be{\begin{equation}}
\def\ee{\end{equation}}
\def\bea{\begin{eqnarray}}
\def\eea{\end{eqnarray}}

\begin{document}
\title[Universality for non-integrable dimers]{Non-integrable dimers: Universal fluctuations of tilted height profiles}
\author{Alessandro Giuliani}
\address{Dipartimento di Matematica e Fisica Universit\`a di Roma Tre \\ \small{
\small{L.go S. L. Murialdo 1, 00146 Roma, Italy}}}
\email{giuliani@mat.uniroma3.it}
\author{Vieri Mastropietro}
\address{Dipartimento di Matematica, Universit\`a di Milano \\
  \small{Via Saldini, 50, I-20133 Milano, Italy }}
\email{vieri.mastropietro@unimi.it}
\author{Fabio Lucio Toninelli}
\address{Univ Lyon, CNRS,  Universit\'e Claude Bernard Lyon 1\\
\small{UMR 5208, Institut Camille Jordan,
69622 Villeurbanne cedex, France}}
\email{toninelli@math.univ-lyon1.fr}

\begin{abstract} We study a class of close-packed dimer models on the square lattice, in the 
presence of small but extensive perturbations that make them non-determinantal. Examples include 
the 6-vertex model close to the free-fermion point, and the dimer model with plaquette interaction
previously analyzed in \cite{A,AL,GMT17a,GMT17b}. By tuning the edge weights, we can impose a non-zero average tilt for the height function, 
so that the considered models are in general not symmetric under discrete rotations and reflections.
In the determinantal case, height fluctuations in the massless (or `liquid') phase scale to a Gaussian log-correlated field and their
amplitude is a universal constant, independent of the tilt. When the
perturbation strength $\lambda$ is sufficiently small we prove, by fermionic constructive
Renormalization Group methods, that log-correlations survive, with amplitude $A$ that, generically, depends non-trivially
and non-universally on $\lambda$ and on the tilt. On the other hand, $A$ satisfies a universal scaling relation 
(`Haldane' or `Kadanoff' relation), saying that it equals the anomalous exponent of the 
dimer-dimer correlation. 
\end{abstract}

\maketitle

\footnotetext{\copyright\, 2019 by the authors. This paper may be reproduced, in its
entirety, for non-commercial purposes. }

\section{Introduction}
The question of {\it universality}, that is the independence of the critical properties of macroscopic systems 
from the microscopic details of the underlying model Hamiltonian, is a central issue in
statistical physics, whose mathematical understanding is largely incomplete. 
A convenient framework where it can be studied is that of \emph{planar dimer models}, which exhibit a rich critical behavior: algebraic decay of
correlations, conformal invariance, {and so on}.  The dimer model on a
bipartite planar lattice is integrable and, more precisely, determinantal (also called `free fermionic'): its correlation functions are given by
suitable minors of the so-called inverse Kasteleyn matrix \cite{Ka}.
The model is parametrized by edge weights $\underline t$ and has a
non-trivial phase diagram. By varying $\underline t$, one can impose
an average non-zero tilt $\rho$ for the height field. A central object
of the dimer model is the so-called characteristic polynomial
$P(z,w)$, where $z,w$ are complex variables. For instance, the infinite-volume free energy is given by an
integral of $\log |P(z,w)|$ over the torus $\mathbb
T=\{|z|=|w|=1\}$. Also, the large-distance decay of correlations is
dictated by the so-called spectral curve, i.e. the algebraic curve
$\mathcal C(P)=\{(z,w)\in \mathbb C^2:P(z,w)=0\}$. When the edge
weights are such that the spectral curve intersects $\mathbb T$
transversally one is in the ``liquid'' or ``massless'' phase, where
the two-point dimer-dimer correlation of the model decays like the
inverse distance squared. Correspondingly the height field scales to a
Gaussian Free Field (GFF) and the variance grows like the logarithm of
the distance times $1/\pi^2$. Remarkably, this pre-factor is independent of
the weights $\underline t$ and of the specific choice of the bipartite periodic planar
lattice.  This is related \cite{KOS} to the fact that $\mathcal C(P)$
is a so-called Harnack curve. Summarizing, in the massless phase the
scaling limit of height fluctuations of the dimer model is universal, in a very strong sense: 
the limit is always Gaussian, with logarithmic growth of the variance; moreover, the pre-factor in front of the logarithm in the variance is independent of the
details of the underlying microscopic structure (edge weights and   lattice).

\medskip

The previous results heavily rely on the determinantal structure of
the model, but universality is believed to hold much more
generally. Motivated by this, we consider weak, translation-invariant,
perturbations of the dimer model (for simplicity, we restrict to the
square lattice).  Generically, as soon as we switch on the
perturbation, the determinantal structure provided by Kasteleyn's
theory breaks down. Two particular examples of perturbed,
non-determinantal, models that we consider are: the 6-vertex model
with general weights $a_1,\dots, a_6$, in the disordered phase, close
to, but not exactly at, the free-fermion point; and the dimer model
with plaquette interaction, originally introduced in \cite{HP} and recently reconsidered in \cite{A,AL,PLF} in the context of quantum dimer models. There
is a basic difference between  these two cases: the 6-vertex model, even
if non-determinantal, is still solvable via Bethe Ansatz (BA), see
\cite{Ba} and reference therein (note that the BA solution is not as
explicit as the Kasteleyn solution of standard dimers: only a few
thermodynamic functions can be explicitly computed).  On the other
hand, dimers with plaquette interaction are believed not to be
solvable, i.e., not even the basic thermodynamic functions admit an
explicit representation.  From the exact solution, one finds that some
of the critical exponents of the 6-vertex model depend continuously on
the vertex weights\footnote{More precisely, the limit of the critical
  exponents of the 8-vertex model as the additional vertex weights
  $a_7=a_8$ tend to zero have a non-trivial continuous dependence on
  the remaining vertex weights $a_1,\ldots,a_6$, see
  \cite[Eqs.10.12.23 and 10.12.27]{Ba}.}; they differ, in general,
from those of the standard dimer model.  On the other hand, the
existence of non-trivial critical exponents in the dimer model with
plaquette interaction, as well as in other planar models in the same
`universality class' (such as coupled Ising models, Ashkin-Teller and 8-vertex models) can be
proved by constructive Renormalization Group (RG) methods 
\cite{BFM1,BFMprl,GM,M}, which allow one to express
them as convergent power series in the interaction strength.

In this setting, it is natural to ask whether the height fluctuations
are still described by a GFF at large scales and, in case, whether the
pre-factor in front of the logarithm still displays some universal
features. The very fact that the critical exponents depend
non-trivially on the interaction strength suggest that universality
cannot then be true in the naive, strong, sense that `large-scale
properties are independent of the microscopic details of the model':
in fact, in this context, a weaker form of universality is expected,
in the form of a number of {\it scaling relations}, originally
proposed by Kadanoff \cite{K}, which allow one to determine all the
critical exponents of the critical theory in terms of just one of
them; this form of universality is often referred to as `weak
universality', see e.g. \cite[Section 10.12]{Ba}.  Support for the Kadanoff scaling relations comes from the so-called
bosonization picture, see e.g. \cite{GMT17b} for a basic
introduction. Only some  of these universality relation have been rigorously proven
\cite{BFM1,BFMprl}; an example is the identity $X_c X_e=1$
\cite[Eq.(13b)]{K}, relating the ``crossover exponent'' $X_c$
and ``energy exponent'' $X_e$, see 
\cite[Eq.(1.10)]{BFM1}. The proof in \cite{BFM1} covers both solvable and non-solvable models, but
only works for scaling relations involving the critical exponents
of the ``local observables'', i.e., those that admit a representation in
terms of a local fermionic operator. Other scaling relations, involving the
critical exponents of non-local observables (e.g. monomer-monomer correlations in
dimer models, or spin-spin correlations in the Ashkin-Teller model)
remained elusive for many years.  In particular, the relation
$X_p=X_e/4$ \cite[Eq.(13a)]{K}, relating the energy exponent $X_e$ to the
``polarization exponent'' $X_p$ in the AT model, remain{s} unproven at a
rigorous level.
% \note{"remained unproven at a rigorous level" insinuates that it
%is proven in the present work, which if I understand the following discussion is not quite the case. If so change "remained" to "remains", otherwise
%modify the following discussion.}

\medskip

In this paper, we prove the stability of the Gaussian nature of the
height fluctuations for non-integrable perturbations of the dimer
model, with logarithmic growth of the variance in the whole liquid
region. The pre-factor $A$ in front of the logarithm
 depends, in general, non-trivially on the strength of the perturbation (see Remark \ref{rem:frattaglie} below)  and on the
 dimer weights, so it is not universal in a naive, strong, sense.
The non-trivial dependence of $A$  on the interface tilt has been also verified numerically for the 6-vertex model
\cite{Inhomo}.
Nevertheless, $A$ satisfies a scaling relation, that connects it with
the critical exponent of the dimer-dimer correlations.

\begin{MT} In a weakly perturbed dimer model with perturbation of strength $\l$, the
variance of the height difference between two faraway points grows
like the logarithm of the distance, with a pre-factor $A/\pi^2$,
where $A=1+O(\l)$ is an analytic function of $\l$ and of the dimer weights. Moreover, the {prefactor} satisfies the 
scaling relation 
\be A=\nu\label{1},\ee
where $2\nu$ is the anomalous decay exponent of the dimer-dimer correlation. Higher  cumulants of the height difference between two points are bounded uniformly in their distance, 
that is, the fluctuations of the height difference are asymptotically Gaussian. 
\end{MT}

For a more precise statement, see Theorem \ref{th:2} and the remarks and comments 
that follow it.  Note that in the {unperturbed} case, $\l=0$, the
dimer-dimer correlation decays at large distances like $(dist.)^{-2}$
in the whole liquid phase, i.e., its decay exponent is equal to $2$
(so that $\nu=1$), irrespective of the specific choice of the dimer
weights. In this case, of course, our result reduces to the one of
\cite{KOS}, $A=1$. Note also that our result covers both integrable models, such as 6-vertex, and non-integrable ones, 
in the spirit of the universality picture.

Scaling relations involving exponents and amplitudes were conjectured
by Haldane \cite{Ha} and proved by Benfatto and Mastropietro
\cite{BMdrude,BMun} in the context of quantum one-dimensional models.
Even if formulated in different notations, the scaling relation
\eqref{1} is strictly related to one of those proposed by Kadanoff, in
particular to the above-mentioned, elusive, identity $X_p=X_e/4$
\cite[Eq.(13a)]{K}.  In fact, there is a duality (called `discrete
bosonization' in \cite{Du}) between the 6-vertex model, which is part
of the class of perturbed dimer models considered in this paper, and
the AT model; the duality implies non-trivial identities between the
correlations of 6-vertex model and those of AT, see \cite[Section
2.6]{Du}. In particular, the two-point correlation of the polarization
operator in AT equals the `electric correlator'
$\media{e^{i\pi(h_x-h_y)}}_{6V}$ of the 6-vertex model, see
\cite[Section 2.6]{Du}\footnote{Here $h_x$ is the height function of
  the 6-vertex model at face $x$ and $\media{\cdot}_{6V}$ is the
  corresponding statistical average; the factor $\pi$ at the exponent
  depends on our definition of height function, which differs by a
  multiplicative factor $2\pi$ from that of \cite{Du}.}, while the
energy critical exponent of AT equals the anomalous decay exponent of
the arrow-arrow correlations of 6-vertex\footnote{In the dimer
  formulation of 6-vertex, the arrow-arrow correlations translate into
  the dimer-dimer correlations.}. Given these identities, \eqref{1}
implies that $X_p=X_E/4$ \cite[Eq.(13a)]{K}, provided that
\be\media{e^{i\pi(h_x-h_y)}}_{6V}\sim
e^{-\frac{\pi^2}{2}\media{(h_x-h_y)^2}_{6V}}\sim
e^{-\frac{A}2\log|x-y|}\label{6Vel}\ee at large distances, as
suggested by the asymptotic Gaussian behavior of the height
difference\footnote{As discussed in \cite[Remark 2]{GMT17a}, our
  method allows us to compute the average of $\exp\{i\pi(h_x-h_y)\}$
  only after coarse-graining the height difference {in the} exponent against
  a smooth test function.}.  \medskip

To prove our results, we start by periodizing the non-integrable dimer
model on the toroidal graph of size $L$. Then we map it into a system
of interacting two-dimensional lattice fermions, by rewriting its moment generating function as an integral over Grassmann variables, with non-quadratic action. At this point, we apply tools
from the so-called constructive fermionic RG to control the $L\to\infty$ limit of the correlation functions. In
particular, we need a very sharp asymptotic description of the
large-distance behavior of the dimer-dimer correlation function
(cf. Theorem \ref{th:1}). The large-scale logarithmic behavior of
height correlations, as well as the validity of the `Haldane' scaling relation \eqref{1},
rely on non-trivial identities (cf. \eqref{eq:32xl}) between the
coefficients appearing in the large-distance asymptotics of the
dimer-dimer correlation function. In turn, \eqref{eq:32xl} is the
result of so-called Ward identities, i.e. exact relations between the correlation
functions of the interacting lattice fermionic model, which the dimer model maps into. 

The analogs of Theorems \ref{th:1} and \ref{th:2} have been proven in
our previous works \cite{GMT17a,GMT17b} for the specific case of
plaquette interaction and uniform edge weights $\underline t\equiv
1$. In this case, the average tilt of the height field is just
$\rho=0$ and the model has all the discrete symmetries of the lattice
$\mathbb Z^2$. The extension to the general case, achieved
here, is non-trivial: the loss of discrete rotation and reflection
symmetries results, in the RG language, in the emergence of four new
running coupling constants (two ``Fermi velocities'' and two ``Fermi
points''), whose flow, along the multi-scale integration procedure,
has to be controlled via the choice of suitable
counter-terms. Another consequence of the loss of rotation and reflection symmetry is 
that the cancellation at the basis of the logarithmic growth of the variance 
does not follow simply from the basic symmetries of the model, as it
was the case in \cite{GMT17a,GMT17b}: the proof of the key identity,
\eqref{eq:32xl}, now requires the use of a lattice Ward Identity for
the dimer model, in combination with an emergent Ward Identity for an
effective continuum model, which plays the role of `infrared fixed
point' of the RG flow.  Quite surprisingly, the loss of rotation and
reflection symmetry plays a role also in the technical control of the
thermodynamic limit of correlations: in \cite{GMT17a,GMT17b}, in order
to simplify the analysis of the finite-size corrections to the
critical correlation functions, we first studied a modified, slightly
massive model of mass $m>0$ (the modification consisted in adding a
modulation of size $m$ on the horizontal dimer weights; in the
tilt-less case, this was enough to guarantee that the modified
correlations decayed exponentially with rate $m$), and then we took
the massless limit $m\to0$ after the thermodynamic limit. However,
this strategy fails for general dimer weights: in this case, neither a
modulation of the dimer weights nor other simple modifications of the
model produce a mass; therefore, in the present paper, we directly
derive quantitative estimates on the corrections to the thermodynamic
limit of the massless correlations, by a careful control of the
finite-size effects in the multi-scale procedure.

\subsection{Related works}
Let us conclude this introduction by mentioning some recent related
works. While most literature on dimer models focuses on the
determinantal case, there have been recently various attempts to go
beyond the exactly solvable situation \cite{Okou}.  As far as ``limit
shape phenomena'' (i.e. laws of large numbers for the height profile)
for non-solvable random interface models are concerned, let us mention
for instance \cite{5vertex,MTassy,CoSpo}. Closer in spirit to our
results is \cite{CoSpe}, which provides a central limit theorem for
height fluctuations of $\nabla\phi$-interface models with continuous
heights and strictly convex potential. This work uses the
Helffer-Sj\"ostrand formula, that is not available for discrete-height
model like the dimer model.  Let us mention also \cite{Ahn}, which
obtains convergence to the GFF (with interaction-dependent amplitude)
for a dimer model with a special non-local interaction that makes it integrable,
although not determinantal. Finally, a very interesting recent
development is \cite{BGL}: while in this work the convergence to the
GFF is proven only for the non-interacting dimer model, the method of
proof, that goes through Temperley's bijection and Wilson's algorithm
rather than via Kasteleyn's theory, might prove robust enough to allow
for extensions to some non-determinantal situations.

\subsection{Organization of the article}
The rest of this work is organized as follows.  The dimer model is
defined in Section \ref{sec:mamr}. There, we recall the large-scale
behavior of the integrable model and we state our results for the
non-integrable one.  In Section \ref{sec:grass} we give the Grassmann
representation of the interacting dimer model and its lattice Ward
identities. In Section \ref{sec:IR} we recall the continuum reference
model that plays the role of infrared fixed point of interacting
dimers. Theorems \ref{th:1}-\ref{th:2} are proven in Section
\ref{sec:proveth}, conditionally on technical results, based on the
multi-scale expansion, whose proofs are postponed to Section \ref{sec:RG}.

\section{Model and main results}
\label{sec:mamr}

\subsection{Dimers and height function}

A dimer covering, or perfect matching, of a graph $\Gamma$ is a subset
of edges that covers every vertex exactly once.  The set of dimer
coverings of $\Gamma$ is denoted $\Omega_\Gamma$. We color the
vertices of the bipartite graph $\mathbb Z^2$ black and white so that
neighboring vertices have different colors.  A white vertex is
assigned the same coordinates $x=(x_1,x_2)$ as the black vertex just at
its left.  The choice of coordinates is such that the vector
$\vec e_1$ is the one of length $\sqrt 2$ and angle $-\pi/4$ w.r.t the horizontal axis, while
$\vec e_2$ is the one of length $\sqrt 2$ and angle $+\pi/4$. The
finite graph $\mathbb T_L$ denotes $\mathbb Z^2$ periodized (with
period $L$) in both directions $\vec e_1,\vec e_2$. See Fig. \ref{fig:1}.

\begin{figure}
	\begin{center}
		\includegraphics[height=7cm]{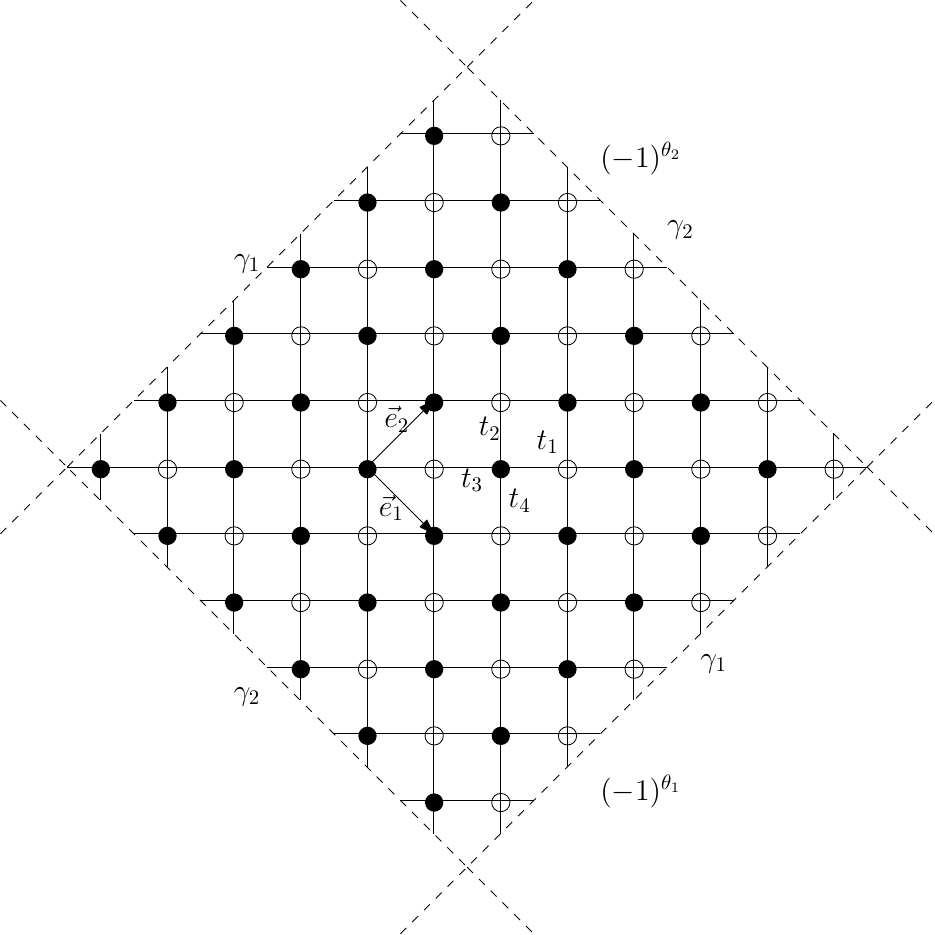}
		\caption{The graph $\mathbb T_L$ for $L=4$. The coordinate axes $\vec e_1,\vec e_2$, as well as the corresponding coordinates of some black/white vertices, are explicitly indicated. In the right drawing, the weights $t_1,\dots,t_4$ and the
                corresponding edges of types $1,\dots,4$.}
		\label{fig:1}
	\end{center}
\end{figure}

For simplicity we assume that $L$ is even. Black/white sites
are therefore indexed by coordinates
$x\in\Lambda=\{(x_1,x_2), 1\le x_i\le L\}$. An edge $e=(b,w)$ of
$\mathbb T_L$ is said to be of type $r\in\{1,2,3,4\}$ if its white
endpoint $w$ is to the right, above, to the left or below the black
endpoint $b$.  If $e=(b,w)$ is an edge of type $r$ and $x(b)$ is the
coordinate of $b$ then $x(w)=x+v_r$, with
\begin{equation}
v_1=(0,0)\quad v_2=(-1,0)\quad v_3=(-1,-1)\quad v_4=(0,-1).\label{vii}
\end{equation}
If $\Gamma$ is planar and bipartite, the height function allows us to
interpret a dimer covering as a two-dimensional discrete surface. Let us recall
the standard definition of height function for the infinite lattice
$\mathbb Z^2$.  Given $M\in\Omega_{\mathbb Z^2}$, the height function
$h(\cdot):=h_M(\cdot) $ is defined on the dual lattice
$(\mathbb Z^*)^2$, i.e. on the faces $\eta$ of $\mathbb Z^2$. We set
$h(\eta_0):=0$ at a given reference face $\eta_0$, and we let its
gradients be given by
\begin{gather}
  \label{eq:30}
  h(\eta')-h(\eta)=\sum_{e\in C_{\eta\to \eta'}}\sigma_e (\mathds 1_{e}-1/4)
\end{gather}
where $\eta,\eta'$ are any two faces, $\mathds 1_{e}$ denotes the
dimer occupancy, i.e., the indicator function that $e$ is occupied by
a dimer in $M$, while $C_{\eta\to \eta'}$ is any nearest-neighbor path
on the dual lattice $(\mathbb Z^*)^2$ from $\eta$ to $\eta'$ (the
right side of \eqref{eq:30} is independent of the choice of
$C_{\eta\to \eta'}$).  The sum runs over the edges crossed by the path
and $\sigma_e=+1/-1$ depending on whether the oriented path
$C_{\eta\to \eta'}$ crosses $e$ with the white site on the right/left.

\subsection{Definition of the   model}
We define here both the non-interacting dimer model \cite{Ka} and the interacting one. Both are probability measures on $\Omega_L:=\Omega_{\mathbb T_L}$,
denoted $\mathbb P_{L,\underline t}$ and $\mathbb P_{L,\lambda, \underline t}$ respectively, where $\lambda\in \mathbb R$ is the interaction strength and $\underline t$ are the edge weights. For lightness of notation, the index $\underline t$ will be dropped. 

\subsubsection{The non-interacting dimer model}

We assign a positive weight to each edge. More precisely, an
edge of type $r\in\{1,2,3,4\}$ is given a weight $t_r> 0$. Then, the
weight of a configuration $M\in\Omega_L$ is
\begin{eqnarray}
\label{eq:Z}
\mathbb P_L(M) =\frac{ t_1^{N_1(M)}t_2^{N_2(M)}t_3^{N_3(M)}t_4^{N_4(M)}}{Z^0_L},\\ Z^0_L=\sum_{M'\in\Omega_L} t_1^{N_1(M')}t_2^{N_2(M')}t_3^{N_3(M')}t_4^{N_4(M')}
\end{eqnarray}
with $N_i(M)$ the number of dimers on edges of type $i$ in
configuration $M$.  Since the total number of dimers is constant, we
can rescale all weights by a common factor and we will set
$t_4\equiv 1$ from now on. 
It is known that the free energy per site
has a limit  as $L\to\infty$ (the
infinite volume free energy):
\begin{eqnarray}
  \label{eq:Ft}
  F(\underline t)=\lim_{L\to\infty}\frac1{L^2} \log Z^0_L=\frac1{(2\pi)^2}\int_{[-\pi,\pi]^2}d k \log \mu(k),\\
  \mu(k)=t_1+ it_2 e^{i k_1}-t_3e^{i k_1+i k_2}-i e^{i k_2}.
\end{eqnarray}
Note that
\begin{eqnarray}
  \label{eq:simmetria1}
\m(k)=\m^*((\pi,\pi)-k).
\end{eqnarray}
The ``characteristic polynomial'' mentioned in the introduction is  $P(z,w):=\mu(-i \log z, -i \log
  w)$.

Also, the measure $\mathbb P_L$ itself has a limit $\mathbb P$ as
$L\to\infty$, in the sense that the probability of any local event
converges.  The non-interacting model is integrable, and both the
measure $\mathbb P_L$ and its limit $\mathbb P$ admit a determinantal
representation, recalled in Section \ref{sec:Ktheory}. 

In the special case where $t_1=t_3=: t$ and $t_2=1$,
i.e. assigning weight $t$ to horizontal edges and $1$ to vertical
ones, one recovers the model originally solved by Kasteleyn
\cite{Ka}. For general weights $t_1,t_2,t_3$, the
model is equivalent to Kasteleyn's model with different weights for horizontal and vertical edges, and a non-zero average slope
$\rho=\rho(t_1,t_2,t_3)\in \mathbb R^2$ for the height function, i.e.,
\begin{eqnarray}
  \label{eq:sloppa}
  \mathbb E(h(\eta+\vec e_i)-h(\eta))=\rho_i,\qquad  i=1,2,
\end{eqnarray}
where $\mathbb E$ denotes the average with respect to $\mathbb P$. 
In fact, the weights $t_i$ are  chemical potentials by which one can fix the densities of the four types of edges. Then, the slope $\rho$ is obtained as a function of the four densities using the definition \eqref{eq:30} of height function.

Another special case is obtained letting e.g. $t_3\to0$: then,  the model reduces to the closed-packed dimer
model on the hexagonal graph with weights $1,t_1,t_2$ for the three
types of edges.

Note that the condition $\mu(k)=0$ gives
\begin{eqnarray}
  e^{i k_2}=\frac{t_1+i t_2 e^{i k_1}}{i+ t_3 e^{i k_1}}
\end{eqnarray}
that determines the intersections of two circles in the complex plane. We will make the following important assumption:
\begin{Assumption}
\label{ass:liquid}
The parameters $\underline t$ are such that $\mu(\cdot)$ has two
distinct simple zeros, that we call $p^+$ and $p^-$, on $[-\pi,\pi]^2$
(i.e. the two circles intersect transversally). In view of
\eqref{eq:simmetria1}, one has $p^++p^-=(\pi,\pi)$.
\end{Assumption}
\begin{Remark}
  \label{rem:liquid}
  Note that, under Assumption \ref{ass:liquid}, none of the weights
$t_1,t_2,t_3,1$ exceeds the sum of the other three, otherwise $\mu(k)$ would vanish nowhere on $[-\pi,\pi]^2$. Note also that $p^\o,\o=\pm$ cannot coincide with any of the four values $k=(\epsilon_1\pi/2,\epsilon_2\pi/2), \epsilon_1=\pm1, \epsilon_2=\pm1$, otherwise one would have $p^+=p^-$ (modulo $(2\pi,2\pi)$).
\end{Remark}
Under  Assumption \ref{ass:liquid}, it is known \cite{KOS} that the infinite-volume
measure has power-law decaying correlations (in the language of
\cite{KOS}, the dimer model is said to be in a ``liquid phase'').
With the nomenclature of condensed matter theory, the zeros $p^\pm$ are called ``Fermi points''. 

\subsection{The interacting dimer model, and relation to the 6-vertex  model}

In order to study the effect of the breaking of integrability we
introduce interacting dimer measures of the following form:
\begin{eqnarray}
  \label{eq:BGibbs}
  \mathbb P_{L,\l}(M)= \frac{p_{L,\l}(M)}{Z_{L}}
\end{eqnarray}
where
\begin{equation}\label{tWL}
  \begin{aligned}
&  p_{L,\l}(M)=t_1^{N_1(M)}t_2^{N_2(M)}t_3^{N_3(M)} 
\, e^{\lambda W_L(M)},\\
&  Z_{L}=\sum_{M\in\Omega_L}p_{L,\l}(M)
  \end{aligned}\end{equation}
and the interaction potential $W_L$ is given as
\begin{eqnarray}
  \label{eq:WL}
  W_L(M)=\sum_{x\in\Lambda}f(\tau_x M),
\end{eqnarray}
where $f$ is some fixed local function of the dimer configuration and
$\tau_x M$ denotes the configuration $M$ translated by $x_1 \vec e_1+x_2\vec e_2$.  We \emph{do not}
require $f(\cdot)$ to be symmetric under reflections or rotation by
$\pi/4$.

Let us mention two interesting particular examples of
  interaction $W_L(M)$. The first one is the plaquette interaction
  that was considered in our works \cite{GMT17a,GMT17b} and previously in the
  theoretical physics literature \cite{A} in the context of quantum dimer models. Namely,
  \begin{eqnarray}
    \label{eq:placchetta}
    W_L(M)=\sum_{\eta\in\mathbb T^*_L}{\bf 1}_\eta(M)
  \end{eqnarray}
  where the sum runs over all faces of $\mathbb T_L$ and
  ${\bf 1}_\eta(M)$ is the indicator function that two of the four
  edges surrounding $\eta$ are occupied by dimers. In this case the
  function $f$ in \eqref{eq:WL} is
  \begin{eqnarray}
    \label{eq:plaq}
f_P(M)=\mathds 1_{e_1}\mathds 1_{e_2}+\mathds 1_{e_3}\mathds 1_{e_4}+\mathds 1_{e_1}\mathds 1_{e_5}+\mathds 1_{e_6}\mathds 1_{e_7}
  \end{eqnarray}
with $e_1,\dots,e_7$ as in Fig. \ref{fig:esempi}.
\begin{figure}
	\begin{center}
		\includegraphics[height=3cm]{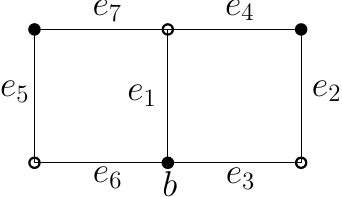}
		\caption{The edges appearing in \eqref{eq:plaq}.}
		\label{fig:esempi}
	\end{center}
\end{figure}

Another important example is
\begin{eqnarray}
  \label{eq:6v}
  f_{6v}(M):=\mathds 1_{e_1}\mathds 1_{e_2}+\mathds 1_{e_3}\mathds 1_{e_4}.
\end{eqnarray}
In this case, the interaction $W_L(M)$ in \eqref{eq:placchetta} is
modified in that the sum runs only over one of the two sub-lattices of
$\mathbb T_L^*$ (the subset of faces with black top-right
vertex). Then, it is known that this interacting dimer model is
equivalent to the 6-vertex model \cite{Ba1, EKLP, F}. Recall that configurations
of the 6-vertex model are assignments of orientations (arrows) to the
edges of $\mathbb Z^2$ such that at each vertex there are two incoming
and two outgoing arrows. There are 6 possible arrow configurations at
any vertex, each being assigned a positive weight $a_1,\dots,a_6$ (see
Fig. \ref{fig:6v1}) and the weight of a configuration is the product
of the weights over all vertices.
\begin{figure}
	\begin{center}
		\includegraphics[height=1.7cm]{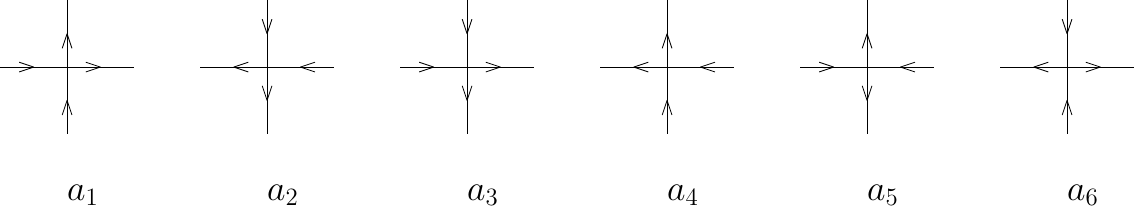}
		\caption{The six possible vertex configurations of the 6-vertex model and the associated weights.}
		\label{fig:6v1}
	\end{center}\end{figure}
      By multiplying all weights by a common factor, one can reduce
      e.g.  to $a_3=1$. Moreover, on the torus, the
      number of vertices of type $5$ equals the number of vertices of
      type $6$, so one can set without loss of generality $a_5=1$. One
      is left with four positive weights $a_1,a_2,a_4,a_6$ and the
      model can be mapped to the interacting dimer model with weights $t_1,t_2,t_3$, 
      interaction \eqref{eq:6v} and interaction parameter $\l$ such
      that
      \begin{eqnarray}
        \label{eq:6vpesi}
t_1=a_1,\;t_2=a_4,\;t_3=a_2, \;(t_1 t_3+t_2)e^\l=a_6.        
      \end{eqnarray}
More precisely, as in Fig. \ref{fig:6v3}, the dimer model lives on a
square grid rotated by $45$ degrees w.r.t. the lattice of the 6-vertex model.
      \begin{figure}
	\begin{center}
		\includegraphics[height=6cm]{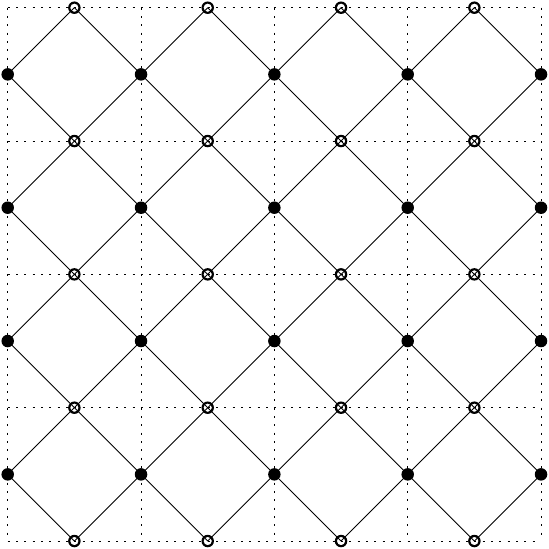}
		\caption{The 6-vertex model lives on the square grid
                  $\mathcal G_{6v}$ with dotted edges, while the dimer
                  model lives on the square grid $\mathcal G_{d}$ with full
                  edges. Faces of $\mathcal G_{d}$ containing a vertex
                  of $\mathcal G_{6v}$ are called ``even faces''
                  and the others ``odd faces''. }
		\label{fig:6v3}
	\end{center}\end{figure}
      The mapping is obtained by associating to the arrow
      configuration at a vertex $x$ of $\mathcal G_{6v}$ a dimer
      configuration at the even face of $\mathcal G_{d}$ containing
      $x$, as in Fig. \ref{fig:6v2}. 
      \begin{figure}
	\begin{center}
		\includegraphics[height=3.7cm]{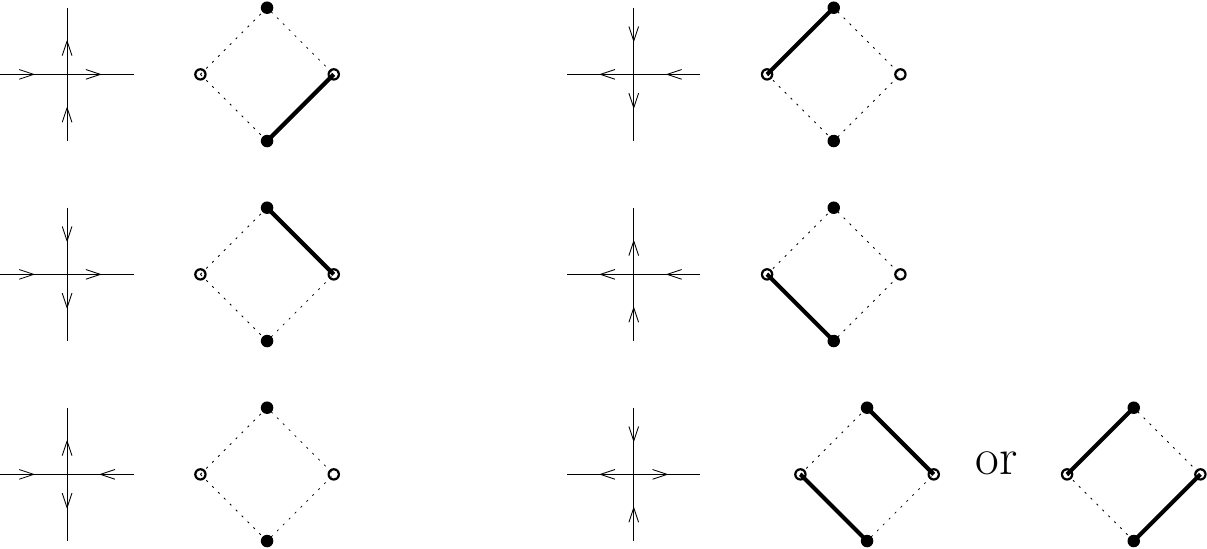}
		\caption{The local arrow-to-dimer mapping.}
		\label{fig:6v2}
	\end{center}\end{figure}
      The map is one-to-many because arrow configurations of type $6$
      are mapped to two possible dimer configurations. However, it is
      easily checked that the partition functions of the two models
      are equal provided the parameters are identified as in
      \eqref{eq:6vpesi}. Moreover, the height function of the dimer
      model, restricted to odd faces of $\mathcal G_d$, equals (up to
      a global prefactor) the canonical height function of the
      6-vertex model \cite{vB}.  The 6-vertex model is known to be
      free-fermionic (i.e. determinantal) if and only if
      \begin{eqnarray}
        \label{eq:Delta6v}
        \Delta:=\frac{a_1 a_2+a_3 a_4-a_5a_6}{2\sqrt{a_1 a_2 a_3 a_4}}=0.
      \end{eqnarray}
      It is immediately checked that this condition is equivalent to $\lambda=0$ for the interacting dimer model.

\subsection{Non-interacting model: dimer-dimer correlations and logarithmic height fluctuations}
\label{sec:recall}
It is known \cite{KOS} that, under the infinite-volume measure
$\mathbb P$ of the non-interacting model, dimer-dimer correlations
decay like the inverse distance squared and the height field behaves
on large scales like a massless Gaussian field. We briefly recall the basic facts here, since they serve to motivate our main result for the interacting dimer model.
For $\o=\pm$, we let
\begin{align}
   \label{eq:alphabeta}
   \alpha_\o&=\partial_{k_1}\mu(p^\o)=-t_2e^{i p^\o_1}-i t_3 e^{i(p^\o_1+p^\o_2)}=-i t_1- e^{i
   p^\o_2}, \\
   \beta_\o&=\partial_{k_2}\mu(p^\o)=-i t_3 e^{i (p^\o_1+p^\o_2)}+e^{i p^\o_2}=-i t_1+t_2e^{i p^\o_1},
 \end{align}
 where $p^\pm$ are the two zeros of $\mu(\cdot)$, as in Assumption \ref{ass:liquid}. (The complex numbers $\alpha_\o,\beta_\o$ are called ``Fermi velocities'' in the jargon of condensed matter.)
Define also 
\begin{eqnarray}
  \label{eq:Phi}
  \phi_\o:x\in \mathbb R^2\mapsto \phi_\o(x):=\o(\beta_\o x_1-\alpha_\o x_2)\in \mathbb C.
\end{eqnarray}

\begin{Remark}
  \label{rem:platonico}
  Under Assumption \ref{ass:liquid} on the weights $\underline t$, the complex numbers $\a_\o$ and $\b_\o$ are not colinear,
as elements of the complex plane \cite{KOS}, i.e. $\alpha_\o/\beta_\o$ is not real. Therefore, $\phi_\o$ is a bijection from $\mathbb R^2$ to the complex plane.
More precisely, one has that
\begin{eqnarray}
  \label{eq:sign}
  {\rm Im}(\b_+/\a_+)>0.
\end{eqnarray}
In fact, parametrize the weights $t_1,t_2,t_3$ as
\[
  t_1=t e^{-B_1},\quad t_2=e^{-B_1-B_2},\quad t_3=t e^{-B_2},\quad B_1,B_2\in\mathbb R.
\]
For $B_1,B_2=0$ it is immediately checked that
$p^+=(0,0),p^-=(\pi,\pi)$ and that \eqref{eq:sign} holds. On the other
hand, once $t$ is fixed, it is known \cite{KOS} that the set
$\mathcal B_t$ of values of $B=(B_1,B_2)$ for which Assumption
\ref{ass:liquid} holds is a connected subset of $\mathbb R^2$ on which
${\rm Im}(\b_+/\a_+)$ vanishes nowhere, and it is therefore everywhere
positive.
\end{Remark}
Because of the symmetry \eqref{eq:simmetria1}, one has  $\a_\o=-\a^*_{-\o}$, $\b_{\o}=-\b^{*}_{-\o}$ and $\phi_\o^*(\cdot)=\phi_{-\o}(\cdot)$. 

The relation between the massless Gaussian field and the height
function is given by the following results. Let $n$ be an integer and
$\eta_j,j\le 2n$ be faces of $\mathbb Z^2$.  With some abuse of
notation, we identify a face $\eta$ with its mid-point. Then,
\begin{eqnarray}
\label{eq:35}
  &&\mathbb E\left[(h(\eta_1)-h(\eta_2));(h(\eta_3)-h(\eta_4))
     \right]\\
  &&\quad =\frac1{2\pi^2}
     \Re\log \left(\frac{(\phi_+(\eta_4)-\phi_+(\eta_1))(\phi_+(\eta_3) - \phi_+(\eta_2))}{(\phi_+(\eta_4)-\phi_+(\eta_2))(\phi_+(\eta_3)-\phi_+(\eta_1))}
     \right)\nonumber\\
  &&\quad +O\left(\frac1{\min_{i\ne j\le 4}|\eta_i-\eta_j|+1}\right)\nonumber
\end{eqnarray}
where $\phi_+(\eta_i)-\phi_+(\eta_j)$ should be read as $1$ in
case $\eta_i=\eta_j$. Also, for $n>2$
\begin{multline}
  \label{eq:cum}
  \mathbb E\left[(h(\eta_1)-h(\eta_2));\dots; (h(\eta_{2n-1})-h(\eta_{2n}))
  \right]\\=O\left(\frac1{\min_{i\ne j\le 2n}|\eta_i-\eta_j|+1}\right)
\end{multline}
where $\mathbb E(X_1;\dots;X_k)$ denotes the joint cumulant of the random variables $X_1,\dots,X_k$.
In particular, as ${|\eta_1-\eta_2|\to\infty}$,
\begin{eqnarray}
  \label{eq:varni}
   {\rm Var}_{\mathbb P}(h(\eta_1)-h(\eta_2))=\frac{1}{\pi^2}\Re\log (\phi_+(\eta_1)-\phi_+(\eta_2))+O(1)
\end{eqnarray}
while the cumulants of order $n\ge3$ of $(h(\eta_1)-h(\eta_2))$ are bounded from above, uniformly in $\eta_1,\eta_2$. It is well known that
\eqref{eq:35} and \eqref{eq:cum} imply that the height field tends, in
the scaling limit, to a GFF with covariance
\be -\frac1{2\pi^2}\Re\log(\phi_+({ x})-\phi_+({y})),\ee
{with $x,y\in\mathbb R^2$.}
For \eqref{eq:35} see
\cite{KOS} and for \eqref{eq:cum} see e.g. \cite[Th. 5]{GMT17a} (in
\cite{GMT17a} the weights $t_i$ are all $1$ and
$\eta_1=\eta_3=\dots \eta_{2n-1}, \eta_2=\eta_4=\dots=\eta_{2n}$; the
proof of \eqref{eq:cum} in the general case works the same
way).
\begin{Remark}
  Note that the prefactor $1 /\pi^2$ is independent of the
  weights $\underline t$.  In \cite{KOS}, such universality is related
  to the fact that the spectral curve, i.e. the algebraic curve
  defined by the zeros on $\mathbb C^2$ of the polynomial
  $P(z,w):=\mu(-i \log z,-i\log w)$, is a so-called Harnack curve.
\end{Remark}

It is useful to recall the key points of the proof of
\eqref{eq:35}  in order to understand the main new
features posed by the presence of the interaction.
From the definition of height function,
\begin{multline}
  \label{eq:sumdd}
\mathbb E\left[(h(\eta_1)-h(\eta_2));(h(\eta_3)-h(\eta_4))
\right]\\= \sum_{e\in C_{\eta_1\to \eta_2}} \sum_{e'\in C_{\eta_3\to \eta_4}}\sigma_e
  \sigma_{e'} \mathbb E(\mathds 1_e;\mathds 1_{e'})
\end{multline}
where $\mathbb E(\mathds 1_e;\mathds 1_{e'})$ is the dimer-dimer correlation
function
\[ \mathbb E(\mathds 1_e;\mathds 1_{e'}):= \mathbb E(\mathds 1_e\mathds 1_{e'})- \mathbb
  E(\mathds 1_e)\mathbb E(\mathds 1_{e'}).
\]
This correlation function has an exact expression involving the
inverse Kasteleyn matrix of the infinite lattice; at large distances,
it can be expressed as
  \begin{eqnarray}
    \label{eq:dd}
&&    \mathbb E(\mathds 1_e;\mathds 1_{e'})= A_{r,r'}  (x,x')+B_{r,r'}(x,x')+R_{r,r'}(x,x'),\\
&&A_{r,r'}(x,x')=\frac1{4\pi^2}\sum_{\o=\pm}
\frac{K_{\o, r} K_{\o, r'}}{(\phi_\o(x-x'))^2}\nonumber\\
&&B_{r,r'}(x,x')=\frac1{4\pi^2}
\sum_{\o=\pm} \frac{K_{-\o, r} K_{\o,r'}
}{|\phi_\o(x-x')|^2} e^{i (p^\o-p^{-\o})\cdot(x-x')}
   \nonumber
  \end{eqnarray}
  where:
\begin{itemize}
\item the edge $e$ (resp. $e'$) is of type $r=r(e)$ (resp. $r'=r(e')$) and the coordinate of its black endpoint is $x=x(e)$ (resp. $x'=x(e')$);
\item   $K_{\o,r}= K_r e^{-i p^\o\cdot v_r}$ (see \eqref{vii} for the definition of $v_i$) with
\begin{eqnarray}
\label{eq:Kr} K_1=t_1,\quad K_2=i t_2, \quad K_3=-t_3, \quad K_4=-i;
\end{eqnarray}
note that $K_{-\o,r}=K^*_{\o,r}$.
\item  $R_{r,r'}(x,x')$ is a remainder, decaying like $|x-x'|^{-3}$ at large distance.
\end{itemize}
{By plugging \eqref{eq:dd} into \eqref{eq:sumdd}, one obtains, letting $d:=\min_{i\neq j}|\eta_i-\eta_j|$, 
\begin{eqnarray}\label{5:13}
\eqref{eq:sumdd}&=&%\mathbb E\left[(h(\eta_1)-h(\eta_2));(h(\eta_3)-h(\eta_4))\right]= 
\frac1{4\pi^2} \sum_{\o=\pm} \sum_{\substack{e\in C_{\eta_1\to \eta_2}\\ e'\in C_{\eta_3\to \eta_4}}}\sigma_e
  \sigma_{e'} \Big[\frac{K_{\o, r} K_{\o, r'}}{(\phi_\o(x-x'))^2}\\
  &&+\frac{K_{-\o, r} K_{\o,r'}
}{|\phi_\o(x-x')|^2} e^{i (p^\o-p^{-\o})\cdot(x-x')}\Big]+O\Big(\frac{1}{d+1}\Big),\nonumber
\end{eqnarray}
where the $O(\frac1{d+1})$ in the second line comes from the remainder $R_{r,r'}(x,x')$ (in order to prove this error bound, one 
needs to choose the paths $C_{\eta_1\to \eta_2}$, $C_{\eta_3\to \eta_4}$ to be as `well separated' as possible, cf. with \cite[Fig.3]{GMT17a}
and the discussion in  \cite[Sect.3.2]{GMT17a}). }

{In order to conveniently rewrite the double sum over $e,e'$ in \eqref{5:13},
we assume for simplicity that the paths
$C_{\eta_1\to \eta_2},C_{\eta_3\to \eta_4}$ are a concatenation of
elementary steps in direction $\pm \vec e_1$ and $\pm \vec e_2$,
connecting faces of the same parity: e.g., assume that an elementary
step $s(x,1)$ in direction $+\vec e_1$ `centered at $x$' consists in
crossing the two bonds $(\bullet,\circ)=(x,x+v_3)$ and
$(\bullet,\circ)=(x,x+v_4)$ with the white vertex on the right, while
an elementary step $s(x,2)$ in direction $+\vec e_2$ centered at $x$
consists in crossing the two bonds $(\bullet,\circ)=(x,x)$ and
$(\bullet,\circ)=(x-v_4,x)$ with the white vertex on the right. A
simple but crucial observation is that
\begin{eqnarray}\label{eq:32x}
&&\hskip-1.truecm\sum_{e\in s(x,1)} \s_e K_{\o,r(e)}
= K_{3}e^{-i p^\o v_3}  +K_{4}e^{-i p^\o v_4}
=-i \b_\o=
-i \o\D_1\phi_\o\\
&&\hskip-1.truecm\sum_{e\in s(x,2)} \s_e K_{\o,r(e)}
= K_1e^{-ip^\o v_1}+K_4e^{-ip^\o v_4}
=i\a_\o=
-i \o\D_2\phi_\o\label{eq:32x.1}
\end{eqnarray}
where $\Delta_j\phi_\o$ denotes the discrete gradient in direction
$\vec e_j$ of the linear function $\phi_\o$ defined in \eqref{eq:Phi}. Therefore, by using these identities, we can rewrite the right side of \eqref{5:13}, up to the error $O(\frac1{d+1})$, as
\begin{eqnarray}\label{5:14}
&&-\frac1{4\pi^2} \sum_{\o=\pm} \sum_{\substack{s(x,j)\in C_{\eta_1\to \eta_2}\\ s(x',j')\in C_{\eta_3\to \eta_4}}}\frac{\Delta_j\phi_\o(x)\Delta_{j'}\phi_\o(x')}{(\phi_\o(x-x'))^2}\\
&&+\frac1{4\pi^2} \sum_{\o=\pm} \sum_{\substack{s(x,j)\in C_{\eta_1\to \eta_2}\\ s(x',j')\in C_{\eta_3\to \eta_4}}}\frac{\Delta_j\phi_{-\o}(x)\Delta_{j'}\phi_\o(x')}{|\phi_\o(x-x')|^2} e^{i (p^\o-p^{-\o})\cdot(x-x')}.\nonumber
\end{eqnarray}
Note that, since $p^\pm $ are distinct by assumption, the summand in the second line is genuinely oscillating. By using the oscillations, one finds 
that the contribution from the second line can be bounded by $O(\frac1{d+1})$, like the one from the sum over $R_{r,r'}(x,x')$; cf. with the discussion 
in \cite[Sect.3.2]{GMT17a}. We are left with the term in the 
first line, which is the Riemann sum approximation of the integral in the complex plane
\begin{equation}
\label{eq:35bis}
 - \frac1{2\pi^2}\Re\int_{\phi_+(\eta_1)}^{\phi_+(\eta_2)}dz\int_{\phi_+(\eta_3)}^{\phi_+(\eta_4)}dz' \frac {1
    }{(z-z')^2}
\end{equation}
whose explicit evaluation gives the main term in the r.h.s. of \eqref{eq:35}.}

\subsection{The interacting case: main results}

In the presence of the interaction, $\l\not=0$, Kasteleyn theory is not valid anymore, 
so that one cannot rely on an explicit computation of the dimer correlations to check
the validity of the asymptotic Gaussian behavior of the height function. 
However, dimer correlations can be written as a renormalized expansion
based on multiscale analysis.
From now on, we will assume that the interaction is small:
\begin{eqnarray}
  \label{eq:ll0}
  |\lambda|\le \e
\end{eqnarray}
and all claims above hold if $\e$ is small enough (uniformly in $L$).

Our first result is:
\begin{Theorem}\label{th:1}  Given a local function $g$ of the dimer configuration, the limit
  \begin{eqnarray}
    \label{eq:exlim}
    \mathbb E_\l(g):=\lim_{L\to\infty} \mathbb E_{L,\l}(g)
  \end{eqnarray}
  exists. The infinite-volume dimer-dimer correlations are given by
  \begin{eqnarray} &&\label{eq:32xx}
  \mathbb E_{\l}(\mathds 1_{e};\mathds 1_{e'})=  \bar A_{r,r'}(x,x')+\bar B_{r,r'}(x,x')+\bar R_{r,r'}(x,x')\label{eq:32a}\\
  &&\bar A_{r,r'}(x,x')=\frac1{4\pi^2} \sum_{\o=\pm} \frac{\bar K_{\o, r}
    \bar K_{\o, r'}
  }{\bar \phi_\o(x-x')^2}\label{eq:AAl}   \\
  &&\bar B_{r,r'}(x,x')=\frac1{4\pi^2} \sum_\o \frac{\bar H_{-\o, r}
    \bar H_{\o, r'} }{|\bar \phi_\o(x-x')|^{2\n} } e^{i
    (\bar p^\o-\bar p^{-\o})\cdot(x-x')} \label{eq:BBl} \end{eqnarray}
where:
\begin{itemize}
\item $r=r(e)$ is the type of the edge $e$, $x=x(e)$ is the coordinate of the black site of $e$, and similarly for $r',x'$;
  
\item $\bar \phi_\o(x)=\o(\bar \beta_\o x_1-\bar \alpha_\o x_2)$;
  
\item one has
  \begin{eqnarray}
\label{eq:analitiche}
&&\n=1+O(\l)\in \mathbb R,\label{nu1+O}\\ && \bar K_{\o, r}=K_{\o, r}+O(\l)\in \mathbb C,\quad \bar H_{\o, r}= K_{\o, r}+O(\l)\in\mathbb C\label{KorHor1}\\
&&\label{eq:alphabar}
   \bar \a_\o=\alpha_\o+O(\l)\in \mathbb C,\quad \bar\beta_\o=\beta_\o+O(\l)\in \mathbb C,\\
&&\label{eq:pbaro}\bar p^\o= p^\o+O(\l)\in [-\pi,\pi]^2;
\end{eqnarray}
these are all analytic functions of $\lambda$ and satisfy the symmetries
\begin{eqnarray}
  \label{eq:symmetries}
  &&\bar  \alpha_\o^*=-\bar\alpha_{-\o},\quad \bar\beta_\o^*=-\bar\beta_{-\o},\\
  &&\bar K_{\o, r}^*=\bar K_{-\o,r},\quad \bar H_{\o, r}^*=\bar H_{-\o,r}\label{KorHor2}
 \\ && \bar p^++\bar p^-=(\pi,\pi)\label{eq:pbaro2}.
\end{eqnarray}
\end{itemize}
Finally, $\bar R_{r,r'}(x,x')=O(|x-x'|^{-5/2})$ (the exponent $5/2$ could be replaced by any $\d<3$ provided $\l$ is small enough). 
\end{Theorem}
[A warning on notation: given a quantity (such as $\alpha_\o,\phi_\o$)
referring to the non-interacting model, the corresponding
$\lambda$-dependent quantity for the interacting model will be
distinguished by a bar, such as $\bar\alpha_\o,$ etc. On the other
hand, we denote by $z^*$ the complex conjugate of a number $z$.]

Note that the interaction modifies the decay rate of the correlation,
producing a non-trivial (`anomalous') critical exponent $\nu$.  The
analytic functions appearing in \eqref{eq:analitiche} are expressed as
convergent power series but, due to the complexity of the expansion,
the coefficients can be explicitly evaluated only at the lowest
orders. This makes impossible to verify directly the validity of
relations like \eqref{eq:32x}, which were essential for the proof of
large-scale Gaussian behavior of the height field in the
non-interacting case.  However, we can prove non-perturbatively that
the parameters appearing in \eqref{eq:32xx} are not independent, but
related by exact relations, which are the central result of the
present work:
\begin{Theorem} \label{th:2} One has
\begin{equation}\label{eq:32xl}
\sum_{e\in s(x,j)}
 \sigma_e \bar K_{\o,r(e)}=-i\o\sqrt{\n}\D_j\bar \phi_\o,
\end{equation}
where  $\nu=\nu(\l)$ is the same as the critical exponent in Theorem
\ref{th:1}.
Here, $s(x,j)$ is the elementary step in direction $+\vec e_{j}$
centered at $x$, thought of as a collection of two bonds, as defined
before \eqref{eq:32x}. 
As a consequence,
\begin{multline}
\label{eq:35l}
\mathbb E_\l\left[(h(\eta_1)-h(\eta_2));(h(\eta_3)-h(\eta_4))
\right]\\=\frac \nu{2\pi^2}
    \Re\log \left(\frac{(\bar \phi_+(\eta_4)-\bar \phi_+(\eta_1))(\bar \phi_+(\eta_3) - \bar \phi_+(\eta_2))}{(\bar \phi_+(\eta_4)-\bar \phi_+(\eta_2))(\bar \phi_+(\eta_3)-\bar \phi_+(\eta_1))}
    \right)\\
+O\left(\frac1{\min_{i\ne j\le 4}|\eta_i-\eta_j|^{1/2}+1}\right)
\end{multline}
(the exponent $1/2$ could be replaced by any $\d<1$ provided $\l$ is
small enough; as in \eqref{eq:35}, when $\eta_i=\eta_j$,
$\bar \phi_+(\eta_i)-\bar \phi_+(\eta_j)$ has to be read as $1$).
\end{Theorem}
Note that the result contains two non-trivial pieces of information: first, the
sum of $\sigma_e \bar K_{\o,r(e)}$ along a step in direction
$\vec e_i$ is proportional to the discrete gradient of $\bar \phi_\o$
in the same direction; second, the coefficient of proportionality is
related in an elementary way to the critical exponent $\nu$ that
appears in \eqref{eq:BBl}.  The latter relation immediately implies
the identity (cf. \eqref{eq:35l}) between height fluctuation amplitude
and critical exponent $\nu$ and is a form of universality. % Such
% relation is analogous to the so-called Haldane relations in
% condensed-matter physics \cite{Ha,BMdrude,BMun}; as discussed in the introduction, 
% it is also strictly related to one of the Kadanoff scaling relations, namely to the
% identity $X_p=X_e/4$ \cite[Eq.(13a)]{K} between the polarization critical exponent $X_p$ and the 
% energy critical exponent $X_e$ of the Ashkin-Teller model (or perturbations thereof).  

\begin{Remark}
  \label{rem:frattaglie}
  Recall that for the non-interacting model $\nu=1$, in particular it
  is independent of the weights $t_i$. This is not true {any more} for
  the interacting model. Indeed, an explicit calculation of $\nu$ at
  first order in $\lambda$ for the model with plaquette interaction shows a non-trivial dependence both on $\l$
  and on the weights. % \cite{frattaglie}. 
\end{Remark}

Theorem \ref{th:2} follows from a combination of exact 
relations among correlation functions of the interacting dimer model (``lattice Ward identities'') together with chiral gauge symmetry
emerging in the continuum scaling limit; it is remarkable that such a
symmetry, valid only in the continuum limit, implies
nevertheless exact relations for the coefficients of the lattice
theory.

\begin{Remark}
  The analog of Theorem \ref{th:1} has been proven in \cite{GMT17a},
  \cite{GMT17b} in the special case $t_1=t_2=t_3=1$ and with plaquette
  interaction as in \eqref{eq:plaq}, which has the same discrete
  symmetries as the lattice. In that case, for symmetry reasons one
  obtains automatically that the ratios
  $ \frac{\bar K_{\o,r}}{K_{\o,r}} $ are independent of $r,\o$ and
  that
  $ \frac{\bar \alpha_\o}{\alpha_\o}=\frac{\bar \beta_\o}{\beta_\o}, $
  the ratios being again $\o$-independent. Thus, the analog of Theorem
  \ref{th:2} is trivial in that case.

      Let us add also that, in the works \cite{GMT17a,GMT17b}, the
      existence of the $L\to\infty$ limit of the measure
      $\mathbb P_{L,\l}$ itself was not proven: instead, we modified
      the measure $\mathbb P_{L,\l}$ by an infra-red cut-off $m>0$
      (mass) and then we took the limit where first $L\to\infty$ and
      then $m\to0$.  We explain in Section \ref{sec:RG} how the need
      of the cut-off $m$ can be bypassed.
\end{Remark}

To upgrade Theorem \ref{th:2} into a statement of convergence of the
height field to a Gaussian Free Field with
covariance $$-\frac\nu{2\pi^2}\Re\log (\bar\phi_+(x)-\bar\phi_+(y)),$$
one needs to complement \eqref{eq:35l} with the statement that higher
cumulants are negligible, i.e. that, for $n>2$ and some $\theta>0$,
\begin{eqnarray}
\nonumber
  \mathbb E_\l\left[(h(\eta_1)-h(\eta_2));\ldots;(h(\eta_{2n-1})-h(\eta_{2n}))
\right]=O((\min_{i\ne j}|\eta_i-\eta_j|+1)^{-\theta}).
\end{eqnarray}
In turn, this requires an analog of \eqref{eq:32xx} for multi-dimer
correlation functions. This can be done following the ideas of
Sections \ref{sec:proveth} and \ref{sec:RG} below but, in order to
keep this work within reasonable length, we decided not to develop
this point. The interested reader may look at \cite[Theorem 3 and
Sec. 7]{GMT17a}, where the precise statements on multi-dimer
correlations and on the convergence to the GFF are given in detail for
the model with edge weights $\underline t\equiv 1$ and interaction
\eqref{eq:placchetta}.

\section{Grassmann integral representation}
\label{sec:grass}
\subsection{Kasteleyn theory}
\label{sec:Ktheory}
For the statements of this section and more details on Kasteleyn theory, we
refer the reader for instance to \cite{Kenyonnotes,KOS}.

The partition function and the correlations of the non-interacting
model \eqref{eq:Z} can be explicitly computed in determinantal form,
via the so-called Kasteleyn matrix $K$. This is a square matrix of
size $L^2\times L^2$ with rows/columns indexed by black/white vertices
$b/w$ of $\mathbb T_L$, as follows. If $b,w$ are not neighbors, then
$K(b,w)=0$. Otherwise, if $(b,w)$ is an edge of type $r$ one sets
$K(b,w)=K_r$, cf. \eqref{eq:Kr}. We actually need four Kasteleyn
matrices $K_{\bt}$, $\bt=(\theta_1,\theta_2)\in\{0,1\}^2$, where the two
indices label periodic/anti-periodic boundary conditions (depending on
whether the index is $0/1$) in the directions $\vec e_i$. To obtain
$K_{\bt}$ from $K$, one multiplies by $(-1)^{\theta_1}$
(resp. by $(-1)^{\theta_2}$) the matrix elements corresponding to edges
$(b,w)$ where $w$ has first coordinate equal $L$ and $b$ has first
coordinate equal $1$ (resp. $w$ has second coordinate equal $L$ and $b$ has
second coordinate equal $1$).  See Fig. \ref{fig:1}. Of course, $K_{00}=K$.
We have then \cite{Ka,Kenyonnotes}   that
\begin{equation}
  \label{eq:2}
  Z^0_L=\frac12\sum_{\bt\in \{0,1\}^2}c_{\bt} \det(K_{\bt})
\end{equation}
where $c_\bt\in\{-1,+1\}$ and, moreover, three of the $c_\bt$ have the same sign and the fourth one has the opposite sign.
More precisely, for the square grid, with our choice of Kasteleyn matrix, one finds
\begin{eqnarray}
  \label{eq:ctt}
 c_{\bt}=\left\{
  \begin{array}{lll}
    +1&\text{if}& \bt=(0,1) \text{ or } \bt=(1,0)\\
    (-1)^{{\bf 1}_{L=0\!\!\!\!\mod 4}} &\text{if}& \bt=(0,0) \\
    (-1)^{{\bf 1}_{L=0\!\!\!\!\mod 2}} &\text{if}& \bt=(1,1)
  \end{array}
                                                   \right.
\end{eqnarray}
(recall that we are assuming that $L$ is even).
The matrices $K_{\bt}$ are diagonalized in the Fourier basis and
\begin{eqnarray}
  \label{eq:detKtt}
  \det(K_{\bt})=\prod_{k\in  \mathcal P(\bt)}\mu(k),
\end{eqnarray}
where $\mu(\cdot)$ is as in \eqref{eq:Ft} and 
\begin{eqnarray}
  \label{eq:Dtt}
  \mathcal P(\bt)=\{k=(k_1,k_2),k_i=\frac{2\pi}L\left(n_i+\theta_i/2\right),n_i=0,\dots,L-1\}.
\end{eqnarray}

The matrices $K_{\bt}$ are not necessarily invertible (e.g.,
if $t_i\equiv 1$ then $K_{00}$ is not because $\mu(0)=0$) and 
this question will play a role in Section \ref{sec:RG}.  When the four matrices $K_{\bt}$ are invertible, the
correlation functions of the non-interacting measure can be written as
\begin{multline}
  \label{eq:corrni}
  \mathbb P_L(e_1,\dots,e_k\in M)=\frac{1}{2Z^0_L}\\
  \times\sum_{\bt\in\{0,1\}^2}c_{\bt}\det(K_{\bt}) \left[\prod_{j=1}^k K_{\bt}(b_j,w_j)\right]\det\{K^{-1}_{\bt}(w_n,b_m)\}_{1\le n,m\le k}
\end{multline}
where the edge $e_j$ has black/white vertex $b_j/w_j$.
The  inverse of the matrix $K_{\bt}$  can be computed explicitly as 
\begin{gather}
  \label{eq:8}
  K^{-1}_{\bt}(w_x,b_y)=\frac1{L^2}\sum_{k\in \mathcal P(\bt)}\frac{e^{-i k  (x-y)}}{\mu(k)}=:g_{L}^\bt(x,y),
\end{gather}
where $w_x$ (resp. $b_y$) is the white (resp. black) site with coordinate $x$ (resp. $y$).
Provided that
\begin{eqnarray}
  \label{eq:distanti}
  |k-p^\pm|\gg L^{-2}, \qquad \forall k\in \mathcal P(\bt),
\end{eqnarray}
it is easy to see that
$K^{-1}_{\bt}(w_x,b_y)=g(x,y)+o(1)$ as $L\to\infty$, where
\begin{eqnarray}
  \label{eq:Kinf}
  g(x,y):=\int_{[-\pi,\pi]^2}\frac{dk}{(2\pi)^2}\frac{e^{-i k(x-y)}}{\mu(k)}.
\end{eqnarray}
Condition
\eqref{eq:distanti} can fail for some values of $L$ and of
$\bt$. For this reason,  in Section \ref{sec:RG} the values $k^\pm_\bt\in \mathcal P(\bt)$ that are closest to the zeros of $\mu$ will be treated separately,
see in particular Sections \ref{sec:preliminari} and \ref{sec:ultimiquattro}.

Due to the {fact that $\mu$ has two simple zeros}, the matrix element $g(x,y)$ decays only
as the inverse distance between $w_x$ and $b_y$. More precisely
\begin{eqnarray}
  \label{eq:12}
  g(x,y) = \frac 1{2\pi}\sum_{\o=\pm} \frac{e^{-ip^\o(x-y)}}{\phi_\o(x-y)}+r(x,y)
\end{eqnarray}
where $r(x,y)=O(1/|x-y|^2)$ and
$\phi_\o$ was defined in \eqref{eq:Phi}.

\subsection{Grassmann representation of the generating functions}
We refer for instance to \cite{GMreview} for an introduction to
Grassmann variables and Grassmann integration; here we just recall a
few basic facts.  It is well known that determinants can be
represented as Gaussian Grassmann integrals. For our purposes, we
associate a Grassmann variable $\psi^+_x$ (resp. $\psi^-_x$) with the
black (resp. white) site indexed $x$. We denote by
$\int D\psi f(\psi)$ the Grassmann integral of a function $f$ and
since the variables $\psi^\pm_x$ anti-commute among themselves and
there is a finite number of them, we need to define the integral only
for polynomials $f$. The Grassmann integration is a linear operation
that is fully defined by the following conventions:
\begin{eqnarray}
  \label{eq:Dpsi}
\int D\psi\,
\prod_{x\in\Lambda}\psi^-_x\psi^+_x=1  ,
\end{eqnarray}
the sign of the integral changes whenever the positions of two
variables are interchanged (in particular, the integral of a monomial
where a variable appears twice is zero) and the integral is zero if
any of the $2|\Lambda|$ variables is missing. We also consider
Grassmann intergrals of functions of the type $f(\psi)=\exp(Q(\psi))$,
with $Q$ a sum of monomials of even degree. By this, we simply mean
that one replaces the exponential by its finite Taylor series
containing only the terms where no Grassmann variable is repeated.

It is well known that the definition of Grassmann integration allows one
to write the determinant of a matrix as the integral of the
exponential of the associated Grassmann quadratic form (such integral
will be called a ``Gaussian Grassmann integral'', for the obvious
formal analog with usual Gaussian integrals). In particular, {recalling the definition of $K_{\bt}$ given before \eqref{eq:2},}
\begin{eqnarray}
  \label{eq:3}
  \det (K_{\bt})=\int_{(\bt)} D\psi\, e^{S(\psi)},
\end{eqnarray}
where
\begin{eqnarray}\label{eq:S} S(\psi)=-\sum_{x,y\in \Lambda}
K_{00}(b_x,w_y)\psi^+_x\psi^-_{y}
\end{eqnarray}
and the index $(\bt)$ below the integral means that one has to identify 
\begin{eqnarray}
  \label{eq:identif}
\psi^\pm_{(L+1,x_2)}:=(-1)^{\theta_1}\psi^\pm_{(1,x_2)}, \quad \psi^\pm_{(x_1,L+1)}:=(-1)^{\theta_2}\psi^\pm_{(x_1,1)}.  
\end{eqnarray}
More compactly we write \[S(\psi)=-\sum_e E_e\]
where the sum runs over edges of $\mathbb T_L$ and, if
$e$ is an edge $(b,w)$,
\begin{eqnarray}
  \label{eq:Ee}
  E_e=K_{00}(b,w)\psi^+_{x(b)}\psi^-_{x(w)}.
\end{eqnarray}

Our goal here is to express, via a Grassmann integral, the partition
function of the interacting dimer model, and more generally the
generating function $\mathcal W_\Lambda(A)$ defined by
\begin{eqnarray}
  \label{eq:Wa}
  e^{\mathcal W_\Lambda(A)}:=\sum_{M\in \Omega_{L}}p_{L,\lambda}(M)\prod_{e}e^{A_e \mathds 1_{e}{(M)}} 
\end{eqnarray}
where the product runs over the edges of $\mathbb T_L$ and
$A_e\in \mathbb R$.  Note that $e^{\mathcal W_L(0)}$ is the partition
function and that any multi-dimer truncated correlation function of
the type
$\mathbb E_{L,\lambda} ( \mathds 1_{e_1};\dots;\mathds 1_{e_k})$ can be
obtained by differentiating $\mathcal W_\Lambda(A)$ with respect to
$A_{e_1},\dots,A_{e_k}$ and setting $A\equiv 0$.

Recall that the perturbed probability weight $p_{L,\lambda}$ depends on the local `energy function' $f$ via \eqref{tWL}-\eqref{eq:WL}. 
Without loss of generality, we can assume that \eqref{eq:WL} holds with 
\begin{eqnarray}
  \label{eq:wlog}
f(M)=\sum_{s=1}^n c_s\mathds 1_{P_s}(M)  
\end{eqnarray}
where $c_s$ are real constants, $n$ is an integer, $P_s$ are finite
collections of edges such that no space translation of $P_s$ coincides
with a $P_{s'}, s\ne s'$ and
$\mathds 1_{P_s}=\prod_{e\in P_s}\mathds 1_e$ is the indicator that all
edges in $P_s$ belong to $M$. Again without loss of generality we
assume that each $P_s$ contains at least $2$ edges (if $P_s$ consists
in just one edge, its effect is just to modify the weights
$\underline t$). Under these assumptions, the following representation holds. 

\begin{Proposition}\label{prop:clusterexp} Let $\lambda$ be small enough. Then, 
 one has
  \begin{eqnarray}
    \label{eq:eacho}
    e^{\mathcal W_L(A)}=
    \frac12\sum_{\bt\in\{0,1\}^2}c_{\bt}
    \int_{(\bt)} D\psi\,e^{S(\psi)+V(\psi,A)}
  \end{eqnarray}
  where
  \begin{eqnarray}
    \label{eq:clusterexp}
    V(\psi,A)=-\sum_e (e^{A_e}-1)E_e+\sum_{\gamma\subset\Lambda}c(\gamma)\prod_{b\in\gamma}E_b e^{A_b}.
  \end{eqnarray}
  The first sum runs over all edges of $\mathbb T_L$ and $E_e$
  is as in \eqref{eq:Ee}. In the second sum, $\gamma$ are finite subsets
   of disjoint edges of
  $\mathbb T_L$ such that $|\gamma|\ge2$, and $c(\gamma)$ is a real constant satisfying translation
  invariance ($c(\gamma)=c(\tau_x\gamma)$) and the bound
\begin{eqnarray}
  \label{eq:ub}
|c(\gamma)|\le (a |\lambda|)^{\max\{1,b \d(\g)\}},  
\end{eqnarray}
for some constants $a,b>0$, independent of $L$, and $\d(\gamma)$ the tree distance of $\gamma$, that is, the length of the shortest tree graph on $\L$ containing $\gamma$ 
(the precise definition of $c(\gamma)$ is given below).\end{Proposition}

\begin{Remark}
  \label{rem:coniugato}
Both $S(\psi)$ and $V(\psi,A)$ are invariant under the following symmetry transformation of the Grassmann fields:
\begin{equation} \psi^\pm_x\to (-1)^x\psi^\pm_x,\qquad c\to
  c^*,\end{equation} where $c\to c^*$ indicates that all the constants
appearing in $S(\psi)$ and $V(\psi,A)$ are mapped to their complex
conjugates. Also, we used the notation $(-1)^x:=(-1)^{x_1+x_2}$. It is
straightforward to check that, under this transformation,
$E_e\to E_e$, for all the edges $e$, which clearly shows that the
considered transformation is in fact a symmetry of the Grassmann
action. This symmetry will play a role in Section \ref{sec:RG}, in
reducing the number of independent running coupling constants arising
in the multiscale computation of the Grassmann generating function.
\end{Remark}

\begin{proof}[Proof of Proposition \ref{prop:clusterexp}]
The  proposition has been proven in \cite{GMT17a} in the case of
constant weights $t_i\equiv 1$ and plaquette interaction as in
\eqref{eq:plaq};  the extension to the present situation is
rather straightforward, so we will be concise.

Let
\[
S=\{\tau_x P_s,s=1,\dots,n, x\in \Lambda\}
\]
and remark that by assumption all elements of $S$ are distinct and
contain at least two edges.  If $B\in S$ is a space translation of
$P_s$, set
\begin{eqnarray}
  \label{eq:us}
  u(B)=\exp(\lambda c_s)-1.
\end{eqnarray}
We start by writing 
\begin{multline}
  \label{oag}
  e^{\mathcal W_L(A)}=\sum_{M\in \Omega_L} w^{(A)}(M)\prod_{x\in \Lambda}\prod_{s=1}^n(1+(e^{\lambda c_s}-1) \mathds 1_{\tau_x P_s}(M))\\=Z_L^{0,(A)}
  \sum_{\sigma\subset S}\mathbb E_{L}^{(A)}\left(\prod_{B\in \sigma}u(B)\mathds 1_B(M)\right)
\end{multline}
with
\[w^{(A)}(M)=t_1^{N_1(M)}t_2^{N_2(M)}t_3^{N_3(M)}e^{\sum_{b\in M}A_b}, \qquad Z_L^{0,(A)}=\sum_{M\in \Omega_L} w^{(A)}(M) 
\]
and $\mathbb P_L^{(A)}$ the probability measure  with density $w^{(A)}(M)/Z_L^{0,(A)}$.
By manipulating the sum in the r.h.s. of \eqref{oag}, one can rewrite it as
\begin{eqnarray}
  \label{dm}
  \sum_{n\ge0}\sum^*_{\gamma_1,\dots,\gamma_n} Z_L^{0,(A)}\mathbb E_{L}^{(A)}\left(\prod_{i=1}^n \tilde c(\gamma_i)\mathds 1_{\gamma_i}(M)\right)
\end{eqnarray}
where the term $n=0$ has to be interpreted as equal to $1$ and the sum
$\sum^*$ is over non-empty, mutually disjoint subsets $\gamma_i$ of
edges of $\mathbb T_L$.  The constant $\tilde c(\gamma)$ is given as
follows.  Let $\Sigma_\gamma$ be the set of all collections of the
type $Y=\{B_1,\dots,B_{|Y|}\}$ where: $B_i\in S$, $B_i\ne B_j$ for
$i\ne j$, $\cup_i B_i=\gamma$ and such that $Y$ cannot be divided into
two non-empty sub-collections $\{B_{i_1},\dots,B_{i_k}\}$ and
$\{B_{i_{k+1}},\dots,B_{i_{|Y|}}\}$ with
$(\cup_{j\le k}B_{i_j})\cap(\cup_{j>k} B_{i_j})=\emptyset$.  Then
\begin{eqnarray}
  \tilde c(\gamma)=\sum_{Y\in\Sigma_\gamma}\prod_{B\in Y} u(B).
\end{eqnarray}
Now we rewrite \eqref{dm} as
\begin{eqnarray}
  \label{eq:dm1}
  \sum_{n\ge0}\sum^*_{\gamma_1,\dots,\gamma_n} \prod_{j=1}^n \tilde c(\gamma_j)
 \left[ \prod_{b\in \gamma_j}\partial_{A_b}\right]Z_L^{0,(A)}.
\end{eqnarray}
The partition function $ Z_L^{(A)}$ corresponds to a non-interacting dimer model, with edge-dependent weights $t_e e^{A_e}$. Then, as in \eqref{eq:2} and \eqref{eq:3} we have
\begin{eqnarray}
  \label{eq:dm2}
  Z_L^{0,(A)}=\frac12\sum_{\bt\in \{0,1\}^2}c_{\bt} \int_{(\bt)}D\psi\, e^{S(\psi)-\sum_e (e^{A_e}-1)E_e}.
\end{eqnarray}
Using expression  \eqref{eq:dm2} in \eqref{eq:dm1} one readily concludes, as in \cite{GMT17a}, that \eqref{eq:clusterexp} holds with 
\begin{eqnarray}
  \label{eq:dm3}
  c(\gamma)=(-1)^{|\gamma|}\tilde c(\gamma).
\end{eqnarray}
If $\lambda$ is small enough, it is easy to see that the bound
\eqref{eq:ub} holds.
\end{proof}

For the 6-vertex model with interaction \eqref{eq:6v}, the potential
$V$ is exactly quartic in the fields $\psi$: indeed, $c(\gamma)\ne0$
only if $\gamma$ is the pair of edges $\gamma=\{e_1,e_{2}\}$ or
$\gamma=\{e_{3},e_{4}\}$ as in Fig. \ref{fig:esempi} or a translation
thereof.  For the plaquette model with interaction \eqref{eq:plaq},
instead, $c(\gamma)$ is non-zero only if $\gamma$ is a collection of
$|\gamma|\ge 2$ adjacent parallel edges, in which case
$c(\gamma)=(-1)^{|\gamma|}(e^\lambda-1)^{|\gamma|-1}$.

In the following (in the comparison between the discrete lattice model
and the continuum reference model) we will also need the generating
function for mixed dimer and fermionic correlations. Namely, let
$\{\phi^+_x,\phi^-_x\}_{x\in\Lambda}$ be Grassmann variables that
anti-commute among themselves and with the $\psi^\pm$ variables. Then,
we let
\begin{eqnarray}
  \label{eq:Wap}
  e^{\mathcal W_L^{(\bt)}(A,\phi)}:=\int_{(\bt)} D\psi\, e^{S(\psi)+V(\psi,A)+(\psi,\phi)}
\end{eqnarray}
and
\begin{equation}
  \label{eq:BP}
    e^{\mathcal W_L(A,\phi)}:=
    \frac12\sum_{\bt\in\{0,1\}^2}c_{\bt}
    e^{\mathcal W_L^{(\bt)}(A,\phi)}.
  \end{equation}
Here, $V(\psi,A)$ is as in Proposition \ref{prop:clusterexp}, while
\[
(\psi,\phi):=\sum_{x\in\Lambda} (\psi^+_x\phi^-_x+\phi^+_x\psi^-_x).
  \]

  We define $g_{L}(e_1,\dots,e_k;x_1,\dots,x_n;y_1,\dots,y_n)$ as the
  truncated correlations associated with the generating
  function\footnote{See e.g. \cite[Remark 5]{GMT17b} for the
    conventions in the definition of derivatives with respect  to
    Grassmann variables} $\mathcal W_L(A,\phi)$:
  \begin{multline}
    \label{eq:Gaux}
g_{L}(e_1,\dots,e_k;x_1,\dots,x_n;y_1,\dots,y_n)\\:=\left.    \partial_{A_{e_1}}\dots \partial_{A_{e_k}}\partial_{\phi^-_{y_1}}\dots\partial_{\phi^-_{y_n}}\partial_{\phi^+_{x_1}}\dots\partial_{\phi^+_{x_n}}\mathcal W_{L}(A,\phi)\right|_{A\equiv 0,\phi\equiv0}.
\end{multline}
Two cases that will play a central role in the following are $k=0$, $n=1$ (the interacting propagator), and $k=n=1$ (the interacting vertex function), which deserve 
a distinguished notation.\\
{\it Interacting propagator}: 
\begin{equation} 
g_{L}(\emptyset;x;y)=\frac1{2Z_{L}}\sum_{\bt}c_{\bt}\int_{(\bt)}D\psi\, e^{S(\psi)+V(\psi,0)}\psi^-_{x}\psi^+_{y}=:G^{(2)}_L(x,y);\end{equation}
that is, $G^{(2)}_L(x,y)=\media{\psi^-_{x}\psi^+_{y}}_L$, where $\media{f}_L$ indicate the Grassmann ``average''
$\frac1{2Z_{L}}\sum_{\bt}c_{\bt}\int_{(\bt)}D\psi\, e^{S(\psi)+V(\psi,0)}f(\psi)$.

\medskip

\noindent{\it Interacting vertex function}:\\
if $\mathcal I_{e}=\partial_{A_e}V(\psi,A)\big|_{A=0}$ is the Grassmann counterpart of the dimer observable at $e$, and $e$ is an edge of type $r$ with black site labelled $z$, then
\begin{equation} g_{L}(e;x;y)=\media{{\mathcal
      I}_{e}\psi^-_{x}\psi^+_{y}}_L-\media{{\mathcal
      I}_{e}}_L\media{\psi^-_{x}\psi^+_{y}}_L=:G^{(2,1)}_{r,L}(z,x,y);\end{equation}
that is,
$G^{(2,1)}_{r,L}(z,x,y)=\media{\mathcal I_e;\psi^-_x\psi^+_y}_L$,
where the semicolon indicates truncated expectation.

\medskip
In the following we will also need a distinguished notation for the two-point dimer-dimer correlation: if $e_1,e_2$ are two edges of type $r,r'$, and black sites 
labelled $x,y$, respectively, we let 
\begin{equation} g_{L}(e_1,e_2;\emptyset;\emptyset)=:G^{(0,2)}_{r,r',L}(x,y).\end{equation}

Note that all the correlations
$g_L(e_1,\dots,e_k;x_1,\dots,x_n;y_1,\dots,y_n)$ are well defined for
any finite $L$, despite the fact that the Kasteleyn matrix $K_{\bt}$
may not be invertible for some choices of $\bt,L$. The multipoint
correlations,
$$g_{L}(e_1,\dots,e_k;x_1,\dots,x_n;y_1,\dots,y_n),$$
admit a thermodynamic limit
as $L\to\infty$, as shown in Section \ref{sec:RG}; the limit
can be expressed as a convergent multiscale fermionic
expansion and will be denoted
\[g(e_1,\ldots,e_k;x_1,\dots,x_n;y_1,\dots,y_n).\] In particular, the
thermodynamic limit of the two-point dimer-dimer correlation will be
denoted by $G^{(0,2)}_{r,r'}(x,y)$, while the $L\to\infty$ limit of
the interacting propagator and vertex function will be denoted
$G^{(2)}(x,y)$ and $G^{(2,1)}_r(z,x,y)$.

\subsection{Lattice Ward Identity}

The generating function $\mathcal W_{L}(A,\phi)$ has a gauge symmetry
property that implies certain identities (lattice Ward identities)
involving its derivatives.  These identities were derived in
\cite{GMT17b} for the model with $t_i\equiv 1$ and they hold (with the
same proof) also for the general model studied here. We recall here,
without giving the proof, the Ward Identity for the `vertex function',
but similar relations can be easily derived for higher point
correlations: for any finite $L$,
\begin{eqnarray}
&& \sum_{r=1}^4 G^{(2,1)}_{r,L}(x,y,z)=-\d_{x,z}G^{(2)}_L(y,x),\label{3.29}\\
&&\sum_{r=1}^4 G^{(2,1)}_{r,L}(x-v_r,y,z)=-\d_{x,y}G^{(2)}_L(x,z),\label{3.30}
\end{eqnarray}
with $\delta_{x,y}$ the Krokecker delta, see
\cite[Eq.(4.9)-(4.10)]{GMT17b}. By taking the difference between these
two equations, we get (see \cite[Eq.(4.17)]{GMT17b})
\begin{equation}
\label{eq:WIreticoloL}
\delta_{x,y}G^{(2)}_L(x,z)  -\delta_{x,z}G^{(2)}_L(y,x)=-\sum_{r=2}^4\nabla_{-v_r}G^{(2,1)}_{r,L}(x,y,z),\end{equation}
where $(\nabla_n f)(x,y,z):=f(x+n,y,z)-f(x,y,z)$ is the (un-normalized) discrete derivative acting on the $x$ variable. 
{By taking the limit $L\to\infty$, we see that \eqref{3.29}--\eqref{eq:WIreticoloL} also hold for the infinite volume correlation functions
$G^{(2)}(x,y), G^{(2,1)}_{r}(x,y,z)$.}

In Fourier space, we define
\begin{gather}\label{eq:FG21}
\hat G^{(2)}(p)=\sum_{x} G^{(2)}(x,0)e^{i p x}\\
\hat G^{(2,1)}_r(k,p)=\sum_{x,z}e^{-i px-ikz}
G^{(2,1)}_r(x,0,z)\\
\hat G^{(0,2)}_{r,r'}(p)=\sum_x e^{-ipx} G^{(0,2)}_{r,r'}(x,0).
\end{gather}
Then, the infinite-volume limit of
\eqref{3.29}--\eqref{eq:WIreticoloL} can be rewritten as
\begin{eqnarray}
&& \sum_{r=1}^4 \hat G^{(2,1)}_{r}(k,p)=-\hat G^{(2)}(k+p),\label{3.37}\\
&&\hat G^{(2)}(k+p)-\hat G^{(2)}(k)=\sum_{r=2}^4(e^{-ipv_r}-1)\hat G^{(2,1)}_r(k,p).\label{eq:WIF}\end{eqnarray}
In the following the asymptotic behavior at large
distances of the interacting propagator and vertex
function will be computed in terms of a reference continuum model, see next section,
which plays the role of the `infrared fixed point'
of our lattice dimer model in its Grassmann formulation. 

\section{The infrared fixed point theory}
\label{sec:IR}

In order to introduce the ``infra-red fixed point'' of our theory (referred to in the following as ``the continuum model'' or ``the reference model''), we need a couple of preliminary definitions.
First, we let $\mathcal M$ be the $2\times 2$ matrix with unit determinant
\begin{eqnarray}
  \label{eq:M}
  \mathcal M=\frac1{\sqrt{\Delta}}
  \begin{pmatrix}
    \bar\beta^1& \bar \beta^2\\
    -\bar\alpha^1&-\bar\alpha^2
  \end{pmatrix}
\end{eqnarray}
where $\bar \alpha^j,\bar\beta^j\in \mathbb R$, $j=1,2$ and
$\Delta=\bar\alpha^1\bar\beta^2-\bar\alpha^2\bar\beta^1>0$ (for the moment, these are free parameters; eventually, they will be the real and imaginary parts of the functions $\bar \a_\o,\bar\b_\o$ 
that appear in Theorem \ref{th:1}).
Also, given $L>0$ (the system size), 
an integer $N$ (ultra-violet cut-off) and $Z>0$, we introduce a
Grassmann Gaussian integration\footnote{\label{foto:Gint} We recall
  (cf. e.g. \cite[Sec. 4]{GMreview}) that, given a family
  $\{\psi^-_x,\psi^+_x\}_{x\in\mathcal I}$ of Grassmann variables and a
  $|\mathcal I|\times |\mathcal I|$ matrix $g$, the ``Grassmann Gaussian
  integration with propagator $g$'', denoted sometimes $\int P_g(d\psi)\dots$ in the
  following, is the linear map acting on polynomials of the Grassmann
  variables, such that
  $\int
  P_g(d\psi)\psi^-_{x_1}\psi^+_{y_1}\dots\psi^-_{x_n}\psi^+_{y_n}=\det
  G_n(\underline x,\underline y)$ with
  $G_n(\underline x,\underline y)$ the $n\times n$ matrix with entries
  $[G_n(\underline x,\underline y)]_{ij}= g(x_i,y_j)$. If the matrix $g$ is non-singular, one can write more explicitly
\begin{equation}
    \label{eq:esplicito}
    \int P_g(d\psi)f(\psi)=\det(g)\int D\psi\, e^{-\psi^+ g^{-1} \psi^-}\,f(\psi).    
  \end{equation}} 
  $P_Z^{[\le N]}(d\psi)$ on the family of Grassmann variables
 \[
  \{ \hat\psi^\pm_{k,\o},\; \o=\pm1, \; k\in \mathcal K\},\qquad \mathcal K=\left\{\mathcal M\cdot p\; \left|\;p\in\left(\frac{2\pi}L\right)(\mathbb Z+1/2)^2\right.\right\},
   \]
 defined by the propagator
\begin{eqnarray}
  \label{eq:slf}
  \int P_Z^{[\le N]}(d\psi)\hat\psi^-_{k,\o}\psi^+_{k',\o'}=
  \delta_{\o,\o'}\delta_{k,k'}\frac{L^2 }Z \frac{\chi_{N}(k)}{ \bar D_\o(k)}
\end{eqnarray}
where:
\begin{itemize}
\item $\chi_{N}(k)=\chi(2^{-N}|\mathcal M^{-1} k|)$, with $\chi:\mathbb R^+\to [0,1]$ a $C^\infty$ cut-off function that is equal to $1$ if its argument is smaller than $1$ and 
  equal to $0$ if its argument is larger than $2$;
\item $\bar D_\omega(k)=\bar\a_\o k_1+\bar\b_\o k_2$, with
  \begin{eqnarray}
    \label{eq:a1b1}
  \bar \a_\o=\o \bar\a^1+i \bar\a^2,\quad \bar\b_\o=\o \bar\b^1+i \bar\b^2.   
  \end{eqnarray}
 
  Observe that, since we are assuming $\Delta>0$, we have that
  \begin{eqnarray}
    \label{eq:waa}
    \frac{\bar\alpha_\o}{\bar\beta_\o}\not \in \mathbb R.
  \end{eqnarray}
  
\end{itemize}
While $\mathcal K$ is an infinite set, we effectively have only a finite number of non-zero Grassmann variables $\hat\psi^\pm_{k,\o}$, because  $\chi_N(k)$ is non-zero only for a finite number of values of $k$ in $\mathcal K$.

Note that, setting $q=\mathcal M^{-1}k$, the r.h.s. of \eqref{eq:slf} equals
\begin{equation}
  \delta_{\o,\o'}\delta_{q,q'}\frac{L^2}{Z \sqrt\Delta
    }\frac{\chi(2^{-N}|q|)}{-i q_1+\o q_2}\label{gq}.
\end{equation}
 In the language of Quantum Field Theory, in the limit
$\lim_{L\to\infty,N\to\infty}$, \eqref{gq} is just the
propagator of chiral massless relativistic fermions.

It is convenient to define, for $x\in\mathbb R^2$, the Grassmann variables
\begin{eqnarray}
\psi^\pm_{x,\o}:=\frac1{L^2}\sum_{k\in \mathcal K}e^{\pm i k x}\hat\psi^\pm_{k,\o}.
\end{eqnarray}
Note that $\psi^\pm_{x,\o}$ has anti-periodic boundary conditions on
\[
\Lambda:=(\mathcal M^T)^{-1}\mathcal T_L, \quad \mathcal T_L=\mathbb R^2/(L\mathbb Z^2)
\]
and that
\begin{equation}\label{gth}
\frac{g^{[\le N]}_{R,\o}(x-y)}Z:=\int P_Z^{[\le N]}(d\psi)\psi^-_{x,\o}\psi^+_{y,\o}=\frac{1}{ ZL^2}\sum_{k\in\mathcal K}e^{-i k(x-y)}\frac{\chi_{N}(k)}{\bar D_\o(k)}.
\end{equation}

The generating functional $\mathcal W_{L,N}(J,\phi)$ of the continuum
model is 
\begin{equation}\label{vv1}
e^{\WW_{L,N}(J,\phi)} = {\int\! P_Z^{[\le N]}(d\psi)
e^{\VV(\sqrt{Z}\psi) + \sum_{j=1}^2
(J^{(j)},\,\rho^{(j)})+
Z\,(\psi,\phi)}}
\;,\end{equation}
where $J=\{J^{(j)}_{x,\o}\}^{j=1,2}_{\o=\pm,\, x\in\L}$ are external
``sources'' (real-valued test functions) and
$\phi=\{\phi^\sigma_{x,\o}\}^{\s,\o=\pm}_{x\in\Lambda}$ are ``external
Grassmann sources'', i.e. $\phi^\sigma_{x,\o}$ is a Grassmann
variable. Also, we used the notation
\[(J^{(j)},\r^{(j)}):=\sum_{\o=\pm}\int_\L dx\;
  J^{(j)}_{x,\o}\r^{(j)}_{x,\o},\] with
\begin{equation}\label{rhodef}
\r^{(1)}_{x,\o} = \psi^+_{x,\o} \psi^{-}_{x,\o}\;,
\qquad 
\r^{(2)}_{x,\o} = \psi^+_{x,\o} \psi^{-}_{x,-\o}\;
\end{equation}
and
$$(\psi,\phi):=\sum_{\o=\pm}\int_\L dx\,(\psi^+_{x,\o}\phi^-_{x,\o}+\phi^+_{x,\o}\psi^-_{x,\o})\;.$$
Finally, the interaction $\mathcal V$ in \eqref{vv1} is
\begin{equation}\label{gjhfk} \VV(\psi)=\frac{\l_\io}2 
\sum_{\o=\pm}\int_\L dx\int_\L dy\;   v(x-y) \psi^+_{x,\o}
\psi^-_{x,\o}\psi^+_{y,-\o}\psi^-_{y,-\o}\;, \end{equation}
where $\l_\io\in\mathbb R$, $v(x)=v_0(\mathcal M^T x)$ and
$v_0(\cdot)$ is a smooth rotationally invariant potential,
exponentially decaying to zero at large distances, normalized as
\begin{eqnarray}
  \label{eq:normaas}
\int_{\mathbb R^2}dx\, v_0(x)=  \int_{\mathbb R^2}dx\, v(x)=1.
\end{eqnarray}
  We emphasize that, while this
expression seems to depend on an uncountable set of Grassmann
variables $\{\psi^\pm_{x,\o},\phi^\pm_{x,\o}\}_{x\in\L}$, writing
everything in Fourier space there is only a finite number of non-zero
Grassmann variables.

In the special case $\bar\alpha_\o=(-i-\o),\bar\beta_\o=(-i+\o)$, {which}
is relevant for the interacting dimer model with
$\underline t\equiv 1$, the continuum model reduces to that studied in
\cite[Sec. 5]{GMT17b}, if the
constants $Z^{(1)}$ and $Z^{(2)}$ that appear there are fixed to $1$.  Setting
instead $\bar\alpha_\o=-i, \bar\beta_\o=\o$ in \eqref{vv1} (so that
$\Delta=1$) one obtains, apart from minor differences, the model
studied in \cite[Sec. 3]{BFM1} and \cite[Sec. 3]{BMdrude}.
\begin{Remark} In order to recognize the equivalence of the model \eqref{vv1} with  $\bar\alpha_\o=-i, \bar\beta_\o=\o$
and the one in, e.g., \cite[Section 3]{BMdrude} (or, analogously, the one in \cite[Section 3]{BFM1}), one needs to set to zero some of the external fields, rotate the coordinate system and 
rescale some constants. More precisely, if $\mathbb W_{L,N}(J^{(1)},\phi)$ denotes the generating functional used in \cite{BMdrude} with $J^{(1)}_{x,\o}=Z^{(3)}J_x+\o\tilde Z^{(3)}\tilde J_x$, see 
\cite[Eq. (28)]{BMdrude}, then, setting $J^{(2)}_x\equiv 0$ in \eqref{vv1},
\begin{eqnarray}\label{eq:assurdita}
  \mathcal W_{L,N}((J^{(1)},0),\phi;\lambda_\infty)=const.+\mathbb W_{L,N}(\mathcal J^{(1)},\varphi;-\Delta^{-1}\lambda_\infty)
\end{eqnarray}
where the constant is independent of $J^{(1)},\phi$ (so that it does not influence the
correlation functions; it depends upon $\Delta$ and is due to the rescaling of the Grassmann fields), while
\begin{eqnarray}\label{eq:assurdo2}
\mathcal J^{(1)}(x):=\Delta^{1/2}J^{(j)}((\mathcal M^T)^{-1}x), \quad \varphi^\pm(x):=\Delta^{1/4}\phi^\pm((\mathcal M^T)^{-1}x),
\end{eqnarray}
  and we denoted explicitly the dependence of the generating function on $\lambda_\infty$.
 This immediately implies obvious relations between  the correlation functions
  $G^{(2,1)}_{R,\o',\o}(x,y,z)$, $G^{(2)}_{R,\o}(x,y)$ and
  $S^{(j,j)}_{R,\o,\o'}(x,y)$, defined below, and the analogous ones of \cite{BMdrude}.
  \end{Remark}

  The peculiarity of the continuum model is that its correlations can
  be computed exactly. This is because, as compared to its lattice
  counterpart, the continuum model is ``chiral gauge invariant'',
  which means that the correlation functions satisfy two hierarchies
  of Ward Identities, distinguished by the choice of the `chirality
  index' $\omega$, see \eqref{h11} below. These additional symmetries,
  together with other identities among correlation functions (the
  so-called Schwinger-Dyson equations), allow one to get closed
  equations for correlations functions.  In this sense, the infrared
  fixed point theory can be regarded as ``integrable''.

We {define the following correlation functions of the reference model}: if $x,y,z$ are distinct 
points of $\L$,
\begin{eqnarray} &&  G^{(2,1;L,N)}_{R,\o',\o}(x,y,z) =  \frac{\partial^3}{\partial J_{x,\o'}^{(1)}\partial\phi^-_{z,\o}\partial\phi^+_{y,\o} }
\WW_{L,N}(J,\phi)|_{J=\phi=0} \nn\\
&& G^{(2;L,N)}_{R,\o}(x,y):=\frac{\partial^2}{ \partial\phi^-_{y,\o}
\partial\phi^+_{x,\o}} \WW_{L,N}(J,\phi)|_{J=\phi=0} \label{eq:5.6op}\\
&& S^{(j,j;L,N)}_{R,\o,\o'}(x,y):=\frac{\partial^2}{\dpr J_{x,\o}^{(j)} \dpr
J_{y,\o'}^{(j)}}\WW_{L,N}(J,\phi)|_{J=\phi=0} \nn.
\end{eqnarray}
 From the construction of the correlation functions of the 
model, see e.g. \cite[Section  3 and 4]{BFM1}, one obtains in particular the existence of the following limits where cut-offs are removed:
\begin{eqnarray}
  &&G^{(2,1)}_{R,\o',\o}(x,y,z) =\lim_{L\to\infty}\lim_{N\to\infty} G^{(2,1;L,N)}_{R,\o',\o}(x,y,z)\;,\nn\\
  &&G^{(2)}_{R,\o}(x,y) =\lim_{L\to\infty}\lim_{N\to\infty}G^{(2;L,N)}_{R,\o}(x,y) \;,\label{eq:5.6}\\
  &&S^{(j,j)}_{R,\o,\o'}(x,y)=\lim_{L\to\infty}\lim_{N\to\infty}S^{(j,j;L,N)}_{R,\o,\o'}(x,y)\;.\nn
\end{eqnarray}
Away from $x=0$, the so-called ``density-density''
correlation $S^{(1,1)}$ is given by \cite[Eq. (5.12)]{GMT17b}
\begin{equation}
\label{googleee}
S^{(1,1)}_{R,\o,\o}(x,0) = 
\frac1{4\pi^2 Z^2(1-\t^2)}\frac1{(\bar\phi_\o(x))^2}+R_1(x),
\end{equation}
where $ \bar\phi_\o(x):=\o(\bar\beta_\o x_1-\bar\alpha_\o x_2)$,
\begin{equation}
 \label{eq:deftau} \tau=-\frac{\lambda_\infty}{4\Delta\pi}
  \end{equation}
and $|R_1(x)|\le C|x|^{-3}$. On the other hand, the ``mass-mass correlation'' $S^{(2,2)}$ satisfies (see \cite[Eq.(6.14)]{GMT17b})
\begin{equation}
\label{instangraaam}S^{(2,2)}_{R,\o,-\o}(x,0) =
\frac{\bar B}{4\pi^2Z^2}\frac1{|\bar\phi_\o(x)|^{2\nu}}+R_2(x),
\end{equation}
where $\bar B$ is an analytic function of $\l_\infty, Z, \bar\a_\o,\bar\b_\o$, which is equal to $1$ at $\l_\infty=0$,
\begin{equation}\label{nunonnu}\nu=\frac{1-\t}{1+\t},\end{equation}
see \cite[Eq.(6.15)]{GMT17b} and \cite[Appendix C]{GMT17b}, and $R_2$ is a correction term such that $|R_2(x)|\le C|x|^{-2-\th}$, for some $\th>0$ that, e.g., can be chosen $\th=1/2$.

We will not need the explicit form of $G^{(2)}_{R,\o}(x,0)$ and
$ G^{(2,1)}_{R,\o',\o}(x,0,z)$; let us just mention that they diverge
as $x,z$ tend to zero but they are locally integrable functions (see,
e.g., the expression of the interacting propagator in
\cite[eq.(4.18)]{BFM1}) and therefore admit Fourier transforms in the
sense of distributions\footnote{On the other hand, the notion of
  Fourier transform for $S^{(j,j)}_{R,\o,\o'}(x,y)$ requires a little
  more care. Regarding $S^{(1,1)}_{R,\o,\o'}(x,y)$, from its
  expression one sees that it is not locally integrable; still, it
  defines a tempered distribution if the singularity at the origin is
  interpreted in the sense of the principal part: therefore, its
  Fourier transform $\hat S^{(1,1)}_{R,\o,\o'}(p)$ exists in the sense
  of distributions.  This is not the case for
  $S^{(2,2)}_{R,\o,\o'}(x,y)$ when $\nu\ge1$ (in particular, when
  $\l=0$, where $\nu=1$) since $1/|x|^{2\nu}$ is not locally
  integrable on $\mathbb R^2$. In this respect,
  \cite[eq.(6.2)]{GMT17b} does not make sense as is: however, that
  equation is correct if $\hat S^{(2,2)}_{R,\o,-\o}(p)$ is replaced by
  $\tilde S^{(2,2)}_{R,\o,-\o}(p)$, that is the Fourier transform of
  $ S^{(2,2)}_{R,\o,-\o}(x,0)$ multiplied by a $C^\infty$ function
  that vanishes for $|x|\le 1/2$ and equals $1$ for $|x|\ge 1$.},
\begin{eqnarray}
&&\hat G^{(2,1)}_{R,\o',\o}(k,p)=
\int dx \int dy\, e^{-ipx +i(k+p)y}\, G^{(2)}_{R,\o',\o}(x,y,0) \;,\nn\\
&&\hat G^{(2)}_{R,\o}(k) = 
   \int dx\, 
e^{ikx}\,G^{(2)}_{R,\o}(x,0).\label{eq:FT}
\end{eqnarray}

For later use, let us mention that the small-momenta behavior of $\hat G^{(2,1)}_{R,\o,\o}$ and $\hat G^{(2)}_{R,\o}$ are (see, e.g., \cite[Theorem 2]{BM02})
\begin{equation}
  \label{eq:Ghat2R}
  |\hat G^{(2)}_{R,\o}(p)|\sim const\times {|p|}^{-1+O(\lambda_\infty^2)}, \quad {\text{as}\quad p\to 0},
 \end{equation}
{and, if $0<\mathfrak c\le |p|,|k|,|k+p|\le 2\mathfrak c$,} 
\begin{equation}  \label{eq:Ghat21R} |\hat G^{(2,1)}_{R,\o,\o}(k,p)|\sim const\times  \mathfrak c^{-2+O(\lambda_\infty^2)}, \quad {\text{as}\quad \mathfrak{c}\to 0}.
\end{equation}
A very useful consequence of the exact solution of the continuum model
is that the ``propagator'' and the ``vertex function'' satisfy the
following Ward Identity (see \cite[Eq.(5.9)]{GMT17b}):
\begin{equation}\label{h11}
{Z}\sum_{\omega'=\pm}\bar D_{\omega'}(p)\hat G^{(2,1)}_{R,\omega',\o}(k,p)=
\frac1{1-\tau \hat v(p)} [\hat G^{(2)}_{R,\o}(k) - \hat G^{(2)}_{R,\o}(k+p)]\;.\end{equation}
Note that this identity resembles formally the lattice Ward identity \eqref{eq:WIF} of the dimer model, with the crucial difference that \eqref{h11} are actually \emph{two} identities (one for each choice of $\o$).

\section{Comparison between lattice and continuum model, and proof of
  Theorems \ref{th:1}-\ref{th:2}}

\label{sec:proveth}

The reason why the continuum model plays the role of the ``infrared
fixed point theory'' for our interacting dimer model is that the large
distance behaviour of the dimer correlation functions can be expressed
in terms of linear combinations of the correlations of the continuum
model, for a suitable choice of the parameters
$Z,\l_\infty,\bar\a_\o,\bar\b_\o$.  Let us spell out the explicit
relation between correlation functions of the two models, in the
special cases of the dimer interacting propagator, the vertex function
and the dimer-dimer correlation.  The result is a consequence of the
multi-scale analysis described in Section \ref{sec:RG} {where, in particular, 
 we prove the following:}
\begin{Proposition}\label{prop:comp} {In a small neighborhood of  $\l=0$ in the complex plane, there exist:
\begin{enumerate}
\item two real-valued analytic functions  \footnote{By `real-valued analytic' we mean a function that is analytic in 
 a small complex neighbourhood of the origin, and such that it is real-valued for real values of $\l$.}
$\l\mapsto\bar p^\o$, with $\o=\pm$, called the interacting Fermi
points, satisfying \eqref{eq:pbaro} and \eqref{eq:pbaro2},
%which are the only singularity points of the Fourier transform of the interacting dimer propagator $\hat G^{(2)}(\cdot)$ of \eqref{eq:FG21}. 
%In addition, there exist 
\item four complex-valued  analytic
functions $\lambda\mapsto\bar\a_\o,\lambda\mapsto\bar\b_\o$ satisfying \eqref{eq:alphabar}, \eqref{eq:symmetries},
\item two real-valued analytic functions $\lambda\mapsto Z,\lambda\mapsto\l_\infty$ satisfying $Z=  1+O(\l)$ and $\l_\infty=O(\l)$, 
\end{enumerate}
such that the dimer-dimer correlation can be represented in the following form:
\begin{eqnarray}\label{hh110}
&&G^{(0,2)}_{r,r'}(x,y) = \sum_{\o=\pm} \hat K_{\o,r} \hat K_{\o,r'} S^{(1,1)}_{R,\o,\o}(x,y)\\
&&\qquad + \sum_{\o=\pm} e^{i(\bar p^\o-\bar p^{-\o})(x-y)}  \hat  H_{-\o,r} \hat H_{\o,r'}  
S^{(2,2)}_{R,\o,-\o}(x,y)+R_{r,r'}(x,y)\nonumber\;,
\end{eqnarray}
where: $\lambda\mapsto\hat K_{\o,r}$ and $\lambda\mapsto \hat H_{\o,r}$ are  complex-valued analytic function of $\l$ satisfying 
$\hat K_{+,r}=\hat K^*_{-,r}$, $\hat H_{+,r}=\hat H^*_{-,r}$, $\hat K_{\o,r}=K_{\o,r}+O(\l)$, and $\hat H_{\o,r}=K_{\o,r}+O(\l)$; the correction term $R_{r,r'}(x,y)$  is translational invariant and satisfies $|R_{r,r'}(x,0)|\le C |x|^{-5/2}$.}

{Moreover, the interacting dimer propagator satisfies
\begin{equation} \hat G^{(2)}(k+\bar p^\o) \stackrel{\mathfrak c\to0}= \hat G^{(2)}_{R,\o}(k)[1+O(|k|^\th)],\label{h10ab}\end{equation}
for some $\th>0$\footnote{We can choose $\th=1/2$; from now on, this is the choice that the reader should keep in mind, unless otherwise stated.}.
Finally,  if $0<\mathfrak c\le |p|,|k|,|k+p|\le 2\mathfrak c$, then the interacting vertex function of the dimer model 
satisfies 
\begin{equation}\label{h10a}
\hat G^{(2,1)}_r(k+ \bar p^\o, p)
\stackrel{\mathfrak c\to0}= -\sum_{\o'=\pm}\hat K_{\o',r}\hat G^{(2,1)}_{R,\o',\o}(k,p)[1+O(\mathfrak c^\th)]\;.\end{equation}}
\end{Proposition}

{The statements contained in this proposition are proved in different subsections of Section \ref{sec:RG}: the 
construction of the functions $\bar p^\o, \bar\a_\o, \bar\b_\o$ is given in Sect.\ref{seccount}; the construction of the function $\l_{\infty}$ is given in 
Sect.\ref{remmm}, while the construction of $Z$ is given in 
Sect.\ref{bareZ}; the proof of \eqref{hh110} is given in Sect.\ref{sec:asympt}, and \eqref{h10ab}-\eqref{h10a} can be proved along the same lines.}

\subsection{Proof of Theorem \ref{th:1}}
The proof of the existence of the thermodynamic limit for the average of all the local functions of dimer configurations, Eq.\eqref{eq:exlim}, is postponed to 
Section \ref{sec:ultimiquattro}, see in particular the comments after \eqref{6.105}. The other statements in Theorem \ref{th:1} follow easily
from Proposition \ref{prop:comp}. In fact, if we rewrite Equation \eqref{hh110} by using Equations \eqref{googleee} and \eqref{instangraaam}, we obtain Equation \eqref{eq:32xx} with
\begin{equation}\label{klbark}\bar K_{\o,r}=\hat K_{\o,r}\frac1{Z\sqrt{1-\t^2}}, \qquad \bar
H_{\o,r}=\hat H_{\o,r}\frac{\sqrt{\bar B}}{Z}.
\end{equation}
which satisfy the desired properties, \eqref{KorHor1} and
\eqref{KorHor2}, thanks to the stated properties of
$Z, \hat K_{\o,r},\hat H_{\o,r}$ and of $\bar B$. Recall that $\tau$ is defined in \eqref{eq:deftau}, and the critical exponent $\nu$ is given in
Equation \eqref{nunonnu}: therefore, the fact that $\nu=1+O(\l)$, see \eqref{nu1+O}, follows from
the definition there and from the fact that $\l_\infty=O(\l)$. The other properties of $\bar p^\o,\bar\a_\o,\bar\b_\o$ stated in Theorem \ref{th:1} 
are the same as those stated in items (1)-(2) of Proposition \ref{prop:comp}.

\subsection{Proof of Theorem \ref{th:2}}

The key ingredient in the proof of Theorem \ref{th:2} is the analogue
of \eqref{eq:32x}-\eqref{eq:32x.1} for the interacting case, namely
formula \eqref{eq:32xl}. We start by discussing the proof of this
formula {which}, as we shall see, is a direct consequence of the
identities \eqref{h10ab}--\eqref{h10a}, and of the lattice Ward
Identity \eqref{eq:WIF}. In fact, combining these three identities,
we obtain: 
\begin{eqnarray}\label{xa1}
&& \sum_{\o'=\pm}\mathcal D_{\o'}(p)\hat G^{(2,1)}_{R,\o',\o}(k,p) 
= \left[ \hat G^{(2)}_{R,\o}(k) - \hat
G^{(2)}_{R,\o}(k+p) \right ] [1+O(\mathfrak c^\th)],\qquad {\phantom{\cdot}}
\end{eqnarray}
where (with $v_r$ as in \eqref{vii})
$$\mathcal D_{\o'}(p)=-i \sum_{r=2}^4\hat K_{\o',r}\,p\cdot v_r$$ and,
as before, $0<\mathfrak c\le |p|,|k|,|k+p|\le 2\mathfrak c$.  By
comparing this equation with \eqref{h11}, and recalling that
$\hat v(0)=1$ (see \eqref{eq:normaas}) we get
\begin{equation} Z(1-\t)\sum_{\omega'=\pm}\bar D_{\omega'}(p)\hat G^{(2,1)}_{R,\omega',\o}(k,p)=\sum_{\o'=\pm}\mathcal D_{\o'}(p)\hat G^{(2,1)}_{R,\o',\o}(k,p) 
[1+O(\mathfrak c^\th)]\;.\label{z1tau}\end{equation}
This implies that 
\begin{equation}\label{tildenontilde}-i\sum_{r=2}^4\hat K_{\o,r}(v_r)_1=
  Z(1-\t)\bar\a_\o,\qquad
-i\sum_{r=2}^4\hat K_{\o,r}(v_r)_2=
Z(1-\t)\bar\b_\o.\end{equation} In order to
deduce \eqref{tildenontilde} from \eqref{z1tau}, one can proceed as follows: by  \cite[Eq. C.24]{GMT17b}, we have that
\begin{equation}\label{G21oo} G^{(2,1)}_{R,-\omega,\o}(k,p)=\tau \hat v(p)\frac{\bar D_\o(p)}{\bar D_{-\o}(p)}\hat G^{(2,1)}_{R,\omega,\o}(k,p).\end{equation}
By plugging this identity into \eqref{z1tau} we get  (keeping the terms of dominant order as $p\to 0$ only and using $\hat v(0)=1$):
\begin{equation}% \sqrt{\frac2\Delta}
Z(1-\t^2)\bar D_{\omega}(p)\bar D_{-\omega}(p)=\mathcal D_\o(p) \bar D_{-\omega}(p)+
\tau \bar D_{\omega}(p)\mathcal D_{-\omega}(p).\end{equation}
Computing this formula at $p_2=0, p_1\neq0$ first, both for $\o=+$ and $\o=-$ and then repeating
the computation for $p_1=0, p_2\neq0$, one gets a system of linear
equations for the coefficients $-i\sum_{r=2}^4\hat K_{\o,r}(v_r)_j$, with $j=1,2, \o=\pm$, whose
solution is \eqref{tildenontilde}.

By replacing \eqref{klbark} into \eqref{tildenontilde} and recalling that $\nu=\frac{1-\t}{1+\t}$, cf. \eqref{nunonnu},
we find
\begin{equation}\label{tildenonbar} \bar K_{\o,2}+\bar K_{\o,3}=-i\sqrt\nu \bar\a_\o,\qquad 
  \bar K_{\o,3}+\bar K_{\o,4}=-i\sqrt\nu \bar\b_\o.\end{equation}
We claim that $\sum_{r=1}^4\bar K_{\o,r}=0$ (we shall prove this fact in a moment): therefore, the first equation can be rewritten as
$\bar K_{\o,1}+\bar K_{\o,4}=i\sqrt\nu \bar\a_\o$. In terms of the `elementary steps' $s(x,j)$ in direction $\vec e_{j}$ centered at $x$, introduced before \eqref{eq:32x}, the two equations in 
\eqref{tildenonbar} become
\begin{eqnarray}\label{eq:bbhh.1}
&&\hskip-1.truecm\sum_{e\in s(x,1)}
\s_e \bar K_{\o,r(e)}=-i\sqrt\nu \bar\b_\o=-i\o\sqrt\n\D_1\bar \phi_\o\\
&&\hskip-1.truecm\sum_{e\in s(x,2)}
\s_e \bar K_{\o,r(e)}=i\sqrt\nu \bar\a_\o=
-i \o\sqrt{\nu}\D_2\bar \phi_\o,\label{eq:bbhh.2}
\end{eqnarray}
{which are the two cases of Equation \eqref{eq:32xl}.}

In order to complete the proof of \eqref{eq:bbhh.1}-\eqref{eq:bbhh.2},
we need to prove that $\sum_{r=1}^4\bar K_{\o,r}=0$, as claimed
above. For this purpose, we consider \eqref{3.37}, and combine it with
\eqref{h10ab}-\eqref{h10a}, thus getting, if
$0<\mathfrak c\le |p|,|k|,|k+p|\le 2\mathfrak c$
\begin{equation}
\sum_{r=1}^4\sum_{\o'=\pm}\hat K_{\o',r} \hat G^{(2,1)}_{R,\o',\o}(k,p)[1+O(\mathfrak c^\th)]\stackrel{
\mathfrak c\to0}=\hat G^{(2)}_{R,\o}(k+p)[1+O(\mathfrak c^\th)].\end{equation}
{Using \eqref{G21oo} to rewrite the left hand side and recalling that $\hat v(p)=1+O(p)$}, this becomes
\begin{eqnarray}
  &&\hat G^{(2,1)}_{R,\o,\o}(k,p)\sum_{r=1}^4\Big(\hat K_{\o,r}+\t\hat K_{-\o,r}\frac{\bar D_\o(p)}{\bar D_{-\o}(p)}\Big)[1+O(\mathfrak c^\th)]=\\
  &&\hskip6.9truecm =\hat G^{(2)}_{R,\o}(k+p)[1+O(\mathfrak c^\th)].\nonumber\end{eqnarray}
Now, {recalling that the magnitude of the correlation functions for small $p,k$ has the form given in 
\eqref{eq:Ghat2R} and \eqref{eq:Ghat21R}, for this to hold in the limit $\mathfrak c\to0$ we must have}
\begin{equation}\label{eq5:15}\sum_{r=1}^4\Big(\hat K_{\o,r}+\t\hat K_{-\o,r}\lim_{p_j\to 0}\frac{\bar D_\o(p_j)}{\bar D_{-\o}(p_j)}\Big)=0,\end{equation}
for any sequence $p_j$ along which the ratio
$\bar D_\o(p_j)/\bar D_{-\o}(p_j)$ admits a limit. Note that, in
general, the limit depends upon the chosen subsequence. For instance,
if $p_j=(s_j,0)$ with $s_j\to0$ then the limit is
$-\alpha_\o^2/|\alpha_\o|^2$ while if $p_j=(0,t_j)$ with $t_j\to0$ the
limit is $-\beta_\o^2/|\beta_\o|^2$. On the other hand, these two
values cannot be equal since we know that the ratio
$\alpha_\o/\beta_\o$ is not real (cf. \eqref{eq:waa}).  {Consequently,
\eqref{eq5:15} implies} that $\sum_{r=1}^4\hat K_{\o,r}=0$ that, in
light of \eqref{klbark}, is equivalent to
$\sum_{r=1}^4 \bar K_{\o,r}=0$, as desired.

\medskip

With the identities \eqref{eq:bbhh.1}-\eqref{eq:bbhh.2} at hand, we
can easily prove \eqref{eq:35l}, by repeating the analogue of the
discussion leading, in the non-interacting case, to
\eqref{eq:35bis}. We will be very sketchy since the analogous argument
has been given in detail in \cite{GMT17a} in the case of the model
with weights $\underline t\equiv 1$.  We start from the very definition of the
covariance of the height difference:
\begin{equation} 
\label{covariance}
\mathbb E_\l\left[(h(\eta_1)-h(\eta_2));(h(\eta_3)-h(\eta_4))
\right]= \sum_{e\in C_{\eta_1\to \eta_2}} \sum_{e'\in C_{\eta_3\to \eta_4}}\sigma_e
\sigma_{e'} \mathbb E_\l(\mathds 1_e;\mathds 1_{e'}), \end{equation}
where $C_{\eta_1\to \eta_2}$ and $C_{\eta_3\to \eta_4}$ are two lattice paths connecting $\eta_1$ with $\eta_2$, and $\eta_3$ with $\eta_4$, respectively. 
For simplicity, we assume that $\eta_1$ and $\eta_2$ have the same parity, and similarly for $\eta_3$ and $\eta_4$: in this way, it is possible to choose the two paths 
$C_{\eta_1\to \eta_2}$ and $C_{\eta_3\to \eta_4}$ to be concatenations of `elementary steps' $s(x,j)$ in directions $\pm\vec e_j$, see the discussion after 
\eqref{eq:Kr} above. 
For simplicity, let us also assume  that the mutual distances between the faces $\eta_1,\dots \eta_4$ are all comparable, i.e.
\begin{eqnarray}
  \label{eq:matassa}
0<  c<\frac{\min_{i\ne j}|\eta_i-\eta_j|}{\max_{i\ne j}|\eta_i-\eta_j|}.
\end{eqnarray}
In this case, we choose the two paths $C_{\eta_1\to \eta_2}$ and
$C_{\eta_3\to \eta_4}$ to be of length at most
$C \max_{i\ne j}|\eta_i-\eta_j|$ and to be at mutual distance
$C^{-1} \max_{i\ne j}|\eta_i-\eta_j|$, for some constant $C=C(c)$.

We now insert \eqref{eq:32xx} into \eqref{covariance} and, by repeating the discussion of \cite[Section 3.2]{GMT17a}, we find that the dominant contribution comes from $\bar A_{r,r'}$ (the contribution from $\bar B_{r,r'}$ is sub-dominant due to the oscillating pre-factors):
\begin{eqnarray} 
\label{covariance_2}
  &&\mathbb E_\l\left[(h(\eta_1)-h(\eta_2));(h(\eta_3)-h(\eta_4))
     \right]=\\
  &&\quad = \sum_{e\in C_{\eta_1\to \eta_2}} \sum_{e'\in C_{\eta_3\to \eta_4}}\sigma_e
     \sigma_{e'} \bar A_{r(e),r(e')}(x(e),x(e')) \\
  &&\quad  +\, O\left(\frac1{\min_{i\ne j\le 4}|\eta_i-\eta_j|^{1/2}+1}\right), \end{eqnarray}
where $r(e)$ is the type of the edge $e$, $x(e)$ is the coordinate of the black site of $e$. By using 
the explicit expression of $\bar A_{r,r'}$, \eqref{eq:AAl}, and by decomposing the two paths $C_{\eta_1\to \eta_2}, C_{\eta_3\to \eta_4}$, into a sequence of elementary steps, we obtain 
(denoting the generic elementary step in $C_{\eta_1\to \eta_2}$, resp. $C_{\eta_3\to \eta_4}$, by $s(x,j)$, resp. $s(x',j')$)
\begin{eqnarray} 
  \eqref{covariance_2}&=& \frac1{4\pi^2}\sum_{\o=\pm}\sum_{\substack{s(x,j)\in C_{\eta_1\to \eta_2}\\ s(x',j')\in C_{\eta_3\to \eta_4}}} \sum_{\substack{e\in s(x,j)\\ e'\in s(x',j')}} {\sigma_e\sigma_{e'}}
  \frac{\bar K_{\o,r(e)} \bar K_{\o,r(e')} }{(\bar \phi_\o(x-x'))^2}\label{cov_3}\\
                      &&+\, O\left(\frac1{\min_{i\ne j\le 4}|\eta_i-\eta_j|^{1/2}+1}\right).\nonumber\end{eqnarray}
{                    We now use  \eqref{eq:bbhh.1}-\eqref{eq:bbhh.2} and the symmetry $\bar\phi_\o=\bar\phi^*_{-\o}$ to rewrite the dominant term in \eqref{cov_3} as the Riemann sum approximating the following integral:}
\begin{equation}
\label{eq:35tetris}
 - \frac{\nu}{2\pi^2}\Re\int_{\bar \phi_+(\eta_1)}^{\bar \phi_+(\eta_2)}dz\int_{\bar \phi_+(\eta_3)}^{\bar \phi_+(\eta_4)}dz' \frac {1
    }{(z-z')^2}
\end{equation}
whose explicit evaluation gives the main term in the r.h.s. of
\eqref{eq:35l}. Putting together the error terms, we obtain the
statement of Theorem \ref{th:2}, as desired.

In the case where \eqref{eq:matassa} fails (e.g.  when
$\eta_1=\eta_3$, $\eta_2=\eta_4$ and \eqref{covariance} is just the
variance of the height gradient), one chooses the paths
$C_{\eta_1\to\eta_2},C_{\eta_3\to\eta_4}$ to be ``as well separated as
possible'' (cf. \cite[Sec. 3.2]{GMT17a}) and the rest of the argument
works the same.

\section{Renormalization Group analysis} 
\label{sec:RG}

In this section we discuss the multiscale analysis of the dimer model and the
comparison with the continuum model, which leads us to the
results spelled out  in Proposition \ref{prop:comp}.

We first make a few preliminary manipulations on the Grassmann integral, in order to set it up in a form appropriate for 
multiscale integration (Section \ref{sec:preliminari}). In the following sections, Section \ref{sec:illavorovero} to \ref{sec:asympt}, we first give the iterative definition of 
the effective potentials, then explain how to bound the norm of their kernels, how to fix the dressed velocities and Fermi points in order to make
the expansion convergent uniformly in the thermodynamic limit and, finally, how to adapt the bounds to the computation of the correlation functions. 
For a more detailed guidance on Sections \ref{sec:illavorovero} to \ref{sec:asympt}, see the end of Section \ref{sec:preliminari}.

From now on, $C,C',\ldots,$ and $c,c',\ldots,$ denote universal
constants, whose specific values might change from line to line.

\subsection{Preliminaries}
\label{sec:preliminari}
As a preliminary step, we rewrite the quadratic part $S$ of the action in \eqref{eq:Wap} as a ``dressed'' term $S_0$ plus a "counter-term" $N=S-S_0$, whose role is 
to fix the location of the interacting Fermi points and Fermi velocities. 
Namely, letting as usual
\begin{eqnarray}
  \label{eq:fourier}  \hat\psi^\pm_k=\sum_{x\in\Lambda}\psi^\pm_xe^{\mp i k x}, \quad k\in \mathcal P(\bt), \qquad \psi^\pm_x=\frac1{L^2}\sum_{k\in \mathcal P(\bt)}e^{\pm ikx}\hat\psi^\pm _k,
\end{eqnarray}
we write:
\begin{equation}
S(\psi)=-L^{-2}\sum_{k\in \mathcal P(\bt)} \mu (k)\hat \psi^+_k\hat\psi^-_k\equiv S_0(\psi)+N(\psi),  \end{equation}
where $S_0(\psi)=-L^{-2}\sum_{k\in \mathcal P(\bt)} \mu_0 (k)\hat \psi^+_k\hat\psi^-_k$, with 
\begin{equation}\mu_0(k)=\mu(k)+\sum_{\o=\pm}\bar\chi_0(k-\bar p^\o)\left[-\mu(\bar p^\o)+a_\o(k_1-\bar p^\o_1)+b_\o(k_2-\bar p^\o_2)\right].\label{mu0}\end{equation}
In this equation: 
\begin{enumerate}
\item $\bar p^\o=\bar p^\o(\l)$, with $\o=\pm$, are points in
  $[-\pi,\pi]^2$, such that $\bar p^++\bar p^-=(\pi,\pi)$, and they
  will be fixed via the multiscale construction. A posteriori they can
  be interpreted as ``dressed Fermi points''; they are the same
  functions appearing in Theorem \ref{th:1}.
\item $a_\o=a_\o(\l)\in \mathbb C$ and $b_\o=b_\o(\l)\in\mathbb C$ are such that $a_\o=-a^*_{-\o}$ and
  $b_\o=-b^*_{-\o}$; they will also be fixed via the multiscale
  construction. A posteriori, their choice fixes the ``dressed Fermi
  velocities'' via the following relations:
  \begin{eqnarray}
 \partial_{p_1}\mu_0(\bar p^\o) =\partial_{p_1}\mu(\bar p^\o)+a_\o=:\bar\alpha_\o,\label{eq:6.4}\\
    \partial_{p_2}\mu_0(\bar p^\o)=\partial_{p_2}\mu(\bar p^\o)+b_\o=:\bar\beta_\o,\label{eq:6.5}
  \end{eqnarray}
where $\bar\a_\o,\bar\b_\o$ are the same functions appearing in Theorem \ref{th:1}.
\item the function $\bar\chi_0$ is defined as: 
  $\bar\chi_0(k')=\bar\chi(|\mathcal M^{-1}k'|)$, where: (1)
  $\mathcal M$ is the same matrix as \eqref{eq:M}, with $\bar\a^{1}$
  and $\bar\a^2$ (resp. $\bar\b^{1}$ and $\bar\b^2$) the real and
  imaginary parts of $\bar \a_+$ (resp. $\bar\b_+$); (2) $\bar \chi:\mathbb R^+\to [0,1]$ is a $C^\infty$ cut-off
  function in the Gevrey class of order $2$ (see \cite[Appendix
  C]{GMT17a}) that is equal to 1, if its argument is smaller than
  $c_0/2$, and equal to 0, if its argument is larger than $c_0$; here $c_0$ is a small enough constant, such that in particular the support of
  $\bar\chi_0(\cdot-\bar p^+)$ is disjoint from the support of  $\bar\chi_0(\cdot-\bar p^-)$. For
  later reference, we also let for $h$ a negative integer
  \begin{eqnarray}
    \label{eq:chih}
  \bar \chi_h(k'):=\bar\chi_0(2^{-h}k')  .
  \end{eqnarray}
\end{enumerate}

From the properties just stated of $\bar p^\o,a_\o,b_\o$ and
$\bar\chi(\cdot)$, we see that
\begin{eqnarray}
  \label{eq:simmmu0}
  \mu_0((\pi,\pi)-k)=\mu_0^*(k).
\end{eqnarray}

In the integration over $\psi$ in \eqref{eq:Wap}, the Fourier modes
$k$ that are the closest to the zeros of $\mu_0(\cdot)$ play a
somewhat special role, so they have to be treated separately, at the
very last step of the multi-scale procedure (see Section
\ref{sec:ultimiquattro}).  Namely, given $\bt\in \{0,1\}^2$, let
$k^\pm_{\bt}$ be the values of $k\in \mathcal P(\bt)$ that are closest
to $\bar p^\pm$ and note that $k^+_\bt=(\pi,\pi)-k^-_\bt$ [If there is
more than one momentum at minimal distance from $\bar p^\pm$ (there
are at most four), any arbitrary choice will work]. Next, we decompose
the quadratic action $S_0(\psi)$ as a sum of a term depending only on
$k^\pm_\bt$ plus a term depending only on the modes in
\[\mathcal P'(\bt):=\mathcal P(\bt)\setminus \{k^+_\bt,k^-_\bt\},\]
and we rewrite \eqref{eq:Wap} as
\bea
  e^{\mathcal W_L^{(\bt)}(A,\phi)}&=&
  \int D\psi\,e^{-L^{-2}\sum_{\o=\pm}\mu_0(k_\bt^\o)\hat\psi^+_{k_\bt^\o}\hat\psi^-_{k_\bt^\o}}\\
  &\times&  e^{-L^{-2}\sum_{k\in \mathcal P'(\bt)}\mu_0(k)\hat\psi^+_{k}\psi^-_{k}+N(\psi)+V(\psi,A)+(\psi,\phi)}.\nonumber\eea
We  multiply and divide by 
\begin{eqnarray}
  e^{L^2 E^{(0)}}:=\prod_{k\in \mathcal P'(\bt)}
  \mu_0(k),\end{eqnarray}
(the product is non-zero since we singled out the possibly zero modes $k^\pm_{\bt}$) and, letting
\begin{eqnarray}
  \label{eq:apc}
\hat \Psi_\o^\pm:=\hat \psi^\pm_{k_\bt^\o},  
\end{eqnarray}
we rewrite the generating function as 
\begin{equation}
\label{eq:Wanophi}
e^{\mathcal W_L^{(\bt)}(A,\phi)}=\int  D\hat \Psi e^{-L^{-2}\sum_{\o=\pm}\mu_0(k_\bt^\o)\hat\Psi^+_{\o}\hat\Psi^-_ {\o}+\mathbb W_L^{(\bt)}(A,\phi,\Psi)}.\end{equation}
Here
\begin{eqnarray}
\Psi^\pm_x=\frac1{L^2}\sum_{\o=\pm}e^{\pm i k_\bt^\o x}\hat\Psi^\pm_{\o},  \quad
 \int  D\hat \Psi\prod_{\o=\pm} \hat\Psi^-_{\o}\hat \Psi^+_{\o}= L^4,\label{crollo}
\end{eqnarray}
(the $L^4$ factor comes from the fact that \eqref{eq:Dpsi} translates in Fourier space into $
  \int D\psi \prod_{k\in \mathcal P(\bt)}[L^{-2}{\hat\psi^-_k\hat\psi^+_k}]=1$)
 and
\begin{equation}
e^{\mathbb W_L^{(\bt)}(A,\phi,\Psi)}:=e^{L^2 E^{(0)}} \int P_{g_{0}}(d\psi) e^{N(\Psi+\psi)+V(\Psi+\psi,A)+(\Psi+\psi,\phi)},\label{mathbbW}
\end{equation}
with  $P_{g_{0}}$ the Grassmann Gaussian integration (see footnote \ref{foto:Gint}) with propagator 
\begin{equation} g_{0}(x,y)=L^{-2}\sum_{k\in \mathcal P'(\bt)}
  \frac{e^{-ik(x-y)}}{\mu_0(k)}.\label{eq:g00}\end{equation} From this
point, we proceed as follows. First, we perform in a multi-scale way
the integration over the Grassmann variables $\psi$, i.e. over the
Fourier modes except $k^\pm_{\bt}$: the inductive integration
procedure, including the definition of the {\it running coupling
  constants (RCC)}, is described in Section \ref{sec:illavorovero};
the outcome of the construction can be conveniently expressed in terms
of a Gallavotti-Nicol\`o tree expansion, similar to the one described
in \cite[Section 6.2]{GMT17a}. The main definitions (and the main
differences compared to the case treated in \cite{GMT17a}) are
summarized in Section \ref{sectree}; in the same section, we also
state the bounds satisfied by the kernels of the effective potential,
see Proposition \ref{prop:an}, {\it under the assumption that the RCC
  are uniformly bounded in the infrared}, see condition
\eqref{ggg0}. The proof that the RCC remain in fact bounded under the
iterations of the renormalization group map is given in Section
\ref{secfl}; the flow of the RCC can be controlled only if their
initial data are properly fixed: as shown there, the choice of the
initial data fixes the dressed Fermi points $\bar p^\o$ and the
dressed Fermi velocities $\bar\a_\o$, $\bar\b_\o$, as anticipated
after \eqref{mu0}. In Section \ref{sec:ultimiquattro} we describe the
integration of the last two modes and prove the existence of the
thermodynamic limit for the correlation functions, with explicit
bounds on the speed of convergence as $L\to\infty$. Finally, in
Section \ref{sec:asympt}, we compute the fine asymptotics of the
correlations functions, via a comparison of the tree expansion of the
dimer model with that of the continuum model of Section \ref{sec:IR}, and complete the proof of Proposition 
\ref{prop:comp}.

\subsection{Multi-scale analysis}
\label{sec:illavorovero}
 In this section we describe the multi-scale computation of $\mathbb W_L^{(\bt)}(A,\phi, \Psi)$ {defined in} \eqref{mathbbW}. 
 We consider explicitly only the case $\phi=0$; the general case can be treated  analogously but we will not belabor the details in this paper. 
  
The procedure is based on a systematic use of the `addition principle' for Gaussian Grassmann integrals,
namely the following property \cite[Sec. 4]{GMreview}:
if $P_{g}(d\psi)$ is the Grassmann Gaussian integration with propagator $g$ and $g=g_1+g_2$ then 
\begin{equation}
  \label{eq:addpr}
\int P_{g}(d\psi) F(\psi)=\int P_{g_1}(d\psi_1)  P_{g_2}(d\psi_2) F(\psi_1+\psi_2).
\end{equation}
We apply this formula to $P_{g_{0}}$, in connection with the following decomposition of the propagator $g_{0}(x,y)$:
\begin{equation}
g_{0}(x,y)=g^{(0)}(x,y)+\sum_{\o=\pm} e^{-i\bar p^\o(x-y)}g_{\o}^{(\le -1)}(x,y)\label{decompo}
\end{equation}
where %letting $\bar\chi_h(k')=\bar\chi(2^{-h}|\mathcal M^{-1}k'|)$, 
\begin{equation}
  \label{eq:gpar0par}
g^{(0)}(x,y)= L^{-2}\sum_{k\in \mathcal P'(\bt)}e^{-i k  (x-y)}\frac{1-\bar\chi_{-1}(k-\bar p^+)-\bar\chi_{-1}(k-\bar p^-)}
{\mu_0(k)}
\end{equation}
and, if $\mathcal P_\o'(\bt)=\{k': k'+\bar p^\o\in \mathcal P'(\bt)\}$,
\begin{equation}
g^{(\le -1)}_{\o}(x,y)=L^{-2}\sum_{k'\in \mathcal P_\o'(\bt)}e^{-ik'(x-y)}\frac{\bar\chi_{-1}(k')}{\mu_0(k'+\bar p^\o)}.\label{eq:6.11}
\end{equation}
% The reason why we extracted the factors $e^{-i\bar p^\o(x-y)}$ in \eqref{decompo} is that, this way, $g_\o^{(\le -1)}$ has a power-law decay for $x-y$ large, without oscillating prefactor.
By using the decomposition \eqref{decompo} and \eqref{eq:addpr}, we rewrite \eqref{mathbbW} as
\begin{eqnarray}
  && e^{\mathbb
     W_L^{(\bt)}(A,0,\Psi)}=%e^{L^2 E^{(0)}}\int P_{g_{0}}(d\psi) e^{N(\psi)+V(\psi,A)}=
     e^{L^2 E^{(0)}}\int P_{{(\le -1)}}(d\psi^{(\le -1)})\times\label{6.13}\\
  &&\hskip1.5truecm\times\int P_{(0)}(d\psi^{(0)})
     e^{N(\psi^{(0)}+\psi^{(\le
     -1)}+\Psi)+V(\psi^{(0)}+\psi^{(\le
     -1)}+\Psi,A)},\nonumber\end{eqnarray} where
$\psi^{(0)}+\psi^{(\le -1)}+\Psi$ is a shorthand
notation for
\begin{multline}
  \label{eq:sh}
  \{\psi^{(0)\pm}_x+\sum_{\o}e^{\pm i\bar p^\o
  x}\varphi^\pm_{x,\o}\}_{x\in\L}, \quad \varphi^\pm_{x,\o}:=\psi^{(\le -1)\pm}_{x,\o}+L^{-2}e^{\pm i (k^\o_\bt-\bar p^\o)x}\hat \Psi^\pm_\o. 
  \end{multline}
$P_{(0)}$ is the Grassmann Gaussian measure with
propagator
$g^{(0)}(x,y)$, while $P_{(\le -1)}$ is the Grassmann Gaussian
measure with propagator
$$\d_{\o,\o'}g^{(\le -1)}_{\o}(x,y)=\int P_{(\le -1)}(d\psi)
\psi^{(\le -1)-}_{x,\o}\psi^{(\le -1)+}_{y,\o'}.$$

Since the cutoff function $\bar\chi_{-1}$ in \eqref{eq:6.11} is a Gevrey
function of order $2$, the propagator $g^{(0)}$ has
stretched-exponential decay at large distances:
\be \label{eq:bbpprr}|g^{(0)}(x,y)|\le C e^{-\kappa \sqrt{|x-y|}},
\ee for suitable $L$-independent constants $C,\kappa>0$, if $|x-y|$ is
the distance on the torus $\L$.  This is seen by writing $g^{(0)}$
via the Poisson summation formula as a sum of Fourier integrals, as in
\cite[App. A]{GMT17a}; each integral decays in the desired way because
it is the Fourier transform of a Gevrey function \cite{Rodino}.

Next, we denote by $V^{(0)}(\cdot,J)$ the combination
$N(\cdot)+V(\cdot,A)$, re-expressed in terms of the variables
$J=\{J_{x,r}\}_{x\in\L,\, 1\le r\le 4}$, instead of $A$: here, if $b$
is the bond of type $r$ and black site $x$, we let
$J_{x,r}:=e^{A_b}-1$. The result of the integration over $\psi^{(0)}$
is rewritten in exponential form:
\begin{equation}e^{L^2 E^{(0)}}\int P_{(0)}(d\psi^{(0)}) e^{V^{(0)}(\psi^{(0)}+\f,J)}=
e^{L^2 E^{(-1)} + S^{(-1)}(J)+V^{(-1)}(\f,J)},\label{eq:tb0}\end{equation}
where \cite[Sec. 4]{GMreview} % letting $V^{(0)}(\psi,A)=N(\psi)+V(\psi,A)$, 
\begin{eqnarray} && L^2 (E^{(-1)}-E^{(0)}) +S^{(-1)}(J)+ V^{(-1)}(\f,J)=\label{eq:troppobella}\\
&&\qquad =\sum_{n\ge 1}\frac1{n!}\mathcal E^T_0(\underbrace{V^{(0)}(\psi^{(0)}+\f,J);\cdots;
V^{(0)}(\psi^{(0)}+\f,J))}_{n\ {\rm times}},\nonumber
\end{eqnarray}
with $\mathcal E^T_0$ the truncated expectation\footnote{in other
  words,
  $\mathcal E^T_0(\underbrace{V^{(0)};\cdots; V^{(0)}}_{n\ {\rm
      times}})$ is the $n-th $ cumulant of $V^{(0)}$ w.r.t. the
  Grassmann Gaussian integration $P_{(0)}$. See \cite[Sec. 4 and
  App.
  A.3]{GMreview}} % \footnote{see \cite[App. A.3]{GMreview} for the expression of the truncated expectation of $n$ Grassmann monomials w.r.t. a Grassmann Gaussian integration $P_g(d\psi)$ in terms of the propagator $g$.
% }
w.r.t. the Grassmann Gaussian integration $P_{(0)}(d\psi^{(0)})$, and
$E^{(-1)}, S^{(-1)}(\cdot)$ are fixed by the condition
$S^{(-1)}(0)=0$, $V^{(-1)}(0,J)=0$.  The series in the r.h.s. is
absolutely summable, for $\l$ sufficiently small (independently of $L$),
see \cite[Sec. 4.2]{GMreview}.  The reason is that the propagator
$g^{(0)}$ has a fast decay in space, uniformly in $L$, as in \eqref{eq:bbpprr}.

The effective
potential on scale $-1$ can be represented as in the following formula
(which is a {\it definition} of the kernels
$W^{(-1)}_{n,m;\ul\o,\ul r}$):
\begin{eqnarray}
  && V^{(-1)}(\f,J)= \sum_{\substack{n,m\ge 0:\\ n\ {\rm even}, \ n\ge 2}}\sum_{\ul x,\, \ul y,\, \ul \o,\,\ul r}W^{(-1)}_{n,m;\ul\o,\ul r}(\ul x, \ul y)\nonumber\\\label{eq:6.16}
  &&\qquad \times\, \f^+_{x_1,\o_1}\f^-_{x_2,\o_2}\cdots\f^+_{x_{n-1},\o_{n-1}}\f^-_{x_n,\o_n} J_{y_1,r_1}\cdots J_{y_m,r_m},\end{eqnarray}
where: $\ul x=(x_1,\ldots, x_n)\in \Lambda^n,\ul y\in \Lambda^m,\ul\o\in\{-1,+1\}^n,\ul r\in\{1,\dots,4\}^m$; the Grassmann variables $\f^\pm_{x,\o}$ were defined in \eqref{eq:sh}. % stands for the Grassmann field  
% $\psi^{(\le -1)\pm}_{x,\o}+\Psi^\pm_{x,\o}$, where $\Psi^\pm_{x,\o}=L^{-2}\hat \Psi^\pm_{k^\o_\bt}e^{\pm i(k^\o_\bt-\bar p^\o)x}$
Moreover, the kernels  can be
written as
\begin{eqnarray}
  \label{eq:rk}
 W^{(-1)}_{n,m;\ul\o,\ul r}(\ul x, \ul y)=\tilde W^{(-1)}_{n,m;\ul r}(\ul x, \ul y)\,\exp\{i\sum_{j=1}^n(-1)^{j-1}\bar p^{\o_j}x_j\}, 
\end{eqnarray}
  with  $\tilde W^{(-1)}_{n,m;\ul r}(\ul x, \ul y)$ a function that is\footnote{
    The properties of $\tilde W^{(-1)}_{n,m;\ul r}$ listed here are a consequence of the translation invariance of the Hamiltonian of the model and of 
    Eq.\eqref{eq:identif}. The exponential oscillating factor in \eqref{eq:rk} has its origin in \eqref{eq:sh}.}
 independent of $\ul\o$, 
% (ii)
 translationally invariant, % (iii)
 periodic of period $L$ in $y_i$, and $\bt$-periodic of period $L$ in
 $x_i$ (here we say that, e.g., a function is $(0,1)$-periodic if it
 is periodic in the first coordinate and anti-periodic in the second,
 and similarly for the other cases).  The kernels $W^{(-1)}$ are not uniquely defined by \eqref{eq:6.16}; due to the anti-commutation of
 Grassmann variables and to the fact that $J_{y,r}$ are ordinary
 commuting variables, one can assume  that
  $W^{(-1)}_{n,m;\ul \o,\ul r}(\ul x, \ul y)$ are symmetric
 under permutations of the indices $(y_1,r_1), \ldots, (y_m,r_m)$, and
 anti-symmetric under permutations of the indices
 $\{(x_{2i},\o_{2i})\}_{1\le i\le n/2}$ and of the indices
 $\{(x_{2i-1},\o_{2i-1})\}_{1\le i\le n/2}$. An analogous
 representation is valid for $S^{(-1)}(\cdot)$, and we denote its
 kernels by $W^{(-1)}_{0,m;\ul\o,\ul r}(\ul y)$.

There is an equivalent expression for $V^{(-1)}$ in
Fourier space. We use the following convention for the Fourier
transforms of the fields $\psi,J$:
$$\varphi^\pm_{x,\o}=L^{-2}\sum_{k\in \mathcal P_\o(\bt)}e^{\pm i k\cdot x}\hat \varphi^\pm_{k,\o},\qquad J_{x,r}=L^{-2}\sum_{p\in \mathcal P(\V0)} \hat J_{p,r} e^{-ipx},$$ where $\mathcal P_\o(\bt):=\{k: k+\bar p^\o\in\mathcal P(\bt)\}$.
The reason why $k\in \mathcal P_\o(\bt)$ (and not in $\mathcal P(\bt)$ as in
\eqref{eq:fourier}) is that the combination
$e^{\pm i \bar p^\o x}\varphi^\pm_{x,\o}$ is $\bt$-periodic, and not
$\varphi^\pm_{x,\o}$ itself.  This sum includes also  the momenta
$k=k^\o_\bt-\bar p^\o$, $\o=\pm$. Of course,
recalling that the only non-zero modes of $\psi_{x,\o}^\pm$
(resp. $\Psi^\pm_{x,\o}$) are in $\mathcal P'_\o(\bt)$
(resp. $\{k^\o_\bt-\bar p^\o\}_{\o=\pm}$), we have that
$$\hat \f_{k,\o}^\pm=\begin{cases} \hat \psi^\pm_{k,\o}, \ {\rm if}\ k\in\mathcal P'_\o(\bt),\\
  \hat \Psi^\pm_{\o}, \ {\rm if}\ k\not\in \mathcal
  P'_\o(\bt). \end{cases}$$

Then, \eqref{eq:6.16} becomes
\begin{eqnarray}
  && V^{(-1)}(\varphi,J)= \hskip-.2truecm\sum_{\substack{n,m\ge 0:\\ n\ {\rm even}, \ n\ge 2}}\hskip-.1truecm L^{-2(n+m)}\hskip-.1truecm\sum_{\ul k,\, \ul p,\, \ul \o,\,\ul r} 
  \hskip-.2truecm\hat W^{(-1)}_{n,m;\ul\o,\ul r}(k_2,\ldots,k_n,p_1,\ldots,p_m)\times \nonumber\\
  &&\qquad \times\, \hat\f^+_{k_1,\o_1}\hat\f^-_{k_2,\o_2}\cdots\hat\f^+_{k_{n-1},\o_{n-1}}\hat\f^-_{k_n,\o_n} \hat J_{p_1,r_1}\cdots \hat J_{p_m,r_m}\d_{\ul \o}(\ul k, \ul p),\label{eq:6.16mo}\end{eqnarray}
where $\underline k=(k_1,\ldots, k_n)$, with $k_i\in \mathcal P_{\o_i}(\bt)$, $\ul p=(p_1,\ldots p_m)\in \big[\mathcal P(\V0)\big]^m$ and 
\be\d_{\ul\o}(\ul k,\ul p)=L^2 \times\left\{
  \begin{array}{ccc}
    1 & \text{if} & \sum_{j=1}^n(-1)^{j-1}(k_j+\bar p^{\o_j})=\sum_{j=1}^m p_j\!\!\!\mod(2\pi,2\pi)\\
    0  & \text{else} & 
  \end{array}\right.
\label{eq:delta}\ee
is the periodized Kronecker delta enforcing momentum
conservation. Also,  $\hat W^{(-1)}_{n,m;\ul\o,\ul r}(k_2,\ldots,k_n,p_1,\ldots,p_m)$ is just the Fourier transform of $\tilde W^{(-1)}_{n,m;\ul r}$, computed at momenta $k_2+\bar p^{\o_2},\ldots, k_n+\bar p^{\o_n},p_1,\ldots,p_m$
% \[
% \hat W_{n,m;\ul\o,\ul r}(k_2,\ldots,k_n,p_1,\ldots,p_m):=\hat W_{n,m;\ul r}(k_2+\bar p^{\o_2},\ldots, k_n+\bar p^{\o_n},p_1,\ldots,p_m).
%   \]
(it  depends only on $n+m-1$ momenta, due to translation
invariance of $\tilde W^{(-1)}_{n,m;\ul r}$  in real space).

Using the Battle-Brydges-Federbush-Kennedy (BBFK) determinant formula and the Gram-Hadamard bound \cite[Sec. 4.2]{GMreview} for the truncated expectation in \eqref{eq:troppobella}, we find that 
$E^{(-1)}$, and 
$W^{(-1)}_{n,m;\ul\o,\ul r}(\ul x, \ul y)$ are absolutely convergent series and real analytic functions of \be \label{rcc00}(\n_{0,\o},a_{0,\o},b_{0,\o},\l_0),\ee
for $\max\{|\n_{0,\o}|,|a_{0,\o}|,|b_{0,\o}|,|\l_0|\}\le \e$ with $\e$ sufficiently small, 
where we denoted (for uniformity of notation with the running coupling constants 
$\n_{h,\o},a_{h,\o},b_{h,\o},\l_h$, to be introduced below):
\be\label{rcc0}\n_{0,\o}:=-\m(\bar p^\o),\quad a_{0,\o}:=a_\o,\quad b_{0,\o}:=b_\o,\quad \l_0:=\l.\ee
Moreover, $|E^{(-1)}|\le C\e$ and, using also the exponential decay of the bare potential, \eqref{eq:ub}, we find that
\begin{equation} \|W^{(-1)}_{n,m}\|_{\kappa,-1}\le C^{n+m}\e^{\max\{1,cn\}},\label{boundkappanorm}\end{equation}
for suitable constants $\kappa,C,c>0$ independent of the system size. Here
\begin{equation} \|W^{(-1)}_{n,m}\|_{\kappa,-1}:=L^{-2}\sup_{\ul \o,
    \ul r}\sum_{\ul x, \ul y}|W^{(-1)}_{n,m;\ul\o,\ul r}(\ul x, \ul
  y)|e^{\kappa \sqrt{2^{-1}d(\ul x,\ul
      y)}}, \label{kappanorm}\end{equation} and $d(x_1,\ldots,x_l)$ is
the length of the shortest tree on the torus $\L$ connecting the $l$
points in $(x_1,\ldots,x_l)$.
% (the term with $n=0$ in the right side of \eqref{dist} should be interpreted as being equal to $\tilde d(x_1,\ldots,x_l)$). 
The choice of the stretched-exponential weight in \eqref{kappanorm} is related to the stretched-exponential decay of the propagator, see \eqref{eq:bbpprr}. 
For technical details about the proof \eqref{kappanorm}, or, better, of its analogue in a similar context, the reader can consult, e.g., \cite[Section III.A and Eq. (3.19)]{GGM}.

\begin{Remark}\label{rem.L} The fact that the kernels $W^{(-1)}_{n,m;\ul\o,\ul r}$ are absolutely convergent series of $(\n_{0,\o},a_{0,\o},b_{0,\o},\l_0)$, 
that each term in the expansion admits a limit as $L\to\infty$ (as one can check by inspection) and that they satisfy uniform bounds as $L\to\infty$, see \eqref{boundkappanorm}, 
implies that their infinite volume limits exist and satisfy the same bounds. For later reference, the infinite volume limit of $W^{(-1)}_{n,m;\ul\o,\ul r}$ will be denoted 
by $W^{(-1),\infty}_{n,m;\ul\o,\ul r}$, and similarly for its Fourier transform. 
\end{Remark}

\medskip

After this first integration step, we still need to integrate
$\psi^{(\le -1)}$ out, see \eqref{6.13}. Let us first informally
explain how this is done, before giving the precise inductive
procedure in Sections \ref{sec:tih}--\ref{sec:tis}.  The idea is to
repeat the same procedure as above: we rewrite (via the addition
principle)
$\psi^{(\leq -1)}_{\o} = \psi^{(-1)}_\o+\psi^{(\leq -2)}_{\o}$, where
$\psi^{(-1)}_\o$ (resp. $\psi^{(\le -2)}_\o$) is a Grassmann field
with propagator supported, in momentum space, on momenta
$k'\in \mathcal P'_\o(\bt)$ with $|k'|\sim 2^{-1}$
(resp. $|k'|\lesssim 2^{-2}$); we integrate $\psi^{(-1)}_\o$ out; we
exponentiate the result of the integration, thus defining the
effective potential on scale $-2$, in analogy with
\eqref{eq:tb0}-\eqref{eq:troppobella}; and so on.  One after the
other, we integrate the fields $\psi^{(-2)}, \ldots, \psi^{(h+1)}$
out, define the effective potential $V^{(h)}$ on scale $h$ (which
involves fields $\psi^{(\le h)}$ with momenta
$k'\in \mathcal P'_\o(\bt)$ that belong to the support of
$\bar\chi_{h}(\cdot)$ (cf. \eqref{eq:chih}), and continue until we
reach the `last scale', $h_L$, fixed by the finite volume $L$, which
induces a natural infrared cut-off. More precisely, $h_L$ is fixed as
the smallest (in absolute value) negative integer $h$ such that the
support of $\bar \chi_{h}(\cdot)$ has empty intersection with
$\mathcal P'_\o(\bt)$.  Note that, since all momenta in
$\mathcal P'_\o(\bt)$ are at distance at least $\pi/L$ from
$\bar p^\o$, we have $h_L\sim -\log_2L$ for $L$ large.
% therefore, the propagators at scales smaller than
% $h_L:=\lfloor \log_2(\pi/(c_0L))\rfloor$, with $c_0$ the same as in
% item (3) of the list after \eqref{mu0}, are identically
% zero.
  The result of the
integration of the Grassmann fields $\psi^{(\le h_L)}$ gives the
generating function $\mathbb W_L^{(\bt)}(A,0,\Psi)$, as desired.

\medskip

In order for the bounds on the generating function to be uniform in
$L$, we need to improve the procedure roughly described here: at each
step, before integrating the field on the next scale, we actually need
to isolate and re-sum a certain selection of potentially dangerous
contributions to the effective potential, the so-called marginal and
relevant terms. We refer, e.g., to \cite[Sec. 5]{GMT17a}, see in
particular \cite[Section 5.2.2]{GMT17a} for a dimensional
classification of the divergent terms arising in a `naive' multiscale
scheme. As discussed there, see \cite[Eq. (5.8)]{GMT17a} and following
lines, the scaling dimension of the kernels with $n$ external fields
of type $\psi$ and $m$ external fields of type $J$ is $2-n/2-m$; in
the renormalization group jargon, positive scaling dimension (that is,
$2-n/2-m>0$ $\Leftrightarrow$ $(n,m)=(2,0)$) corresponds to {\it
  relevant} contributions, vanishing scaling dimension (that is,
$2-n/2-m=0$ $\Leftrightarrow$ $(n,m)=(4,0),(2,1)$) corresponds to {\it
  marginal} contributions, and negative scaling dimension corresponds
to {\it irrelevant} ones. In order to cure the potential divergences
associated with the terms with $(n,m)=(2,0),(4,0),(2,1)$, at each step
of the multiscale construction we properly `localize' and re-sum these
terms, via an iterative procedure that we now describe.

\subsubsection{The inductive statement}
\label{sec:tih}
Let us inductively assume that the fields
$\psi^{(0)}, \psi^{(-1)},\ldots,\psi^{(h+1)}$, $h\ge h_L$, have been
integrated out, and that after their integration the generating
function has the following structure, analogous to the one at scales
$0,-1$:
\bea && e^{-L^{-2}\sum_\o\mu_0(k^\o_\bt)\hat\Psi^+_{\o}\hat\Psi^-_{\o}+\mathbb W_L^{(\bt)}(A,0,\Psi)}=e^{L^2 E^{(h)}+S^{(h)}(J)}
\times\label{eq:vh}\\
&&\qquad \times e^{-L^{-2}Z_h\sum_\o\mu_{h,\o}(k^\o_\bt-\bar p^\o)\hat\Psi^+_{\o}\hat\Psi^-_{\o}}\int
P_{(\le h)}(d\psi)e^{V^{(h)}(\sqrt{Z_h}(\psi+\Psi),J)},\nonumber\eea 
for suitable real constants $E^{(h)}$, $Z_h$, and suitable `effective potentials' $S^{(h)}(J)$, $V^{(h)}(\f,J)$, to be defined inductively below, and fixed in such a way that 
$V^{(h)}(0,J)=S^{(h)}(0)=0$. In the second line, 
$$\m_{h,\o}(k):=\bar D_\o(k)+r_\o(k)/{Z_h},$$
where 
$$\bar D_\o(k)=\bar \a_\o k_1+\bar\b_\o k_2$$ and 
\begin{equation}r_\o(k)=\m(k+\bar p^\o)-\mu(\bar
  p^\o)-\partial_{k_1}\m(\bar p^\o)k_1-\partial_{k_2}\m(\bar
  p^\o)k_2\label{eq:rem}\end{equation} is a remainder of order
$O(k^2)$ for $k$ small. Finally, $P_{(\le h)}(d\psi)$ is the Grassmann
Gaussian integration with propagator (diagonal in the index $\o$) \be
\frac1{Z_h}g^{(\le
  h)}_{\o}(x,y)=\frac1{Z_h}\frac1{L^2}\sum_{k\in\mathcal
  P_\o'(\bt)}e^{-ik(x-y)}\frac{\bar\chi_{h}(k)}{\m_{h,\o}(k)}.\label{eq:gleh}\ee
%For later reference, let us denote by $g^{(\le h),1}_{\o}(x,y)$ the dominant part of $g^{(\le h)}_{\o}(x,y)$, that is the one obtained by replacing $r_\o$ with zero:
%\be \frac1{Z_h}g^{(\le h),1}_{\o}(x,y)=\frac1{Z_h}\frac1{L^2}\sum_{k'\in D_\o(\bt)}e^{-ik'(x-y)}\frac{\bar\chi_{h}(k')}{\bar \a_\o k_1'+\bar\b_\o k_2'}.\ee
 We will also prove
inductively that:
\begin{enumerate}
\item 
 $V^{(h)}(\f,J)$ has the same structure as 
 \eqref{eq:6.16mo}, with the upper index $(-1)$ in the kernels replaced by $(h)$;
 
\item the kernels of   $V^{(h)}(\f,J)$ satisfy the following symmetry:
  \begin{eqnarray}
    \label{eq:upps}
   \hat  W^{(h)}_{n,m;-\underline\o,\underline r}(\ul k,\ul p)=     \Bigl[\hat W^{(h)}_{n,m;\underline\o,\underline r}(-\ul k,-\ul p)\Bigr]^*.
  \end{eqnarray}

\end{enumerate}

\begin{Remark}
  \label{rem:domenicale}
It is important to emphasize right away that we will view the kernels
$W^{(h)}_{n,m;\ul\o,\ul r}$, $h\le-2$, 
% as well as their infinite-volume limits $W^{(h-1),\infty}_{n,m;\ul\o,\ul r}$,
as functions of:
\begin{enumerate}
\item [(i)]a sequence of {\it running coupling constants}
\[\{\l_{h'},\n_{h',\o},a_{h',\o},b_{h',\o},Y_{h',r,(\o,\o')}\}_{h< h'\le -1}.\] 
% For $h=0$, these were defined in \eqref{rcc0} and by convention we set $Y_{0,r,(\o,\o')}\equiv 0$. 
  
\item [(ii)] a sequence of single-scale
  propagators $\{g^{(h')}_\o/Z_{h'-1}\}_{h< h'\le -1}$, of the form
      \be
  \frac1{Z_{h-1}}g^{(h)}_\o(x,y):=\frac1{L^2}\sum_{k\in \mathcal
    P_\o'(\bt)}e^{-ik(x-y)}\frac{f_{h}(k)}{\tilde Z_{h-1}(k)\bar
    D_\o(k)+r_\o(k)},\label{eq:ghjy6}\ee where
  $f_{h}(k)=\bar \chi_h(k)-\bar\chi_{h-1}(k)$ and
  \[\tilde Z_{h-1}(k)=Z_{h-1}\bar \chi_h(k)+Z_h(1-\bar\chi_{h}(k));\]

\item [(iii)] the {\it irrelevant} part of $V^{(-1)}$, denoted by
  $\mathcal R V^{(-1)}$. 
\end{enumerate}
The actual values of the running coupling constants (RCCs) will be
defined via an inductive procedure in Sections
\ref{sec:h-1}-\ref{sec:tis}, the outcome of which is the beta function equation \eqref{beta2}. 
In Section \ref{secfl}, we will show that there is only one specific choice of the initial data $(\n_{0,\o},a_{0,\o},b_{0,\o})$, 
which we will determine via a fixed point argument, guaranteeing that the flow of RCCs is uniformly bounded for all $h\le 0$. 
For the fixed point argument itself, it is convenient to allow the beta function, as well as the kernels of the effective potential, to be computed 
at values of the RCCs different from the final, `correct', ones. This is what we mean by saying that 
$W^{(h)}_{n,m;\ul\o,\ul r}$ will be thought of as functions of the RCCs: we will allow ourselves to think of the RCCs as independent variables, 
which can be varied freely, as long as they remain sufficiently small; similarly for the dependences on $g^{(h)}_\o/Z_{h-1}$ and $\mathcal R V^{(-1)}$ mentioned in items (ii)-(iii): for certain manipulations discussed below, we will allow ourselves to modify the definition of the kernels by modifying the form of the single-scale propagators or of the kernel of the irrelevant part at scale $-1$, keeping the rest of the iterative 
definition unchanged.

\end{Remark}
\subsubsection{The inductive statement for $h=-1$}
\label{sec:h-1}
The representation \eqref{eq:vh} with \eqref{eq:gleh}-\eqref{eq:rem} is valid at the initial
step, $h=-1$, with $Z_{-1}=1$. To see this,  one
  needs to use that, if $k$ belongs to the support of
  $\bar \chi_{-1}$, then
  $\mu_0(k+\bar p^\o)=\mu(k+\bar p^\o)-\mu(\bar p^\o)+a_\o k_1 +b_\o
  k_2$, see \eqref{mu0}. Moreover, by using
  \eqref{eq:6.4}-\eqref{eq:6.5}, we can also rewrite
  $\mu(k+\bar p^\o)-\mu(\bar p^\o)+a_\o k_1 +b_\o k_2=\bar
  D_\o(k)+r_\o(k)$, which implies that \eqref{eq:gleh}
  at $h=-1$ is the same as \eqref{eq:6.11}.

  To see that \eqref{eq:upps} holds for $h=-1$, note that it is
  equivalent to requiring that $V^{(-1)}$ is invariant under the
  transformation $\varphi^\pm_{\o,x}\to\varphi^\pm_{-\o,x}$ together
  with complex conjugation of the kernels.  On the other hand, by
  Remark \ref{rem:coniugato}, we know that the potential
  $V^{(0)}(\psi,J)$ is invariant under conjugation of the kernels
  together with the transformation
  $\psi^\pm_x=(\psi_x^{(0)\pm}+\sum_\o e^{\pm i \bar
    p^\o}\varphi^\pm_{\o,x})\to(-1)^x\psi^\pm_x$, i.e.,
  $\psi_x^{(0)\pm}\to (-1)^x\psi_x^{(0)\pm}
  ,\varphi^\pm_{\o,x}\to\varphi^\pm_{-\o,x}$. The statement
    \eqref{eq:upps} for $h=-1$ easily follows from  the relation
  \eqref{eq:troppobella} between $V^{(0)}$ and $V^{(-1)}$ together
  with the fact that the propagator $g^{(0)}$ in \eqref{eq:gpar0par}
  satisfies
  \[
    [g^{(0)}(x,y)]^*=(-1)^{x+y} g^{(0)}(x,y),
  \]
    because $\bar p^++\bar p^-=(\pi,\pi)$.

\subsubsection{The inductive step}
\label{sec:tis}
We assume that \eqref{eq:vh} holds with $V^{(h)}$ satisfying the
properties specified in the inductive statement, and we discuss here
how to get the same representation at the next scale $h-1$. First, we
split $V^{(h)}$ into its {\it local} and {\it irrelevant} parts:
$V^{(h)}=\LL V^{(h)}+\RR V^{(h)}$ where, denoting by
$\hat W^{(h),\infty}_{n,m;\ul\o,\ul r}$ the infinite volume limit of
$\hat W^{(h)}_{n,m;\ul\o,\ul r}$,
\begin{eqnarray} 
&& 
 \LL V^{(h)}(\f,J):=\label{eq:6.21}\\
 &&\qquad =L^{-2}\sum_{\o}\sum_{k\in \mathcal P_\o(\bt)} 
 \hat \f^+_{k,\o} [\hat W^{(h),\infty}_{2,0;(\o,\o)}(0)+k\cdot \partial_{k}\hat W^{(h),\infty}_{2,0;(\o,\o)}({0})\big]\hat \f^-_{k,\o}	\nonumber\\
&& \qquad +\sum_{x\in\L}\sum_{\o_1,\ldots,\o_4} \f^+_{x,\o_1}\f^-_{x,\o_2}\f^+_{x,\o_3}\f^-_{x,\o_4} \hat W^{(h),\infty}_{4,0;(\o_1,\ldots,\o_4)}(0,0,0)\nonumber\\
&& \qquad + \sum_{x\in\L}\sum_{\o_1,\o_2,r} J_{x,r}\f^+_{x,\o_1}\f^-_{x,\o_2}e^{i(\bar p^{\o_1}-\bar p^{\o_2})x}\hat W^{(h),\infty}_{2,1;(\o_1,\o_2),r}(0,\bar p^{\o_1}-\bar p^{\o_2}).\nn\eea
   \begin{Remark} \label{remark:10}A few remarks about this definition are in order:
     \begin{enumerate}
     \item The existence of the limit of
       $\hat W^{(h)}_{n,m;\ul\o,\ul r}$ as
       $L\to\infty$ is a corollary of the inductive bounds on the
       kernels of $V^{(h)}$, which are uniform in
       $L$, as it was the case for
       $h=-1$, cf. with Remark \ref{rem.L}. More details on the
       inductive bounds on the kernels of
       $V^{(h)}$ are discussed below.
     \item The reason why, in the second line of \eqref{eq:6.21}, we
       only include terms where the Grassmann fields have the same
       index $\o$, is that the terms with opposite
       $\omega$ indices give zero contribution to the generating
       function, due to the support properties of the Grassmann
       fields. In fact, in \eqref{eq:vh} we need to compute
       $V^{(h)}$ at Grassmann fields $\hat \psi^{(\le
         h)\pm}_{k,\o}$ that, in momentum space, have the same support
       as $\hat g^{(\le h)}_\o(k)$, i.e.,
       $|\mathcal M^{-1}k|\le
       c_02^{h}$ (note that the support properties of
       $\hat g^{(\le h)}_\o$ are the same as those of
       $\bar\chi_h$ (cf. \eqref{eq:chih}), and these were discussed in
       the third item after \eqref{mu0}). If $h\le -1$ and
       $c_0$ is sufficiently small, quadratic terms of the form
       $\hat \psi^{(\le h)+}_{k,\o}\hat \psi^{(\le h)-}_{k+\bar
         p^\o-\bar
         p^{-\o},-\o}$ would involve two fields that cannot both
       satisfy this support property.
\item Due to the Grassmann anti-commmutation rules and the anti-symmetry of the kernels, the quartic term in \eqref{eq:6.21} can be rewritten as 
\begin{equation}4\sum_{x\in\L}\f^+_{x,+}\f^-_{x,+}\f^+_{x,-}\f^-_{x,-} \hat W^{(h),\infty}_{4,0;(+,+,-,-)}(0,0,0).\label{quartic:symm}\end{equation}

     \end{enumerate}
   \end{Remark}

   Along the induction step, we will need % to define  (inductively)
   a function $W^{(h),R}_{2,0;(\o,\o)}(x_1,x_2)$ (the upper index `$R$' stands for ``relativistic'') which should be thought of as
   the kernel for $n=2,m=0$ of a relativistic model.
   More precisely,
at step $h=-1$, one simply lets $W^{(-1),R}_{2,0;(\o,\o)}(x_1,x_2)\equiv 0$.
For $h<-1$, $W^{(h),R}_{2,0;(\o,\o)}$ is defined as a suitable modification of $W^{(h),\infty}_{2,0;(\o,\o)}$ (that, by the induction hypothesis, has already been defined); more precisely,  
$W^{(h),R}_{2,0;(\o,\o)}$ is obtained by making the following replacements in $W^{(h),\infty}_{2,0;(\o,\o)}$ (which should be thought of as a function of the running coupling constants, 
of the single scale propagators and of the irrelevant part of $V^{(-1)}$, as explained in Remark \ref{rem:domenicale}):
   \begin{enumerate}
   \item  [(i)] the running coupling constants
   $\{\n_{h',\omega},a_{h',\o},b_{h',\o}\}_{h'> h}$ are set zero, (note that the
   running coupling constants $\l_{h'}$ are {\it not} set equal to
   zero);
 \item [(ii)] the single-scale propagators $g^{(h')}_\o/Z_{h'-1}$ are replaced by the `relativistic' single-scale propagators 
   $g^{(h')}_{R,\o}/Z_{h'-1}$, for all $h< h'\le -1$, where
   \be g^{(h')}_{R,\o}(x,y)=\int_{\mathbb R^2} \frac{dk}{(2\pi)^2}e^{-ik(x-y)}\frac{f_{h'}(k)}{\bar D_\o(k)};\label{gRo}\ee
 \item [(iii)] $\mathcal R V^{(-1)}$ is set to zero.
   \end{enumerate}

   The function $W^{(h),R}_{2,0;(\o,\o)}$  will be shown to satisfy both the identity \eqref{eq:upps} and the extra symmetries (in Fourier space)
   \bea && \hat W^{(h),R}_{2,0;(-\o,-\o)}(k)=-[\hat W^{(h),R}_{2,0;(\o,\o)}(k)]^*,\nonumber\\
&& \hat W^{(h),R}_{2,0;(\o,\o)}(A^{-1}\s_1 Ak)=i\o [\hat W^{(h),R}_{2,0;(\o,\o)}(k)]^*,\label{eq:symm}\\
&& \hat W^{(h),R}_{2,0;(\o,\o)}(A^{-1}\s_3 Ak)=[\hat W^{(h),R}_{2,0;(\o,\o)}(k)]^*\nonumber\eea
   where $A=\begin{pmatrix} \bar\a^1 & \bar \b^1 \\ \bar \a^2 & \bar \b^2 \end{pmatrix}$ while $\sigma_1,\sigma_3$ are the first and third Pauli matrices.   Let us assume that $W^{(h'),R}_{2,0;(\o,\o)},h'\ge h$ has been  already shown to satisfy \eqref{eq:symm} and below we explain how to prove the same at scale $h-1$.

   In order to define the running coupling constants on scale $h$, we  decompose the term containing $ \partial_{k}\hat W^{(h),\infty}_{2,0;(\o,\o)}({0})$ in  \eqref{eq:6.21}, by rewriting 
   \be \partial_{k}\hat W^{(h),\infty}_{2,0;(\o,\o)}({0})=\partial_{k}\hat W^{(h),R}_{2,0;(\o,\o)}({0})+\partial_{k}\hat W^{(h),s}_{2,0;(\o,\o)}({0}),\label{Rs}\ee
   ('$s$' stands for `subdominant'). From the symmetries \eqref{eq:symm},
 a straightforward computation (see Appendix \ref{symm}) shows that 
                                                                \be k\cdot \partial_{k}\hat W^{(h),R}_{2,0;(\o,\o)}({0})=-z_h(\bar\a_\o k_1+\bar\b_\o k_2)= -z_h \bar D_\o(k),\label{eq:symm1}\ee
                                                                for some real number $z_h$.
                                                                We now combine this term with the Grassmann Gaussian integration $P_{(\le h)}(d\psi)$, and define: 
                                                                \be P_{(\le h)}(d\psi)e^{-z_hZ_h L^{-2}\sum_\o\sum_{k\in\mathcal P'_\o(\bt)} \bar D_\o(k)\hat \psi^+_{k,\o}\hat \psi^-_{k,\o}}\equiv e^{L^2 t_h}\tilde P_{(\le h)}(d\psi),\label{eh:eh}\ee
                                                                where $\tilde P_{(\le h)}(d\psi)$ is the Grassmann Gaussian integration with propagator
                                                                \be \frac{\tilde g^{(\le h)}_\o(x,y)}{Z_{h-1}}=\frac1{L^2}\sum_{k\in \mathcal P_\o'(\bt)}e^{-ik(x-y)}\frac{\bar\chi_{h}(k)}{\tilde Z_{h-1}(k) \bar D_\o(k)+r_\o(k)},\label{eq:6.28}\ee
                                                                with 
                                                                \be \tilde Z_{h-1}(k):=Z_h(1+z_h\bar \chi_h(k)),\qquad Z_{h-1}:=\tilde Z_{h-1}(0)=Z_h(1+z_h),\label{eq:ecci}\ee
                                                                and $e^{L^2 t_h}$ is a  constant that normalizes $\tilde P_{(\le h)}(d\psi)$ to $1$:
                                                                \be t_h=\frac1{L^2}\sum_\o\sum_{k\in\mathcal P'_\o(\bt)}\log\Big(1+\frac{z_h\bar \chi_h(k)\bar D_\o(k)}{\bar D_\o(k)+r_\o(k)/Z_h}\Big) .\label{eq:th}\ee
   
                                                                By using \eqref{eh:eh}, we rewrite the Grassmann integral in the right side of \eqref{eq:vh} as
                                                                \bea && \int P_{(\le h)}(d\psi)e^{V^{(h)}(\sqrt{Z_h}(\psi+\Psi),J)}=e^{L^2 t_h-z_hZ_h L^{-2}\sum_\o \bar D_\o(k^\o_\bt-\bar p^\o)\hat \Psi^+_{\o}\hat \Psi^-_{\o}}\times\nonumber\\
&&\hskip4.1truecm \times\int\tilde P_{(\le h)}(d\psi)e^{\widehat V^{(h)}(\sqrt{Z_{h-1}}(\psi+\Psi),J)},\label{eq:vieri}\eea
   where 
   \bea \widehat V^{(h)}(\f,J)&=&L^{-2}\sum_{\o}\sum_{k\in \mathcal P_\o(\bt)} 
                                  \hat \f^+_{k,\o} [2^h\nu_{h,\o}+a_{h,\o}k_1+b_{h,\o}k_2\big]\hat \f^-_{k,\o}	\nonumber\\
&+&\l_h \sum_{x\in\L} \f^+_{x,+}\f^-_{x,+}\f^+_{x,-}\f^-_{x,-} \label{eq:whV}\\
&+& \sum_{\o_1,\o_2,r}\frac{Y_{h,r,(\o_1,\o_2)}}{Z_{h-1}}\sum_{x\in\L} J_{x,r}e^{i(\bar p^{\o_1}-\bar p^{\o_2})x}
    \f^+_{x,\o_1}\f^-_{x,\o_2}\nonumber\\
&+&\mathcal RV^{(h)}(\sqrt{{Z_h}/{Z_{h-1}}}\,\f,J),\nn\eea
    and the running coupling constants at scale $h$ are defined as
    \bea && 2^h\nu_{h,\o}=\frac{Z_h}{Z_{h-1}} \hat W^{(h),\infty}_{2,0;(\o,\o)}(0),\label{eq:6.46}\\
&& a_{h,\o}=\frac{Z_h}{Z_{h-1}} \partial_{k_1}\hat W^{(h),s}_{2,0;(\o,\o)}(0), \qquad 
   b_{h,\o}=\frac{Z_h}{Z_{h-1}} \partial_{k_2}\hat W^{(h),s}_{2,0;(\o,\o)}(0), \nonumber \\
&&\l_h =4\Big(\frac{Z_h}{Z_{h-1}}\Big)^2 \hat W^{(h),\infty}_{4,0;(+,+,-,-)}(0,0,0),\nonumber\\
&& Y_{h,r,(\o_1,\o_2)}=Z_h \hat W^{(h),\infty}_{2,1;(\o_1,\o_2),r}(0,\bar p^{\o_1}-\bar p^{\o_2}).\nonumber\eea

%\medskip

   Thanks to the symmetry \eqref{eq:upps} of the kernels
   (that by inductive hypothesis holds at step $h$) the running coupling constants satisfy the following:
   \be \nu_{h,\o}=\nu_{h,-\o}^*, \quad  a_{h,\o}=-a_{h,-\o}^*, \quad  b_{h,\o}=-b_{h,-\o}^*, \quad  Y_{h,r,\ul\o}=Y_{h,r,-\ul\o}^*.\ee
   Moreover $\l_h\in\mathbb R$: for this, one uses both \eqref{eq:upps} 
   and the fact that $$\hat W^{(h)}_{4,0;(+,+,-,-)}(0,0,0)=\hat W^{(h)}_{4,0;(-,-,+,+)}(0,0,0).$$ For later reference, we rewrite the local part of $\widehat V^{(h)}(\f,J)$ as
   \bea \mathcal L{\widehat V}^{(h)}(\f,J)&=&\sum_{\o}\Big[2^h\n_{h,\o}F_{\n;\o}(\f)+a_{h,\o}F_{a;\o}(\f)+b_{h,\o} F_{b;\o}(\f)\Big]\nonumber\\
&+&\l_h F_\l(\f)+ \sum_{r,\ul\o}\frac{Y_{h,r,\ul\o}}{Z_{h-1}}F_{Y;r,\ul\o}(\f,J),\label{lcwhV}\eea
    %where =L^{-2}\sum_{k\in \mathcal P_\o(\bt)} \hat \f^+_{k,\o}\hat \f^-_{k,\o}$, etc.
    (for the definitions of $F_{\n;\o}(\f),F_{a;\o}, F_{b;\o}$, etc., compare \eqref{lcwhV} with the first two lines of \eqref{eq:whV}). 

    We now decompose the propagator \eqref{eq:6.28} as  \[\tilde g^{(\le h)}_\o(x,y)=
    g^{(h)}_\o(x,y)+g^{(\le h-1)}_\o(x,y),\] with $g^{(\le h-1)}_\o$ as in \eqref{eq:gleh} and $g^{(h)}_\o$ as in \eqref{eq:ghjy6}.
    To see that this decomposition holds,  note that $\tilde Z_{h-1}(k)\equiv Z_{h-1}$ on the support of $\bar\chi_{h-1}(\cdot)$.
    
    Then, rewrite \eqref{eq:vieri} as
    \bea && \int P_{(\le h)}(d\psi)e^{V^{(h)}(\sqrt{Z_h}(\psi+\Psi),J)}=e^{L^2 t_h-z_hZ_h L^{-2}\sum_\o \bar D_\o(k^\o_\bt-\bar p^\o)\hat \Psi^+_{\o}\hat \Psi^-_{\o}}\times\nonumber\\
&&\qquad \times \int P_{(\le h-1)}(d\psi)\int P_{(h)}(d\psi')e^{\widehat V^{(h)}(\sqrt{Z_{h-1}}(\psi+\psi'+\Psi),J)},\label{eq:vieriii}\eea
   which implies the validity of the representation \eqref{eq:vh} at scale $h-1$, with $E^{(h-1)}$, $S^{(h-1)}(\cdot)$ and $V^{(h-1)}(\cdot)$ defined by 
\begin{eqnarray} &&e^{L^2 E^{(h-1)} + S^{(h-1)}(J)+ V^{(h-1)}(\sqrt{Z_{h-1}}(\psi+\Psi),J)}=\\
&&\qquad =e^{L^2 (E^{(h)}+t_h)+S^{(h)}(J)}\int P_{(h)}(d\psi') e^{\widehat V^{(h)}(\sqrt{Z_{h-1}}(\psi+\psi'+\Psi),J)},\nonumber\end{eqnarray}
that is, 
\begin{eqnarray} && L^2 (E^{(h-1)}-E^{(h)}-t_h) +(S^{(h-1)}(J)-S^{(h)}(J))+ V^{(h-1)}(\f,J) \label{eq:tb2}\\
&&\quad =\sum_{n\ge 1}\frac1{n!}\mathcal E^T_h(\underbrace{\widehat V^{(h)}(\f+\sqrt{Z_{h-1}}\psi',J);\cdots;
\widehat V^{(h)}(\f+\sqrt{Z_{h-1}}\psi',J)}_{n\ {\rm times}},\nonumber
\end{eqnarray}
   with $\mathcal E^T_h$ the truncated expectation w.r.t. the Grassmann Gaussian integration $P_{(h)}(d\psi)$, and $E^{(h-1)}$, $S^{(h-1)}(\cdot)$ fixed as usual by the conditions
   $S^{(h-1)}(0)=0$ and $V^{(h-1)}(0,J)=0$.

   To conclude the proof of the induction step, it remains to prove that the kernels of $V^{(h-1)}$ satisfy \eqref{eq:upps} and that 
   \eqref{eq:symm} holds, at scale $h-1$.
   The proof of the former statement is very similar (but not identical) to the argument used in Section \ref{sec:h-1} to
   prove  \eqref{eq:upps} at scale $h=-1$ starting from the symmetries of $V^{(0)}$. Namely,  thanks to 
   \eqref{eq:upps} at scale $h$, the potential $V^{(h)}$ is invariant under $\varphi^\pm_{x,\o}\to\varphi^\pm_{x,-\o}$ together with complex conjugation of the kernels. Then, the claim follows from the representation \eqref{eq:tb2}, together with the fact that the propagator $g^{(h)}$ (defined in \eqref{eq:ghjy6}) satisfies the symmetry 
   \begin{eqnarray}
     \label{eq:symmgh}
     [g^{(h)}_\o(x,y)]^*=g^{(h)}_{-\o}(x,y).
   \end{eqnarray}
   As for  \eqref{eq:symm} at scale $h-1$, the proof uses the  symmetries of the relativistic propagator \eqref{gRo}, together with the fact that $\lambda_{h'}$ is real. See Appendix \ref{symm}.
   \begin{Remark}
     \label{rem:gRo}
     Note that, if the function $z_{h'}$ in \eqref{eq:ecci} is
     sufficiently small for all the scales $h\le h'\le-1$, say
     $|z_{h'}|\le \e $ uniformly in $L,h'$, then
     $e^{-c \e |h|}\le Z_h\le e^{c \e |h|}$. As a
     consequence, $g^{(h)}_\o$ satisfies a bound analogous to
     \eqref{eq:bbpprr}, namely
     \be \label{eq:bbpprrh}|g^{(h)}_\o(x,y)|\le C_0 2^h e^{-\kappa
       \sqrt{2^h|x-y|}}. \ee In fact, note that on the support of
     $f_h(\cdot)$ (which is concentrated on $k:|k|\sim 2^h$),
     $\tilde Z_{h-1}(k)/Z_{h-1}=1+O(\e)$ and recall that
     $r_\o(\cdot)$ is quadratic for small values of its argument, so that $r_\o(k)/Z_{h-1}$ is negligible w.r.t. $\bar D_\o(k)$.
     The propagator $g^{(h)}_{R,\o}$ satisfies the same estimate as \eqref{eq:bbpprrh}, while the difference $g^{(h)}_{\o}-g^{(h)}_{R,\o}$ satisfies an estimate that is better by a factor $2^h$.
   \end{Remark}
   \subsubsection{The Beta function}

   The iterative integration scheme described above allows us to express the kernels of $V^{(h)}$  and, in particular, the running coupling 
   constants (RCC) at scale $h$, as functions of the sequence of RCC on higher scales, $\{\l_{h'},\n_{h',\o},a_{h',\o},b_{h',\o},Y_{h',r,\underline \o}\}_{h<h'\le -1}$, of the single-scale propagators $\{g^{(h')}_\o/Z_{h'-1}\}_{h<h'\le -1}$, and of $\mathcal R V^{(-1)}$. That is, we can rewrite Equation \eqref{eq:6.46} in the form
\begin{eqnarray} &&
\n_{h-1,\o}=2\n_{h,\o}+B^\n_{h,\o}\quad a_{h-1,\o}=a_{h,\o}+B^a_{h,\o},\quad b_{h-1,\o}=b_{h,\o}+B^b_{h,\o},\nonumber\\
&&\l_{h-1}=\l_h+B^\l_h,\qquad Y_{h-1,r,\ul\o}=Y_{h,r,\ul\o}+B^Y_{h,r,\ul\o},\label{beta2}
\end{eqnarray}
   where $B^\#_{h,\cdot}$, $h\le -1$, is the so-called Beta function. One has to think of $B^\#_{h,\cdot}$ as a function
   of the RCC on scales $h'$ with  $h\le h'\le 0$.
   Note that the first four equations makes sense also with $h=0$, in which case they express the relation between 
   $(\n_{-1,\o},a_{-1,\o},b_{-1,\o},\l_{-1})$ and $(\n_{0,\o},a_{0,\o},b_{0,\o},\l_{0})$, see \eqref{rcc0}. 
   Note also that by construction the beta function $B^\#_{h,\cdot}$ depends on $Z_{h'}$ only via the combinations 
   $Z_{h'}/Z_{h'-1}=(1+z_{h'})^{-1}$, with $h< h'< 0$. For later reference, we rewrite the definition of $z_h$, \eqref{eq:symm1}, in a form analogous to \eqref{beta2}, 
\begin{equation}\label{betazeta}z_{h-1}=B^z_h,  \quad h\le 0,\end{equation}
where the right side is thought of as a function of $(\l_{h'},z_{h'})_{h\le h'\le 0}$, with the convention that $z_{0}=z_{-1}=0$ (the latter is because $W^{(-1),R}_{2,0;(\o,\o)}\equiv 0$).

   \begin{Remark}
     \label{rem:noY}
The components of the beta function for $\n_{h,\o},a_{h,\o},b_{h,\o},\l_h$
%     \[\underline
%       u_h:=(\{\n_{h,\o},a_{h,\o},b_{h,\o}\}_{\o\in\{\pm\}},\l_h)\]
     are independent of $Y_{h',r,\underline \o},h'>h$. Therefore, we can first
     solve the flow equation for $\n_{h,\o},a_{h,\o},b_{h,\o},\l_h$ and then inject the
     solution into the flow equation for $Y_{h,r,\underline \o}$.
   \end{Remark}
   
   Before we proceed in describing the dimensional bounds satisfied
   by the kernels of the effective potential, let us comment on their
   structure. We have proven inductively that $V^{(h)}$ has, in momentum space, the same structure as \eqref{eq:6.16mo}. If one writes $V^{(h)}$ in real space,  due to 
   iterative action of the $\mathcal R$ operator  in the inductive procedure explained above, the structure  that emerges naturally 
   is that of a polynomial with 
   spatial derivatives acting on some of the  Grassmann fields $\f^\pm_{x,\o}$.
   For an explanation of why this is the case see
   \cite[Section 6.1.4]{GMT17a} and
   Appendix \ref{dimr} below, where finite-size effects associated with the action of $\mathcal R$ are also discussed.
   Correspondingly, $V^{(h)}$ can be represented as
\begin{eqnarray}
&& V^{(h)}(\f,J)= \sum_{\substack{n,m\ge 0:\\ n\ {\rm even}, \ n\ge 2}}\sum_{\ul x,\, \ul y,\, \ul \o,\,\ul r,\, \ul i,\ul q} W^{(h)}_{n,m,\ul i,\ul q;\ul\o,\ul r}(\ul x, \ul y)\times\label{eq:6.16ren}\\
&&\qquad \times\, \hat \partial^{q_{1}}_{i_1}\f^{(\le h)+}_{x_1,\o_1}\cdots\hat \partial^{q_{n}}_{i_n}\f^{(\le h)-}_{x_n,\o_n} J_{y_1,r_1}\cdots J_{y_m,r_m}.\nonumber\end{eqnarray}
% where $\f^{(\le h)\pm}_{x,\o}=\psi^{(\le h)\pm}_{x,\o}+L^{-2}e^{\pm i (k^\o_\bt-\bar p^\o)x}\hat \Psi^\pm_\o$. 
   The main difference between this formula and \eqref{eq:6.16}, besides the scale label $h$ replacing $-1$, is the presence of the indices 
   $\ul i=(i_1,\ldots,i_n)\in\{1,2\}^n$ and $\ul q=(q_1,\ldots,q_n)\in\{0,1,2\}^n$ and the operators 
   $\hat \partial^{q}_i$ acting on the Grassmann fields: this is a  differential operator, dimensionally equivalent to a  derivative of order $q$ in direction $i$. Let us stress that the representation in \eqref{eq:6.16ren} 
   is not unique: the claim is that there exists such a representation, with the kernels satisfying natural dimensional estimates, discussed below. 

\medskip

In order for the iterative construction to allow us to compute the thermodynamic and correlation functions, we need to prove that: 
(i) the RCC $\n_{h,\o},a_{h,\o}, b_{h,\o},\l_{h},z_h$ are small, uniformly in the scale (say, smaller than a sufficiently small constant $\e$), 
provided the functions $\bar p^\o,a_\o,b_\o$ (see \eqref{mu0}) 
have been properly fixed; (ii) the kernels of the effective potential are all well defined (i.e. the sums \eqref{eq:tb2} are convergent uniformly in $L$), quasi-local (i.e., fast decaying, with a stretched-exponential behavior) and satisfy natural scaling properties, i.e., 
\be \label{bouW}\|W^{(h)}_{n,m,\ul i, \ul q}\|_{\kappa,h}\le C^{n+m} \e^{\max\{1,cn\}} 2^{h(2-n/2-m-|\ul q|)} \Big(\max_{h'\ge h}\frac{|Y_{h',\cdot}|}{|Z_{h'}|}\Big)^m, \ee
with $|\ul q|=\sum_{i=1}^nq_{i}$, $|Y_{h',\cdot}|=\max_{r,\ul\o}|Y_{h',r,\ul\o}|$, and  
\be\label{norW}\|W^{(h)}_{n,m,\ul i, \ul q}\|_{\kappa,h}:=L^{-2}\sup_{\ul \o, \ul r}\sum_{\ul x, \ul y}|W^{(h)}_{n,m,\ul i, \ul q;\ul\o,\ul r}(\ul x, \ul y)|e^{\kappa \sqrt{2^{h}d(\ul x,\ul y)}},\ee
for suitable constants $C,c,\kappa$, independent of $L,h$. 

\medskip

   The boundedness of the flow of the RCC and the validity of the dimensional bounds for the kernels of the effective potential will, in fact, be the final outcome of our analysis. The logic of proof goes as follows: 
   one first proves the validity of the dimensional bounds on the kernels, under the assumption that the RCC remain small. 
   These bounds will, in particular, imply that the components of the beta function are well defined and satisfy bounds that are uniform in $L$ and $h$. 
   This part of the proof is pretty standard: it follows from a representation of the effective potential in terms of Gallavotti-Nicol\`o (GN) trees, see Section \ref{sectree} below, and an iterative application of the 
   Battle-Brydges-Federbush-Kennedy (BBFK) determinant formula, see, e.g., \cite[Lemma 3]{GMT17a}. 

   Next, we prove that the RCC remain bounded, by studying the flow generated by the beta function. The key point is that, as long as the RCC on scales larger than $h$ are small, then the beta 
   function on scale $h$ is well defined, and can be used to control the evolution of the RCC for another step. This opens the way to an inductive proof of the smallness of the RCC. Of course, 
   the fact that RCC remain small at all scale requires a specific (model-dependent) structure of the beta function. In our case, we are lucky enough that the beta function has structure which maintains 
   the RCC small at all scale, provided the initial data are small, and that $\bar p^\o,a_\o,b_\o$ are properly fixed, see Section \ref{secfl} below. 
   It is not just a matter of luck, of course: a key point in the analysis is played by the comparison 
   of the $\l$-component of the beta function of our dimer model, with the corresponding quantity for the reference continuum model (the two functions are the same at dominant order). 
   The exact solvability of the reference model implies the validity of a remarkable cancellation in the $\l$-component of the beta function for the reference model and, therefore, a posteriori, for our dimer model, 
   as well. 

\subsection{Tree expansion for the effective potential}\label{sectree}

   As anticipated above, the detailed structure of the kernels of $V^{(h)}$, arising from the iterative construction described in the previous section, can be conveniently 
   represented in terms of GN trees. The definition 
   of GN trees, of their values, and the procedure leading to their introduction
   have been discussed at length in several previous papers and will not be repeated here, see e.g. \cite{GMreview}; in particular, we refer to \cite[Section  5.2.1 and 6.2]{GMT17a} for a description of the GN tree 
   expansion in a context very similar to the present one, i.e., in the case of isotropic, `tilt-less', interacting dimer models with weights $\ul t\equiv 1$ and plaquette interaction. The present case differs from the one treated in \cite{GMT17a}  for the fact that 
   here the model is anisotropic (and, correspondingly, the height has an average slope that is different from zero). Technically, this means that in the present case the expansion involves more running coupling 
   constants than those considered in \cite{GMT17a}: the RCC $\nu_{h,\o},a_{h,\o},b_{h,\o}$ are identically zero in the tilt-less case. In particular, the trees involved in our construction are characterized by the following features, slightly different from those listed in \cite[Section  6.2]{GMT17a}:

   \begin{enumerate}
\item A GN tree $\t$ contributing to $V^{(h)}$, $\tilde S^{(h)}(J):=S^{(h)}(J)-S^{(h+1)}(J)$, or to $\tilde E^{(h)}=E^{(h)}-E^{(h+1)}-t_{h+1}$  
has root on scale $h$ and can have endpoints (either `normal' or `special', which are those represented as black dots or white
squares, respectively, in \cite{GMT17a}, see, e.g., \cite[Fig.13]{GMT17a}) on all possible scales between $h+2$ and $0$. The endpoints $v$ on scales 
$h_v< 0$ are preceded by a node $v'$ of $\t$, on scale $h_{v'}=h_v-1$, that is necessarily a branching point. The family of GN trees with root on scale $h$, $N_n$ normal endpoints and $N_s$ special endpoints 
is denoted by $\mathcal T^{(h)}_{N_n,N_s}$. 
\item A normal endpoint $v$ on scale $h_v\le 0$ can be of five different types, $\l,\n,a,b$, or $\mathcal R V^{(-1)}$. If $v$ is of type $\l,\n,a$ or $b$, then it 
is associated with $\l_{h_{v'}} F_\l(\sqrt{Z_{h_{v'}-1}}\f^{(\le h_{v'})})$, 
or $\sum_\o \n_{h_{v'},\o}F_{\n;\o}(\sqrt{Z_{h_{v'}-1}}\f^{(\le h_{v'})})$, or 
$\sum_\o a_{h_{v'},\o}F_{a;\o}(\sqrt{Z_{h_{v'}-1}}\f^{(\le h_{v'})})$, or 
$\sum_\o b_{h_{v'},\o}F_{b;\o}(\sqrt{Z_{h_{v'}-1}}\f^{(\le h_{v'})})$, depending on its type (recall that the monomials $F_\l$, $F_{\n;\o}$, etc., were defined in \eqref{lcwhV}); in this case,
the node $v'$ immediately preceding $v$ on $\t$, of scale $h_{v'}=h_v-1$, is necessarily a branching point. If $v$ is of type $\mathcal R V^{(-1)}$, then $h_v=0$, and $v$ 
is associated with (one of the monomials contributing to) $\mathcal RV^{(-1)}(\f^{(\le -1)},0)$; in this case, the node immediately preceding $v$ 
on $\tau$, of scale $h_v-1$, is not necessarily a branching point.
\item A special endpoint $v$ on scale $h_v\le 0$ can be either local, or non-local. If $v$ is local, then it is associated with 
\be\label{YZ}\frac{Y_{h_{v'},r,\ul\o}}{Z_{h_{v'}-1}}F_{Y;r,\ul\o}(\sqrt{Z_{h_{v'}-1}}\f^{(\le h_{v'})},J),\ee
for some $r\in\{1,2,3,4\}$, $\ul\o=(\o_1,\o_2)\in\{\pm\}^2$; if $\o_1=\o_2$, we shall say that $v$ is a `density endpoint', while, if $\o_1\neq\o_2$, that $v$ is a `mass endpoint'. 
Note that the factors  $Z_{h_{v'}-1}$ in \eqref{YZ} simplify: the summand equals $Y_{h_{v'},r,\ul\o}F_{Y;r,\ul\o}(\f^{(\le h_{v'})},J)$; in \eqref{YZ}, these factors are kept just for uniformity of notation 
with the cases in the previous item. In the case that $v$ is local, the node $v'$ immediately preceding $v$ on $\t$, of scale $h_{v'}=h_v-1$, is necessarily a branching point. 
If $v$ is non-local, then $h_v=0$, and $v$ is associated with (one of the monomials contributing to) $V^{(-1)}(\f^{(\le -1)},J)-V^{(-1)}(\f^{(\le -1)},0)$; in this case, the node immediately preceding $v$ 
on $\tau$, of scale $h_v-1$, is not necessarily a branching point. 
\end{enumerate}
In addition to the items above, let us recall that each vertex of the tree that is not an endpoint and that is not the special vertex $v_0$ (the leftmost vertex of the tree, immediately following the root on $\t$) is 
associated with the action of an $\RR$ operator.

In terms of the tree expansion, we can express the effective potential and the single-scale contributions to the free energy and generating function as 
\be L^2 \tilde E^{(h)}+\tilde S^{(h)}(J)+V^{(h)}(\sqrt{Z_h}\f,J)=\sum_{\substack{N_n,N_s\ge 0:\\ N_n+N_s\ge 1}}\sum_{\t\in \mathcal T^{(h)}_{N_n,N_s}} V^{(h)}(\t,\sqrt{Z_h}\f,J),
\label{eq:6.63e}\ee
where
\bea && V^{(h)}(\t,\sqrt{Z_h}\f,{J})=\label{f5.49ab}\\
&& \quad =\sum_{
{\bf P}\in\PP_\t}\sqrt{Z_h}^{|P_{v_0}^\psi|}\sum_{T\in{\bf T}}\sum_{\bf i, \bf q}\sum_{ \xx_{v_0}}W_{\t,{\bf P},{T},{\bf i}, {\bf q}}
(\xx_{v_0}) D_{\bf i}^{\bf q}\f({P_{v_0}^\psi})\, J({P_{v_0}^J})
\;.\nonumber\eea
Eq.(\ref{f5.49ab}) is the analogue of \cite[(6.64)]{GMT17a}, and we refer the reader to that paper for the notation and a sketch of the proof (in this formula, the indices $\bf i,\bf q$ replace the 
multi-indices $\b\in B_T$ \cite[(6.64)]{GMT17a}). We recall that 
$P_{v_0}^\psi$ and $P_{v_0}^J$ are two sets of indices that label the Grassmann external fields and the external fields of type $J$, respectively; moreover, 
$J({P_{v_0}^J})=\prod_{f\in P_{v_0}^J}J_{y(f),r(f)}$ and 
\be D_{\bf i}^{\bf q}\f({P_{v_0}^\psi})=\prod_{f\in P_{v_0}^\psi} \hat \partial_{i(f)}^{q(f)}\f_{x(f),\o(f)}^{\e(f)}\;.\label{eq:6.65}\ee
Clearly, the kernels in \eqref{eq:6.16ren} are obtained by summing $W_{\t,{\bf P},{T},{\bf  i},{\bf q}}
(\xx_{v_0})$ over $\t\in \mathcal T^{(h)}_{N_n,N_s}$ and over $N_n,N_s$, 
under the constraint that the number of external fields of type $\psi$ and $J$ is equal to $n$ and $m$, respectively, that the elements of ${\bf i}$ are the same as $\ul i$, etc. 
Similarly, the kernels of the single scale contribution to the generating function, $\tilde S^{(h)}(J)$, which we denote by $W^{(h)}_{0,m;\ul r}(\ul y)$, 
are obtained by summing the tree values $W_{\t,{\bf P},{T}}$
over $\t\in \mathcal T^{(h)}_{N_n,N_s}$ and over $N_n,N_s$, under the constraint that $P_{v_0}^\psi=\emptyset$ and that $\cup_{f\in P_{v_0}^J}\{(y(f), r(f))\}$ matches the tuple $(\ul y,\ul r)$; finally, the single scale contribution to the free energy, $L^2 \tilde E^{(h)}$ is obtained by an analogous sum over GN trees, 
under the constraint that $P_{v_0}^\psi=P_{v_0}^J=\emptyset$.

\medskip

The bound \eqref{bouW}, as well as the analogous one for $W_{0,m:\ul r}$, is a corollary of the following fundamental bound on the weighted $L_1$ norm of $W_{\t,{\bf P},{T},{\bf i},{\bf q}}$, which is the analogue of 
\cite[Proposition 8]{GMT17a} and of \cite[(3.110)]{BM-XYZ}; for the proof, we refer the reader to \cite{BM-XYZ,GMT17a}. See also Appendix  \ref{dimr} below for some technical details.
\begin{Proposition}\label{prop:an} There exists $L$-independent constants $\e,C,c,\kappa>0$ such that, if 
\be \max_{h'>h}\{|\l_{h'}|,|\n_{h',\o}|,|a_{h',\o}|,|b_{h',\o}|,|z_{h'}|\}\le \e,\label{ggg0}\ee
and $\tau\in\mathcal T^{(h)}_{N_n,N_s}$, then 
\bea && \|W_{\t,{\bf P},T,{\bf i},{\bf q}}\|_{\kappa,h}\le C^{N_s}
\,(C\e)^{\max\{N_n,c|I_{v_0}^\psi|\}}\,
2^{h(2-\frac12|P_{v_0}^\psi|-
|P_{v_0}^J|-|{\bf q}|)}\label{5.62}\\
&&\times \Biggl[\prod_{v\ {\rm s.e.p.}}\sup_{r,\ul\o}
\Big|\frac{Y_{h_v-1,r,\ul\o}}{Z_{h_v-1}}\Big|\Biggr]\Big[\prod_{\substack{v\,{\rm not}\\ {\rm e.p.}}}\hskip-.02truecm\frac{C^{\sum_{i=1}^{s_v}|P_{v_i}|-|P_v|}}{s_v!} \;2^{ \frac\e2|P^\psi_v|}2^{2-\frac12|P_v^\psi|-|P_v^J|-z(P_v)}
\Big]\;,\nonumber\eea
where: $|I_{v_0}^\psi|=\sum_{v\ {\rm e.p.}}|P_{v}^\psi|$ is the total number of Grassmann fields associated with the endpoints of the tree; the first product in the second line runs over the special endpoints, while the second over all the vertices of the tree that are not endpoints. 
Moreover $|{\bf q}|=\sum_{f\in P_{v_0}^\psi}q(f)$
and 
\be z(P_v)=\begin{cases} 1& {\rm if}\quad (|P_v^\psi|,|P_v^J|)=(4,0),(2,1),\\
2 & {\rm if}\quad (|P_v^\psi|,|P_v^J|)=(2,0),\\
0 & {\rm otherwise.}\end{cases}\label{eq:zpv}\ee
%\label{prop:If}
\end{Proposition}
The dimensional gain $2^{-z(P_v)}$ associated with the marginal and relevant vertices, i.e., those with $(|P_v^\psi|,|P_v^J|)=(4,0),(2,1),(2,0)$, comes from the action of $\mathcal R$, as explained in 
\cite[Section 6.1.4]{GMT17a} and in Appendix \ref{dimr} below. 

    Since the exponents $2+ \frac\e2|P^\psi_v|-\frac12|P_v^\psi|-|P_v^J|-z(P_v)$ in \eqref{5.62} are all strictly negative, one can  sum \eqref{5.62} over $\t\in \mathcal T^{(h)}_{N_n,N_s}$, over $T\in{\bf T}$, and over ${\bf P}\in\PP_\t$, under the constraint that $|P_{v_0}^\psi|=n$ and $|P_{v_0}^J|=m$,
    we get the bound \eqref{bouW}; see also the discussion after \cite[Proposition 8]{GMT17a}.
    Similarly,
if we sum \eqref{5.62}  over $\t\in \mathcal T^{(h)}_{N_n,N_s}$, $T\in{\bf T}$, ${\bf P}\in\PP_\t$, with $|P_{v_0}^\psi|=n$, $|P_{v_0}^J|=m$, {\it under the additional constraint that $\t$ has 
    at least one node on scale $k>h$}, then we get a bound that is the same as \eqref{bouW} times an additional gain factor $2^{\th' (h-k)}$, where $\th'$ is a positive constant, smaller than 1 (estimates are 
    not uniform as $\th'\to 1^-$;  from here on, we will choose $\th'=3/4$). 
This is the so-called {\it short memory property}, see Remark 16 after \cite[Proposition 8]{GMT17a}.

\subsection{The flow of the running coupling constants}\label{secfl}

    As explained above, as long as the RCC $\n_{h',\o},a_{h',\o}, b_{h',\o},\l_{h'},z_{h'}$ stay small, for all $h'>h$, in the sense of \eqref{ggg0}, the beta function controlling the flow of
    the same constants on all scales larger or equal to $h$, see \eqref{beta2}-\eqref{betazeta},  
    can be represented in terms of a convergent GN expansion, induced by the one of the kernels of the effective potential discussed above. 
    The goal is then to fix the initial data $\n_{0,\o}, a_{0,\o}, b_{0,\o}$, in such a way that the resulting flow of $\n_{h,\o},a_{h,\o},b_{h,\o},\l_h,z_h$
    driven by the beta function stays uniformly small in the scale index. 
    For this purpose, not only we have to make a careful choice of the `counter-terms' $\n_{0,\o}, a_{0,\o}, b_{0,\o}$, but we also need to exploit a number of remarkable cancellations, 
    some of which follow from the exact solution of the reference model of Section \ref{sec:IR}. Let us emphasize that we have the right to fix the counter-terms, which up to now were chosen arbitrarily in \eqref{mu0} (recall \eqref{rcc0}), but we cannot change $\l_0=\l$, which enters the definition of the model.  
    \medskip

We look for a solution of the flow equation for $\ul u_h$ such that, as $h\to-\infty$: 
\begin{enumerate}
\item $\nu_{h,\o},a_{h,\o},b_{h,\o}$ tend exponentially to zero; more precisely, recalling that $|\n_{h,+}|=|\n_{h,-}|$, and similarly for $|a_{h,\o}|$, $|b_{h,\o}|$, we require that
\be \|(\underline\nu,\underline a,\underline b)\|_\th:=\sup_{h\le 0}\{2^{-\th h}|\nu_{h,+}|,\, 2^{-\th h}|a_{h,+}|,\, 2^{-\th h}|b_{h,+}|\}\le \e,\label{nubound}\ee
for $\e$ small enough and $\th=1/2$;
\item $\l_h$ tends exponentially to a finite limiting value $\l_{-\infty}$; more precisely, given a positive constant $\e'$ smaller than the constant $\e$ in the previous item, 
we require that $|\l_0|\le \e'$ and 
\be \|\underline\lambda\|_\th:=\sup_{h\le 0}\{2^{-\th h}|\lambda_{h-1}-\lambda_h|\}\le \e',\label{eq:boundlambda}\ee
where $\th$ is the same as in the previous item; note that, from the condition on $\l_0$ and \eqref{eq:boundlambda}, it follows that 
\be \label{llaa}|\l_h|\le \frac{\e'}{1-2^{-\th}}+\e',\ee uniformly in $h$.
\end{enumerate} 

In order to construct such solution, we proceed in various steps. 
\begin{itemize}
\item First, in Section \ref{secz*}, given any sequence $\ul\lambda$ satisfying the conditions of item (2), we show how to construct the 
solution $z_h$ to the beta function equation \eqref{betazeta} with initial datum $z_0=z_{-1}=0$. \item Next, in Section \ref{secfixedp}, we reformulate the problem of constructing a bounded solution to the beta function 
equations for $\n_{h,\o},a_{h,\o},b_{h,\o},\l_h$ into a fixed point problem of the form $\ul u=T(\ul u)$ in a Banach space of sequences. \item In order to construct and prove uniqueness of such a fixed point we appeal to the Banach fixed point theorem, a.k.a. the contraction theorem: in Section \ref{secinvarianceX} we prove the existence of an invariant set $X_\e$, i.e., a set such that $T(X_\e)\subseteq X_\e$, and in Section \ref{seccontractionX} we prove that $T$ is a contraction on $X_\e$. \item In the remaining subsections, we collect various consequences of the proof of existence of the fixed point (or, equivalently, of the desired bounded sequence of RCCs): 
in Section \ref{seccount} we invert the counterterms and explain how to compute the dressed Fermi points $\bar p^{\omega}$ and Fermi velocities $\bar \alpha_\omega,\bar\beta_\omega$ 
(the functions $\bar p^{\omega}$, $\bar \alpha_\omega,\bar\beta_\omega$ constructed here are the same as items (1)-(2) of Proposition \ref{prop:comp}); 
in Section \ref{secflowZ} we discuss the flow of $Z_h$ and define the associated critical exponent $\eta$; in Sections \ref{remmm} and  \ref{bareZ}
we fix the bare couplings $\lambda_\infty$ and $Z$ of the reference model, in such a way that the flows of the reference and dimer model have the same asymptotic behavior as $h\to -\infty$, including the same
critical exponents (the choices of $\lambda_\infty, Z$ made there are those stated in item (3) of Proposition \ref{prop:comp}); finally, in Section \ref{flowY}, 
we discuss the flow of $Y_{h,r,\ul\o}$ and define the associated critical exponents $\eta,\eta_1$.\end{itemize}

\subsubsection{Fixing $(z_h)_{h\le -1}$.}\label{secz*} Given a sequence $\underline \l:=(\lambda_h)_{h\le -1}$ satisfying \eqref{eq:boundlambda}, we construct the
solution of the beta function equation 
    \be z_{h-1}=B^z_h(\underline\lambda,\underline z)\label{bbetazzeta}\ee
{for $\underline z:=(z_h)_{h\le -1}$} iteratively in $h$, starting from $h=0$. {We denote this solution by $\underline z^*(\underline\lambda)$. We recall that the definition of the function $B^z_h$, first introduced in 
    \eqref{betazeta}, is induced by the iterative definition of $z_h$, \eqref{eq:symm1}. Note that, of course, $B^z_h(\underline\lambda,\underline z)$ 
only depends on the components of $\ul\lambda,\underline z$ of scale index larger or equal to $h$; note also that $B^z_h$ does not depend on the irrelevant part of the effective interaction and, in particular, it does not depend on 
$\l_0$, because its definition only involves $W^{(h),R}_{2,0;(\o,\o)}$ (see \eqref{eq:symm1} and the definition of $W^{(h),R}_{2,0;(\o,\o)}$ given after Remark \ref{remark:10}), specifically item (iii) after \eqref{gRo}).}

By using the tree expansion of the beta function, we now show that the solution $\underline z^*(\underline\lambda)$ of \eqref{bbetazzeta}
is a Cauchy sequence, differentiable in $\ul\l$;
more precisely, we prove that, for $\l_0$ fixed, such that $|\l_0|\le \e'$, and $\underline \l$ satisfying \eqref{eq:boundlambda}, 
\be |z_{h-1}^*(\underline\l)-z_h^*(\underline\l)|\le C_0(\e')^2 2^{\th h}, \qquad \Big|\frac{\partial z^*_h (\underline\l)}{\partial \l_k}\Big|\le C_0\e' 2^{\th(h-k)}, \label{z*}\ee
for all $h\le k\le -1$. Once this is done, we plug $\ul z^*(\ul\l)$ in the flow equations for $\n_{h,\o},a_{h,\o},b_{h,\o},\l_h$, i.e., the first four equations of \eqref{beta2}, 
so that a posteriori their beta functions are re-expressed purely in terms of $\l_0$ and $\underline u$, where
\be \underline u=(\underline \nu, \underline a, \underline b,\underline \l), \label{ulu}\ee
with $\underline \nu:=(\n_{h,\o})_{h\le 0}^{\o\in\{\pm\}}$, $\underline a:=(a_{h,\o})_{h\le 0}^{\o\in\{\pm\}}$ and $\underline b:=(b_{h,\o})_{h\le 0}^{\o\in\{\pm\}}$. 

\medskip

    Let us prove the first inequality in \eqref{z*}, inductively in $h$. Note that at the first step, $h=0$, the inequality is trivially true, simply because $z^*_0(\underline\l)=z^*_{-1}(\underline\l)=0$. 
    We assume that $|z_{h'-1}^*(\underline\l)-z_{h'}^*(\underline\l)|\le C_0(\e')^2 2^{\th h'}$, for all scales $h<h'\le 0$, and we want to prove that the same bound holds for $h'=h$. 
    Note that, for $\e'$ sufficiently small, the first inequality in \eqref{z*} also implies that $|z^*_{h'}(\underline\l)|\le \e'\le \e$, $\forall h\le h'\le -1$, uniformly in $\underline\l$. Recall that the definition 
    of $B^z_{h}$ is induced by \eqref{eq:symm1}. The kernel $W^{(h),R}_{2,0;(\o,\o)}$ can be written as a sum over GN trees, analogous to \eqref{f5.49ab}, and the contribution associated with each tree can be bounded as in Proposition \ref{prop:an}. Note that the trees contributing to $B^z_h$ have only endpoints of type $\l$: this is because $z_h$ is defined in 
terms of $\partial_k \hat W^{(h),R}_{2,0;(\o,\o)}(0)$, see \eqref{eq:symm1}, and $W^{(h),R}_{2,0;(\o,\o)}$ is obtained, by definition, by: setting the RCCs $\n_{h',\omega},a_{h',\o},b_{h',\o}$, as well as the irrelevant coupling $\mathcal R V^{(-1)}$, 
to zero; replacing the single-scale propagators $g^{(h')}_{\o}$ by their `relativistic' counterpart, $g^{(h')}_{R,\o}$; see the discussion after \eqref{quartic:symm}.
Moreover, the contribution from the tree with exactly one $\lambda$ endpoint (which corresponds, in the language of Feynman diagrams, to the `tadpole')
is exactly zero (because the kernel of the tadpole is proportional to a delta function, so that the associated contribution to $\partial_k \hat W^{(h),R}_{2,0;(\o,\o)}(0)$ is zero). These considerations lead to the following representation:
    \be B^z_{h}(\underline \l,\ul z^*)=\sum_{N\ge 2}\sum_{\t\in \mathcal T^{(h)}_{N,0}} \sum_{{\bf P}\in\PP_\t}\sum_{T\in{\bf T}}B^z(\underline\l,\underline z^*;\t,{\bf P},T),\label{6.75-0z}\ee
    where $\ul z^*=\ul z^*(\ul \l)$ (recall that $B^z_{h}$ depends only on the components of $\ul z^*$ with scale index $\ge h$, which have already been inductively defined), 
    and $B^z(\underline \l,\underline z^*;\t,{\bf P},T)$ can be bounded in a way analogous to \eqref{5.62}: 
    \be |B^z(\underline\l,\underline z^*;\t,{\bf P},T)|\le (C\e')^N\Big[\prod_{\substack{v\,{\rm not}\\ {\rm e.p.}}}\hskip-.02truecm\frac{C^{\sum_{i=1}^{s_v}|P_{v_i}|-|P_v|}}{s_v!} \;2^{\frac{\e'}2{|P^\psi_v|}}2^{2-\frac12|P_v|-z(P_v)}\Big]\;.\label{smp}\ee
  Here we used the fact that the endpoints are all of type $\l$, and that their values are bounded as in \eqref{llaa}. 
  We now split $B^z_h(\ul\l,\ul z^*)$ as follows: % (with obvious notation, we write $\ul\l=(\l_{-1},\ul\l_{\le -2})$): 
  \be B^z_h(\ul\l,\ul z^*)=\big[B^z_h(\ul\l,\ul z^*)-B^z_h(\ul\l,\ul z^*)\big|_{\l_{-1}=0}\big]+B^z_h(\ul\l,\ul z^*)\big|_{\l_{-1}=0}.\ee
  By definition, the difference in square brackets is expressed in terms of a sum over trees that have at least one endpoint on scale 0, while 
  $B^z_h(\ul\l,\ul z^*)\big|_{\l_{-1}=0}$ is a sum over trees that have no endpoints on scale 0. By using the short memory property
  (see comments after the statement of Proposition \ref{prop:an}), we find
  \be \Big|B^z_h(\ul\l,\ul z^*)-B^z_h(\ul\l,\ul z^*)\big|_{\l_{-1}=0}\Big|\le C(\e')^2 2^{\th h}.\label{dajjje} \ee
  An important remark is that by rescaling $h\to h+1$, we can re-express $B^z_h(\ul\l,\ul z^*)\big|_{\l_{-1}=0}$ in terms of $B^z_{h+1}$: 
  \be B^z_h(\ul\l,\ul z^*)\big|_{\l_{-1}=0}=B^z_{h+1}(S\ul\l,S\ul z^*),\label{scale}\ee
  where $S$ is the shift operator, namely, $(S\ul\l)_h:=\l_{h-1}$, and similarly for $S\ul z^*$. This follows from the fact that the 
 tree expansion for $B^z_{h}$ involves the relativistic propagators $g^{(h')}_{R,\o}$, which are homogeneous, scale-covariant, functions: $g^{(h')}_{R,\o}(x,y)=2^{-1}g^{(h'+1)}_{R,\o}(x/2,y/2)$.
In conclusion, 
\bea z_{h-1}^*(\underline\l)-z_h^*(\underline\l)&=&\big[B^z_h(\ul\l,\ul z^*)-B^z_h(\ul\l,\ul z^*)\big|_{\l_{-1}=0}\big]\label{second}\\
&+&\big[B^z_{h+1}(S\ul\l,S\ul z^*)-B^z_{h+1}(\ul\l,\ul z^*)\big].\nonumber\eea

    We want to bound 
    the difference in the second line as
    \be \big|B^z_{h+1}(S\ul\l,S\ul z^*)-B^z_{h+1}(\ul\l,\ul z^*)\big|\le C(\e')^2 2^{\th h}\label{quasifatta}.\ee
    The beta function $B^z$ is $O((\e')^2)$ because $N\ge 2$ in \eqref{6.75-0z}, so we have just to get the extra factor $2^{\th h}$.
    The left-hand side  can be rewritten as
    \be B^z_{h+1}(S\ul\l,S\ul z^*)-B^z_{h+1}(\ul\l,\ul z^*)=
    \int_0^1 dt\, \frac{d}{dt} B^z_{h+1}(\underline \l(t),\ul z^*(t)), \label{interp}\ee
    with $\underline \l(t):=\underline\l+t(S\underline \l-\underline \l)$, and similarly for $\ul z^*(t)$.  $B^z_{h+1}$  can be written in terms of its tree expansion, see \eqref{6.75-0z}, so that, 
    when the derivative w.r.t. $t$ acts on it, it can act on the factors $\l_{h'}(t)$ associated with the endpoints $v$ of the tree, or on the factors $z_{h'}^*(t)$ associated with the propagators and with the branches of the tree. 
    If it acts on an endpoint of type $\l$, whose value is $\l_{h'}(t)$, its effect is to replace it by $\l_{h'}-\l_{h'-1}$, which is bounded by $\e' 2^{\th h'}$, see \eqref{eq:boundlambda}; if it acts on a factor $z^*_{h'}(t)$, its effect is to multiply the tree value by 
    $z_{h'}^*(\ul \l)-z^*_{h'-1}(\ul\l)$, which is bounded by $C_0(\e')^2 2^{\th h'}$, thanks to the inductive hypothesis. Using these facts and the short memory property, \eqref{quasifatta} follows.
    Putting this together with \eqref{dajjje}, we get the desired bound on $z^*_{h-1}(\ul\l)-z^*_h(\ul\l)$. 

The proof of the second inequality \eqref{z*} is completely analogous: it can be proved inductively in $h$ (at the first step is trivially valid, again because $z^*_{-1}(\ul\l)\equiv 0$), by using the tree representation of the beta 
function, \eqref{6.75-0z}, and the short memory property. The details are left to the reader. 

\begin{Remark}\label{asymptz}
  The limiting value of $z^*_h(\ul\l)$ as $h\to-\infty$, which
  certainly exists, due to the first of \eqref{z*}, only depends on
  $\l_{-\infty}(\ul\l):=\l_0+\sum_{h\le 0}(\l_{h-1}-\l_h)$. {In order
  to prove this, notice that
  $\l_{-\infty}(\ul \l)-\l_h=\sum_{k\le h}(\l_{k-1}-\l_k)=O(\e' 2^{\th
    h})$, thanks to \eqref{eq:boundlambda}, and that
  $z^*_{-\infty}(\ul\l)-z^*_h(\ul \l)=O((\e')^2 2^{\th h})$, thanks to
  the first of \eqref{z*}: therefore, the error made by replacing in
  \eqref{6.75-0z} the sequence $\ul \l$ (resp. $\ul z^*$) by the
  constant sequence of elements $\l_{-\infty}$ (resp. $z^*_{-\infty}$)
  is of the order $O((\e')^22^{\th h})$. Consequently, letting
  $h\to-\infty$ in \eqref{bbetazzeta}, we find that $z^*_{-\infty}$ is
  the fixed point solution of an equation that only depends on
  $\l_{-\infty}$.}
\end{Remark}

\subsubsection{The solution of the flow equation as a fixed point.}\label{secfixedp}

Given $|\l_0|\le \e'$ and $\ul \l$ satisfying \eqref{eq:boundlambda}, we fix $\ul z=\ul z^*(\ul\l)$ as described in the previous subsection, and plug it into the flow equations for $\n_{h,\o},a_{h,\o},b_{h,\o},\l_h$: these are 
coupled equations, whose beta functions are thought of as functions of $\l_0$ and $\ul u$, with $\ul u$ as in \eqref{ulu}. 
In order to find the desired solution to these flow equations, we first note that the equations for $\n_{h,\o}, a_{h,\o}, b_{h,\o}$ in \eqref{beta2} imply that, for $k<h\le 0$,
\begin{eqnarray} 
&& \nu_{h,\o}=2^{k-h}\n_{k,\o}-\sum_{k<j\le h}2^{j-h-1}B^\n_{j,\o}(\l_0,\underline u),\label{nuho}\\
&& a_{h,\o}=a_{k,\o}-\sum_{k<j\le h}B^a_{j,\o}(\l_0,\underline u),\qquad b_{h,\o}=b_{k,\o}-\sum_{k<j\le h}B^b_{j,\o}(\l_0,\underline u).\nonumber
\end{eqnarray}
[Clearly, $B_{j,\o}^{\cdot}(\l_0,\ul u)$ actually depends only on the the components of $\ul u$ on scales larger than $j$.] 
If we send $k\to-\infty$ in  \eqref{nuho} and impose the desired condition on the exponential decay of $\nu_{h,\o},a_{h,\o},b_{h,\o}$, see \eqref{nubound}, we get $\nu_{h,\o}=-\sum_{j\le h}2^{j-h-1}B^\n_{j,\o}$, 
$a_{h,\o}=-\sum_{j\le h}B^a_{j,\o}$, and $b_{h,\o}=-\sum_{j\le h}B^b_{j,\o}$.

Regarding $\l_h$, we study its flow equation by extracting the first order contribution in $(\l_0,\ul u)$
from the beta function. By inspection, one verifies that the first order contribution does not depend on $\ul u$: therefore, we can write 
   \be B^\l_{h}(\l_0,\underline u)=c_h^\l \l_{0}+\tilde B^\l_h(\l_0,\underline u),\label{dec}\ee
   where $\tilde B^\l_h$ is at least of second order in $(\l_0,\ul u)$ and $c_h^\l$ can be computed in terms of first order perturbation theory. 
   Note that the GN trees that contribute to it have only a normal endpoint at scale $0$, of type $\mathcal R V^{(-1)}$. 
   Then, due to the short memory property, 
   \be \label{chz}|c_h^\l|\le \bar C\, 2^{\theta h}, \ee 
for some $\bar C>0$. By iterating the beta function equation for $\l_h$, we get: 
   \be \l_{h-1}=C_h^\l\l_0+\sum_{j=h}^0 \tilde B^\l_j(\l_0,\underline u),\ee
   where $C_h^\l=1+\sum_{j=h}^0c_j^\l$.

\medskip

In conclusion, given a sufficiently small $\l_{0}$, we look for initial data $\n_{0,\o}, a_{0,\o}$, $b_{0,\o}$, depending on $\l_{0}$, such that the corresponding flow satisfies, for all scales 
$h\le 0$, 
\begin{equation} 
\begin{cases} \nu_{h,\o}=-\sum_{j\le h}2^{j-h-1}B^\n_{j,\o}(\l_0,\underline u),\\
 a_{h,\o}=-\sum_{j\le h}B^a_{j,\o}(\l_0,\underline u),\\
b_{h,\o}=-\sum_{j\le h}B^b_{j,\o}(\l_0,\underline u)\\
 \l_{h-1}=C_h^\l \l_0+\sum_{j=h}^0\tilde B^\l_j(\l_0,\underline u),\end{cases}\label{eq:sys}
\end{equation}
% For all sufficiently small $\l_0$, t
with $\ul u$ satisfying \eqref{nubound}, \eqref{eq:boundlambda}. 
The system \eqref{eq:sys} will be viewed as a fixed point equation $\ul u=T(\ul u)$ for a map $T$ on the space 
of sequences 
  \be \label{norm} X_\e:=\{\underline u: \|(\underline\nu,\underline a,\underline b)\|_\th \le \e,\ \|\underline\l\|_\th\le \e'\}, \ee
see \eqref{nubound}, \eqref{eq:boundlambda}. In this equation, and from now on, we let $\e$ be sufficiently small, and we fix $\th=1/2$ and 
  \be \label{condK}\e'=\e/K,\quad K=\max\left\{1,\frac{C_1}{1-2^{-\th}}, \frac{2C'_1}{1-2^{-\th}}\right\},\ee 
  where $C_1,C_1'$ are the constants in \eqref{6.77} and \eqref{6.82} below, whose explicit values can be computed in terms of the first order contributions in $\ul\l$ 
  to $B^\n_{h,\o}$, $B^a_{h,\o}$, $B^b_{h,\o}$. 
\medskip

We now want to prove that $T$ is a contraction on $X_\e$, with respect to the metric $d(\ul u,\ul u'):=\|\ul u-\ul u'\|$, where
\be \|\ul u\|:= \max\{ \|(\underline\nu,\underline a,\underline b)\|_\th,\, K\sup_{h\le -1}|\l_h|\}.\label{linnorm}\ee
More precisely, we intend to prove that the image of $X_\e$ under the action of $T$ is contained in $X_\e$, and that $\| T(\ul u)-T(\ul u')\|\le (1/2)\, \|\ul u-\ul u'\|$ for all $\ul u,\ul u'\in X_\e$.
If this is the case, then $T$ admits a unique fixed point in $X_\e$, which corresponds to the desired initial data $\n_{0,\o}, a_{0,\o}, b_{0,\o}$, generating 
a flow satisfying conditions (1)-(2) discussed at the beginning of this section. 

\subsubsection{Invariance of $X_\e$ under the action of $T$.}\label{secinvarianceX}
In this subsection we show that $T(X_\e)\subseteq X_\e$, i.e. $\|T(\ul u)\|\le \e$ under the condition that 
\be |\l_{0}|\le \frac{\e}{2K}\min\{1,\bar C^{-1}\},\label{lkk}\ee
where $\bar C$ is the same as in \eqref{chz}. 
Note that, in order to prove that $T(X_\e)\subseteq X_\e$, it is enough to show that
%, if $\ul u\in\mathcal M_\e$, $\e$ is small enough and $\l_0$ satisfies \eqref{lkk}, then: 
\bea && |B^\n_{h,\o}(\l_0,\underline u)|, |B^a_{h,\o}(\l_0,\underline u)|, |B^b_{h,\o}(\l_0,\underline u)|\le C_1K^{-1}\e 2^{\th h},\label{6.77}\\
&& |\tilde B^\l_{h}(\l_0,\underline u)|\le C_2 \e^2 2^{\th h}, \label{6.78}\eea
for some $K$-independent constants $C_1,C_2$ (in order to see that \eqref{6.77}-\eqref{6.78} imply $\|T(\ul u)\|\le\e$, it is enough to plug them in the right side of \eqref{eq:sys} and use \eqref{condK} 
and \eqref{lkk}). 

\medskip

\noindent{\bf 1}. {\it The bound on $B^\n_{h,\o}(	\l_0,\underline u)$.}
We start by proving the bound on $B^\n_{h,\o}(	\l_0,\underline u)$ in \eqref{6.77}. Recall that the definition 
of $B^\n_{h,\o}$ is induced by the first of \eqref{eq:6.46}, combined with the first of \eqref{beta2}. As for the case of $B^z_h$ discussed in Section \ref{secz*}, $B^\n_h$ 
can be written as a sum over trees:
   \be B^\n_{h,\o}(\l_0,\underline u)=\sum_{N\ge 1}\sum_{\t\in \mathcal T^{(h)}_{N,0}} \sum_{{\bf P}\in\PP_\t}\sum_{T\in{\bf T}}B^\n_{h,\o}(\l_0,\underline u;\t,{\bf P},T),\label{6.75-0}\ee
   and $B^\n_{h,\o}(\l_0,\underline u;\t,{\bf P},T)$ can be bounded in a way analogous to \eqref{5.62}: 
   \bea |B^\n_{h,\o}(\l_0,\underline u;\t,{\bf P},T)|&\le& C^{N}\Big[\prod_{\substack{v\, {\rm e.p.}}} |F_v|\Big]\times\label{6.75}\\
&\times&\Big[\prod_{\substack{v\,{\rm not}\\ {\rm e.p.}}}\hskip-.02truecm\frac{C^{\sum_{i=1}^{s_v}|P_{v_i}|-|P_v|}}{s_v!} \;2^{\frac\e2{|P^\psi_v|}}2^{2-\frac12|P_v|-z(P_v)}
  \Big]\;.\nonumber\eea
Here, $|F_v|$ equals $|\nu_{h_{v'},+}|,|a_{h_{v'},+}|,|b_{h_{v'},+}|$ or $|\l_{h_{v'}}|$, if $v$ is of type $\n,a,b,$ or $\l$, 
see item (2) in the list of properties of trees in Section \ref{sectree};
  if $v$ is of type $\mathcal R V^{(-1)}$, then $F_v$ is a kernel of  $\mathcal R V^{(-1)}$ (the one associated with the given choice of $P_v$), and $|F_v|$ is its norm 
  \eqref{kappanorm} (more precisely, it is its un-weighted counter-part, i.e., the case $\k=0$), 
  which is bounded as in \eqref{boundkappanorm}. 
  
  We now split $B^\n_{h,\o}$ in a dominant plus a subdominant contribution, in the same spirit as the decomposition \eqref{Rs}: 
  $B^\n_{h,\o}=B^{\n,R}_{h,\o}+B^{\n,s}_{h,\o}$, where: 
  $B^{\n,R}_{h,\o}$ includes the sum over the trees whose endpoints are all of type $\l$ and all the single-scale propagators $g^{(k)}_\o/Z_{k-1}$ have been replaced by
  $g^{(k)}_{R,\o}/Z_{k-1}$, see \eqref{gRo}; $B^{\n,s}_{h,\o}$ is the remainder, which includes the sum over trees that have at least one endpoint of type $a,b,\n$ or $\mathcal R V^{(-1)}$ (the 
  scale $k$ of the endpoints of type $\l, a,b,\n$ satisfies $h<k\le 0$, while the scale of the endpoints of type $\mathcal R V^{(-1)}$ is necessarily $k=0$), or at least one  `remainder propagator' on some scale $k$ between $h$ and $0$, $(g^{(k)}_{\o}-g^{(k)}_{R,\o})/Z_{k-1}$.
%, whose dimensional bound is better by a factor $2^{\th k}$, as compared to the bound of $g^{(k)}_{R,\o}/Z_{k-1}$.

The key observation is that $B^{\n,R}_{h,\o}=0$: in fact the definition of $B^{\n,R}_{h,\o}$ is induced by the first of \eqref{eq:6.46}, with $\hat W^{(h),\infty}_{2,0;(\o,\o)}(0)$ replaced by 
$\hat W^{(h),R}_{2,0;(\o,\o)}(0)$ which is zero, as follows immediately from \eqref{eq:symm}.

The subdominant contribution, $B^{\n,s}_{h,\o}$, is not zero, but it is easy to bound. We further distinguish various contributions to it. (1) 
Let us start with the contributions from trees with at least two endpoints, one of which is of type $\n,a,b,\mathcal R V^{(-1)}$ and is on scale $k\in[h+1,0]$: these are bounded 
proportionally to $\e^2 2^{\th' (h-k)} 2^{\th k}$, where $\th'=3/4>\th$; the factor $2^{\th' (h-k)}$ is due to the short memory property (see the comment after \eqref{eq:zpv}), while the 
factor $\e 2^{\th k}$ comes from the norm $|F_v|$ associated with the endpoint of type $\n,a,b,\mathcal R V^{(-1)}$ on scale $k$, and the other $\e$ from the second endpoint.
Summing the bound over $k$ in $[h+1,0]$, we get  $const\times \e^2 2^{\th h}$, with the constant being independent of $K$. (2) Next, we consider the contributions from trees 
with exactly one endpoint, of type $\mathcal R V^{(-1)}$ (and, therefore, on scale $0$). Recalling that the norm of the value of the endpoint, $|F_v|$, is proportional to $|\l_0|\le \e/(2K)$, we find that the total 
contribution from these trees is $O(\e K^{-1} 2^{\th h})$, the factor $2^{\th h}$ coming from the short memory property, the proportionality factor being independent of $K$. 
(3) Finally, we are left with the contributions from trees whose endpoints are all of type $\l$ and at least one remainder propagator on some scale $k$ between $h$ and $0$, $(g^{(k)}_{\o}-g^{(k)}_{R,\o})/Z_{k-1}$. Recalling from Remark \ref{rem:gRo} that the dimensional bound of the remainder propagator is better by a factor $2^{\th k}$, as compared to the bound of $g^{(k)}_{R,\o}/Z_{k-1}$, we find that the 
these contributions are bounded by $const\times(\e/K)\sum_{k=h}^02^{\th'(h-k)}2^{\th k}\le const\times (\e/K)2^{\th h}$ (once again, the factor $2^{\th'(h-k)}$ is due to the 
short memory property, and the constant is independent of $K$). 
Putting things together, we obtain the desired estimate on $B^\n_{h,\o}$. 

\medskip

\noindent{\bf 2}. {\it The bound on $B^a_{h,\o}, B^b_{h,\o}$.} By definition, see \eqref{eq:6.46} and the definition of $\hat W^{(h),s}_{2,0;(\o,\o)}$ after \eqref{Rs}, 
  the trees contributing to $B^a_{h,\o}, B^b_{h,\o}$ either have an endpoint of type $\n,a,b,\mathcal R V^{(-1)}$, or their values contain a `remainder propagator' $(g^{(k)}_{\o}-g^{(k)}_{R,\o})/Z_{k-1}$ on some scale $k$ between $h$ and $0$. By proceeding as in the previous item, in particular in the discussion of the bound on $B^{\n,s}_{h,\o}$, 
  we find that $|B^a_{h,\o}|\le const\times(\e/K)\sum_{k=h}^02^{\th'(h-k)}2^{\th k}$, which is the desired estimate, and similarly for $|B^b_{h,\o}|$.

\medskip

\noindent{\bf 3}. {\it The bound on $\tilde B^\l_{h}$.} The fact that  $|\tilde B^\l_{h}|=O(\e^2)$ is obvious, because $\tilde B^\l_h$ is  a sum of trees with two or more endpoints, given that we have extracted the first-order contribution $c^\l_h \l_0$. The non-trivial issue is to show that the bound is proportional to $2^{\th h}$. For this purpose, we split $\tilde B^\l_h$ into a dominant and a subdominant part, following 
once again the same logic: we write $\tilde B^\l_h=B^{\l,R}_{h}+B^{\l,s}_{h}$, where:
$B^{\l,R}_{h}$  the sum over the trees whose endpoints are all of type $\l$ and all the single-scale propagators $g^{(k)}_\o/Z_{k-1}$ have been replaced by
$g^{(k)}_{R,\o}/Z_{k-1}$, see \eqref{gRo}; $B^{\l,s}_{h}$ is the remainder, which includes the sum over trees that have at least one endpoint of type $a,b,\n$ or $\mathcal R V^{(-1)}$, or at least one  
`remainder propagator' on some scale $k$ between $h$ and $0$, $(g^{(k)}_{\o}-g^{(k)}_{R,\o})/Z_{k-1}$.

The key observation is that the dominant term, $B^{\l,R}_{h}=B^{\l,R}_h(\ul\l)$ {\it is the same as the one of the reference model} discussed in Section \ref{sec:IR}: by this, we mean that $B^{\l,R}_{h}(\ul\l)$ 
is the same that we would get in the reference model, by applying the same multi-scale integration procedure. On the other hand, for the continuum model it is known that, if we denote by 
$\l^*\underline 1$ the constant sequence $(\l^*\underline 1)_h\equiv\l^*$, then 
\be \label{l*}|B^{\l,R}_{h}(\l^*\underline 1)|\le ({\rm const.})|\l^*|^2 2^{\th h},\ee
for $\l^*$ sufficiently small, see \cite[Theorem 3.1]{BM04}. Moreover, once the bound \eqref{l*} is known for the beta function computed on the constant sequence $\l^*\underline 1$, 
by using the short memory property, we find that the same bound holds for more general sequences: more precisely, we find that $B^{\l,R}_h(\ul\l)\le ({\rm const.})\e^2 2^{\th h}$, 
for any Cauchy sequence $\ul \l$ satisfying $\|\ul\l\|_\th\le \e$, as desired. 

We are left with the subdominant term, $B^{\l,s}_h(\ul\l)$, that, non surprisingly, can be bounded in a way similar to the subdominant contribution $B^{\n,s}_{h,\o}$; the result is, once again, $|B^{\l,s}_h(\ul\l)|\le ({\rm const.})\e^2 
2^{\th h}$ (details left to the reader). 

This concludes the proof of \eqref{6.77}-\eqref{6.78} and, therefore, that  $T(X_\e)\subset X_\e$.
  
\subsubsection{$T$ is a contraction on $X_\e$.}\label{seccontractionX}

We now show that $\|T(\ul u)-T(\ul u')\|\le (1/2) \|\ul u-\ul u'\|$, for all pairs of sequences $\ul u,\ul u'\in X_\e$ (here $\|\cdot\|$ is the norm in \eqref{linnorm}).
We consider the component at scale $h$ of $T(\underline u)-T(\underline u')$,
\begin{equation} \big[T(\underline u)-T(\underline u')\big]_h=\begin{cases}-\sum_{j\le h}2^{j-h-1}\big(B^\n_{j,\o}(\l_0,\underline u)-B^\n_{j,\o}(\l_0,\underline u')\big)\\
-\sum_{j\le h}\big(B^a_{j,\o}(\l_0,\underline u)-B^a_{j,\o}(\l_0,\underline u')\big)\\
-\sum_{j\le h}\big(B^b_{j,\o}(\l_0,\underline u)-B^b_{j,\o}(\l_0,\underline u')\big)\\
\sum_{j\ge h}\big(\tilde B^\l_{j}(\l_0,\underline u)-\tilde B^\l_{j}(\l_0,\underline u')\big).\end{cases}\label{eq:sysprime}
\end{equation}
In order to prove that $T$ is a contraction, it is enough to show that, if $\underline u,\underline u'\in X_\e$ and $\l_0$ satisfies \eqref{lkk}, 
  then the analogues of \eqref{6.77}-\eqref{6.78} hold, namely:
  \bea && |B^\#_{h,\o}(\l_0,\underline u)-B^\#_{h,\o}(\l_0,\underline u')|\le C_1'K^{-1} \|\underline u-\underline u'\|2^{\th h}, \quad {\rm for}\quad \#=\n,a,b, \nonumber\\
&& |\tilde B^\l_{h}(\l_0,\underline u)-\tilde B^\l_{h}(\l_0,\underline u')|\le C_2'\e \|\underline u-\underline u'\| 2^{\th h}\,,\label{6.82}\eea
   for some $K$-independent constants $C_1',C_2'$%  (the reader can easily check that, by plugging 
% \eqref{6.82} in the right side of \eqref{eq:sysprime} and by using \eqref{condK}, one immediately finds that $\| T(\underline u)-T(\underline u')\|\le \frac12\|\underline u-\underline u'\|$)
   .

   The proof of \eqref{6.82} is very similar to the one of \eqref{6.77}-\eqref{6.78}: in order to illustrate the main ideas, let us focus on $B^\n_{h,\o}(\underline u)-B^\n_{h,\o}(\underline u')$, the other 
   components being treatable in a similar manner. By using \eqref{6.75-0}, we rewrite the difference under consideration as a sum over trees: 
   \bea && B^\n_{h,\o}(\l_0,\underline u)-B^\n_{h,\o}(\l_0,\underline u')=\label{6.75-1}\\
&&=\sum_{N_n\ge 1}\sum_{\t\in \mathcal T^{(h)}_{N_n,0}} \sum_{{\bf P}\in\PP_\t}\sum_{T\in{\bf T}}\big(B^\n_{h,\o}(\l_0,\underline u;\t,{\bf P},T)-
   B^\n_{h,\o}(\l_0,\underline u';\t,{\bf P},T)\big).\nonumber\eea
   We further rewrite the difference in parentheses in the right side in a way similar to \eqref{interp}: 
   \bea && B^\n_{h,\o}(\l_0,\underline u;\t,{\bf P},T)-
           B^\n_{h,\o}(\l_0,\underline u';\t,{\bf P},T)=\\
&&\hskip4.truecm=\int_0^1 dt\, \frac{d}{dt} B^\n_{h,\o}(\l_0,\underline u(t);\t,{\bf P},T),\nonumber\eea
   with $\underline u(t):=\underline u'+t(\underline u-\underline u')$. When the derivative w.r.t. $t$ acts on $B^\n_{h,\o}(\underline u(t);$ $ t,{\bf P},T)$, it can act on  the modified 
   running coupling constants $\n_{h',\o}(t)$, $a_{h',\o}(t)$, $b_{h',\o}(t)$, $\l_{h'}(t)$ associated with the endpoints $v$ of the tree, or on the modified constants
   $z_{h'}^*(\ul\l(t))$ associated with the propagators and with the branches of the tree. If, e.g., it acts on an endpoint $v$ of type $\n$, which is associated with $\n_{h',\o}(t)$, 
   its effect is to replace it by $\n_{h',\o}-\n_{h',\o}'$; when bounding the norm of the tree value, the endpoint $v$ comes with a factor $|\n_{h',\o}-\n_{h',\o}'|$, which leads 
   to a factor $\|\underline u-\underline u'\|$; this has to be compared with the `standard' factor $|\n_{h',\o}|$ appearing in the bound of the un-modified tree value, which led to a factor 
   $\|\underline u\|\le \e$ in (one of the contributions to) the first of \eqref{6.77}: therefore, the bound on 
   $\frac{d}{dt} B^\n_{h,\o}(\underline u(t);\t,{\bf P},T)$ is qualitatively the same as the one on $B^\n_{h,\o}(\underline u;\t,{\bf P},T)$, up to an additional factor $\|\underline u-\underline u'\| / \e$. 
   The terms in which the derivative w.r.t. $t$ acts on other RCCs, or on $z_{h'}^*(\ul(t))$ are treated similarly. In light of these considerations,
   recalling the bound $|B^\n_{h,\o}(\underline u)|\le C_1K^{-1}\e 2^{\th h}$ on the un-modified $\n$-component of the beta function, we obtain the bound in the first line of \eqref{6.82}
   with $\#=\n$. The other components are treated similarly, but we do not belabor further details here.  

\medskip

This concludes the proof that the map $T$ defined by \eqref{eq:sys} is a contraction on $X_\e$: therefore, it admits a unique fixed point $\underline{\mathfrak{u}}$ on $X_\e$, whose 
components at $h=0$ correspond to the initial data generating a flow that satisfies the conditions (1) and (2) given at the beginning of Section \ref{secfl}. 

\subsubsection{Analyticity of the fixed point sequence and inversion of the counterterms.}\label{seccount}

   Thanks to the convergence of the tree expansion for the components of the beta function, the components of $\underline{\mathfrak{u}}$, and, in particular, those at $h=0$, are all real analytic functions of $\l_{0}=\l$, in the ball \eqref{lkk}. 
We write: 
\be \label{id}\n_{0,\o}=f_{\n;\o}(\l),\quad a_{0,\o}=f_{a;\o}(\l),\quad b_{0,\o}=f_{b;\o}(\l).\ee %,\quad \l_{0}=F(\l_{-\infty}).\ee
From now on, with some abuse of notation, we denote by $\n_{h,\o}=\n_{h,\o}(\l)$, $a_{h,\o}=a_{h,\o}(\l)$, $b_{h,\o}=b_{h,\o}(\l)$, $\l_h=\l_h(\l)$, $z_h=z_h(\l)\equiv z^*_h(\ul\l(\l))$ 
the components of the fixed point sequence, thought of as  functions of $\l_0=\l$. %, instead of $\l_{-\infty}$.
Recalling that $\n_{0,\o}(\l)=-\m(\bar p^\o)$, from the first equation in \eqref{id} we calculate $\bar p^\o=\bar p^\o(\l)$ (via the implicit function theorem; recall that $\alpha_\o=\partial_{p_1}\m(p^\o)\neq 0,\beta_\o=\partial_{p_2}\m(p^\o)\neq 0$ and that $\alpha_\o/\beta_\o\not\in\mathbb R$), and find that $\bar p^\o(\l)=p^\o+O(\l)$. 
Finally, recalling that $\bar \a_\o,\bar \b_\o$ are related to $a_\o=a_{0,\o}(\l)$, $b_\o=b_{0,\o}(\l)$ via \eqref{eq:6.4}-\eqref{eq:6.5}, we find that $\bar\a_\o=\bar\a_\o(\l)=\a_\o+O(\l)$ and 
$\bar\b_\o=\bar\b_\o(\l)=\b_\o+O(\l)$, as desired. The functions $\bar p^\o(\l), \bar\alpha_\o(\lambda), \bar\beta_\o(\lambda)$ are those of items (1)-(2) of Proposition \ref{prop:comp}.

\subsubsection{The flow of $Z_h$ and its critical exponent $\eta$.}\label{secflowZ}

Once the initial data are fixed as in \eqref{id} and the corresponding flow of RCC is bounded and exponentially convergent, we
immediately find that 
\be Z_h=\prod_{k=h+1}^{0}(1+z_k)=: (1+z_{-\infty}(\l))^{-h}A_h,\label{eq:6.83}\ee
where $A_h=1+O(\l^2)$ and  $A_h=A_{-\infty}(1+O(\l^2 2^{\th h}))$, as easily follows from \eqref{z*}. Note that, while $z_{-\infty}$ depends only on $\l_{-\infty}$, $A_{-\infty}$ depends on the whole sequence. %  (the fact that in these formulas we have $O(\l^2)$ and $O(\l^2 2^{\th h})$ instead of  $O(\l)$ and $O(\l 2^{\th h})$
% follows from the fact that $B^z_h$ is at least of second order in $\l$, recall the condition $N\ge 2$ in \eqref{6.75-0z}).
For future reference, we let $\h=\h(\l)=\log_2(1+z_{-\infty}(\l))$ be the so-called {\it critical exponent of the wave function renormalization}. Then,
\be Z_h=2^{-\h h}A_h=A_{-\infty} 2^{-\h h}(1+O(\l^2\,2^{\th h})).\label{eq:6.84}\ee

\subsubsection{Fixing the bare coupling $\l_{\infty}$ of the reference model.}\label{remmm} 

The critical exponent $\h(\l)$ only depends on the asymptotic value of $z_h$ as $h\to -\infty$ that, in turn, only depends on $\l_{-\infty}(\l)$, see Remark \ref{asymptz}. 
  Recall that, by its very definition, the flow equation of $z_h$
  involves a beta function expressed in terms of
  $\hat W^{(h),R}_{2,0;(\o,\o)}$ and, therefore, it is the same as we
  would get in a multiscale expansion of the reference model of
  Section \ref{sec:IR}: as a consequence, the critical exponent
  $\h(\l)$ is the {\it same} as the one of the reference model,
  $\h_R(\l_\infty)$, provided that the bare coupling $\l_\infty$ of
  the reference model is fixed in such a way that the infrared limit
  $\l_{-\infty;R}=\l_{-\infty;R}(\l_\infty)$ of its coupling equals
  that of the dimer model, \be
  \l_{-\infty;R}(\l_\infty)=\l_{-\infty}(\l).\label{linftyl-infty}\ee
  This equation is analytically invertible w.r.t. $\l_\infty$, as one
  can show by repeating the study of the flow of $\l_h$ for the
  reference model:  in that case, $\l_{h;R}$ satisfies the
  analogue of the fourth equation in \eqref{eq:sys}, which reads, {for $h<0$,}
  $\l_{h;R}=\l_\infty+\sum_{j= h}^0B_{h,R}^\l(\l_\infty,\ul u_{R})$,
  where $B_{h,R}^\l$ is given by a convergent tree
  expansion, and satisfies
  $|B_{h,R}^\l(\l_\infty,\ul u_{R})|\le ({\rm const.})|\l_\infty|^2
  2^{\th h}$. With respect to the dimer model (cf. \eqref{dec}), note that there is no linear term in the beta function of $\l$: this is because the interaction potential of the
  continuum model is exactly quartic in the Grassmann fields.  From this, one finds that
  $\l_{-\infty}=\l_{\infty}+f_{\l,R}(\l_{\infty})$, where $f_{\l,R}$
  is analytic in $\l_\infty$ and of second order in $\l_\infty$; in
  particular, $\l_{-\infty;R}(\l_\infty)$ is analytically invertible
  with respect to $\l_\infty$.  In conclusion, \eqref{linftyl-infty}
  can be inverted into
  $\l_\infty=\l_{-\infty;R}^{-1}\big(\l_{-\infty}(\l)\big){=O(\l)}$; this
  choice guarantees that the asymptotic couplings as $h\to-\infty$ of the
  dimer and reference models are the same. The function $\lambda_\infty(\l)$ constructed here is the one of item (3) of Proposition \ref{prop:comp}.
  Finally, by inspection of
  second order perturbation theory, it turns out \cite[Th. 2]{BM02} that
  $\h_R(\l_\infty)=a\l^2_{\infty}+O(\l_\infty^3)$, for a suitable
  $a>0$. Therefore, by fixing $\l_\infty$ as in \eqref{linftyl-infty}
  and recalling that
  $\l_{-\infty;R}(\l_\infty)=\l_\infty+O(\l_\infty^2)$, we find
  $\h(\l)=a[\l_{-\infty}(\l)]^2+ O(\l^3)$.

\subsubsection{Fixing the bare wave function renormalization $Z$ of the reference model.}\label{bareZ}

Once  $\l_\infty$ is fixed as in the previous subsection, the flow of the wave function renormalization $Z_{h;R}$ of the reference model 
is given by the analogue of \eqref{eq:6.84}, with the {\it same} critical exponent $\eta$: if $h<0$, 
\be Z_{h;R}=2^{-\h h}A_{h,R}=A_{-\infty,R} 2^{-\h h}(1+O(\l^2\,2^{\th h})),\label{eq:6.84RR}\ee
   for some $A_{-\infty,R}=Z(1+O(\l^2))$. Actually, by a trivial rescaling of the Grassmann fields in \eqref{vv1}, one sees that $ Z_{h;R}$ is proportional to $Z$, so that $A_{-\infty,R}=Z\tilde A_{-\infty,R}$ with $\tilde A_{-\infty,R}=(1+O(\l^2))$ independent of $Z$. We now fix $Z=A_{-\infty}/\tilde A_{-\infty,R}=1+O(\l^2)$,
   which guarantees that $\lim_{h\to-\infty}Z_{h;R}/Z_h=1$. The function $Z(\l)$ constructed here is the one of item (3) of Proposition \ref{prop:comp}.

   \subsubsection{The flow of $Y_{h,r,\ul\o}$.}\label{flowY}
     On scale $-1$, one sees by direct inspection of the non-interacting dimer model that  $Y_{-1,r,(\o_1,\o_2)}:=-K_re^{-i\bar p^{\o_2}\cdot v_r}+O(\l)$.
   Once the fixed point sequence $\underline{\mathfrak{u}}$ has been determined, we can plug it into the beta function equation for $Y_{h,r,\ul\o}$,
   \be Y_{h-1,r,\ul\o}=Y_{h,r,\ul\o}+B^Y_{h,r,\ul\o}(\underline{\mathfrak{u}}, \ul Y_r),\qquad h\le -1,\label{BY}\ee
   where $\ul Y_r=(Y_{h,r,\ul\o})_{h\le -1,\, \ul\o\in\{\pm\}^2}$.
   Note that, by definition, $B^Y_{h,r,\ul\o}$ is linear in $\ul Y_r$. 
   This equation can be solved iteratively in $h$, via a procedure analogous to the one used to compute $\ul z$ given $\ul\l$, see Section \ref{secz*}. In particular,
   $B^Y_{h,r,\ul\o}$ admits a tree expansion, by using which \eqref{BY} can be rewritten as
   \be Y_{h-1,r,\ul\o}=Y_{h,r,\ul\o}+\sum_{k=h}^{-1}\sum_{\ul \o'}B^{Y,R}_{k,h;\ul\o,\ul\o'}(\underline{\mathfrak{u}})Y_{k,r,\ul\o'}
   +\sum_{k=h}^{-1}\sum_{\ul \o'}B^{Y,s}_{k,h;r,\ul\o,\ul\o'}(\underline{\mathfrak{u}})Y_{k,r,\ul\o'},\label{BY.1}\ee
   where: $B^{Y,R}_{h,k;\ul\o,\ul\o'}$ is the relativistic contribution, i.e., it is expressed as a sum over trees whose endpoints are all of type $\l$ and all the propagators have been replaced by relativistic ones 
   (it is easy to check, by inspection, that $B^{Y,R}_{h,k;\ul\o,\ul\o'}$ is independent of $r$), and $B^{Y,s}_{h,k;r,\ul\o,\ul\o'}$ is the remainder.
   Thanks to the short memory property, and the known bounds on the components of the fixed point sequence $\underline{\mathfrak{u}}$, we find that 
   \be |B^{Y,R}_{k,h;\ul\o,\ul\o'}(\underline{\mathfrak{u}})|\le C|\l| 2^{\th'(h-k)}, \quad |B_{k,h;r,\ul\o,\ul\o'}^{Y,s}(\underline{\mathfrak{u}})|\le C|\l| 2^{\th'h}.\ee
 
   We now let $y_{h,r,\ul\o}=Y_{h-1,r,\ul\o}/Y_{h,r,\ul\o}-1$, and iteratively compute $y_{h,r,\ul\o}$ for $h\le -1$ from \eqref{BY.1}, starting from $h=-1$. 
   Proceeding by induction, as in Section \ref{secz*}, we find that $y_{h,r,\ul\o}$ is a Cauchy sequence, whose limiting value as $h\to-\infty$, $y_{-\infty,\ul\o}(\l)$, only depends on $\l_{-\infty}(\l)$, see Remark \ref{asymptz}. 
   This limiting value defines new critical exponents, $\h_{\ul\o}(\l):=\log_2(1+y_{-\infty,\ul\o}(\l))$. 
   By using the same considerations {as in Section \ref{remmm}}, we conclude that $\h_{\ul\o}(\l)$ are the same as the corresponding exponents in the continuum model, provided that 
   the bare coupling $\lambda_\infty$ is fixed in such a way that $\l_{-\infty,R}(\l_\infty)=\l_{-\infty}(\l)$. Thanks to the symmetries of the reference model, it is known that $\h_{(\o_1,\o_2)}(\l)$ are real and only depend on
   the product $\o_1\o_2$; we denote by $\h_1(\l)$, resp. $\h_2(\l)$, the critical exponent with $\o_1=-\o_2$, resp. $\o_1=\o_2$. Remarkably, it is known also that $\h_2(\l)=\h(\l)$,  see \cite[Theorem 3]{BM02}. 
   On the other hand, an explicit computation shows that $\h_1(\l)=b\l_{-\infty}(\l)+O([\l_{-\infty}(\l)]^2)$, for a suitable $b\neq 0$,
   so that in particular $\h_1(\l)\neq \eta(\l)$ (recall that $\h(\l)=a[\l_{-\infty}(\l)]^2+
   O(\l^3)$, as discussed {in Section \ref{remmm}).}
   In terms of these critical exponents, we can rewrite $Y_{h,r,\ul\o}$ in a way analogous to \eqref{eq:6.84}
   \bea Y_{h,r,(\o,\o)}&=&2^{-\h h}B_{h,r,\o}=2^{-\h h}B_{-\infty,r,\o} (1+O(\l\,2^{\th h})),\label{eq:6.86}\\
  Y_{h,r,(\o,-\o)}&=&2^{-\h_1 h}C_{h,r,\o}= 2^{-\h_1 h}C_{-\infty,r,\o}(1+O(\l\,2^{\th h})),\nonumber\eea
                      for suitable complex constants $B_{h,r,\o}$, $C_{h,r,\o}$, such that $B_{h,r,-\o}=B_{h,r,\o}^*$ and $C_{h,r,-\o}=C_{h,r,\o}^*$.

The critical exponent $\nu$ of Theorems \ref{th:1} and \ref{th:2} is given in terms of $\h(\l),\h_1(\l)$  by the simple relation
                      \begin{eqnarray}
                        \label{eq:nuetaeta}
                        \nu(\l)=1
+\eta(\l)-\eta_1(\l).                      \end{eqnarray}
                      
                      \subsection{Thermodynamic limit for the correlation functions}
\label{sec:ultimiquattro}
                    % Putting together the tree expansion of the effective potential with the bounds on the tree values of Proposition \ref{prop:an} and  the results of the previous section on the flow of the RCC,
                    In the previous sections, we have obtained a convergent expansion for the effective potentials, valid for $|\l|$ small enough and a suitable choice of $\bar p^\o=\bar p^\o(\l)$, $\bar \a_\o=\bar \a_\o(\l)$, $\bar \b_\o=\bar\b_\o(\l)$.
                    In particular, after the integration of all the scales $h> h_L$  we obtain\footnote{Recall that $h_L$ is the first scale at which the support of $\bar\chi_h$ has empty intersection with $\mathcal P'_\o$, so that (see \eqref{eq:gleh}) at scale $h_L$ one can remove in \eqref{eq:vh} the integration w.r.t. $P_{(\le h_L)}$ and just replace $\psi$ with $0$. \label{foto:hL} Recall also that $\mu_0(\cdot)$ was defined in \eqref{mu0}, and $J=J(A)$  just before Eq. \eqref{eq:tb0}.}  from \eqref{eq:vh} with $h=h_L$
\bea  \mathbb W_L^{(\bt)}(A,0,\Psi)&=&L^{-2}\sum_\o[\m_0(k^\o_\bt)-Z_{h_L}\m_{\bt,\o}]
\hat \Psi^+_{\o}\hat \Psi^-_{\o} \nonumber\\
&+&L^2 E^{(h_L)}+S^{(h_L)}(J)+V^{(h_L)}(\sqrt{Z_{h_L}}\Psi,J),\label{eq:vhL}\eea
                    where we defined 
$$\m_{\bt,\o}:=\m_{h_L,\o}(k^\o_\bt-\bar p^\o)= \bar D_\o(k^\o_\bt-\bar p^\o)+r_\o(k^\o_\bt-\bar p^\o)/Z_{h_L}$$
%[note also that  $\m_{\bt,\o}=\m_0(k^\o_\bt)-r_\o(k^\o_\bt-\bar p^\o)(1-1/Z_{h_L})$],
    and $E^{(h_L)}$, $S^{(h_L)}(J)$ and $V^{(h_L)}(\Psi,J)$ are given by the convergent tree expansion discussed above. 
    In order to obtain the Grassmann generating function with $\bt$ boundary conditions, 
    $\mathcal W_L^{(\bt)}(A,0)$, we need to integrate out $\Psi$, see \eqref{eq:Wanophi}; finally, the dimer generating function is obtained by taking a linear combination of $e^{\mathcal W_L^{(\bt)}(A,0)}$ {with all values of $\bt$}, see \eqref{eq:BP}. Using \eqref{eq:vhL} we write: 
    \bea\label{eq:Wnonophi}
    e^{\mathcal W_L^{(\bt)}(A,0)}&=&e^{L^2 E^{(h_L)}+S^{(h_L)}(J)}\times\\
&\times &\int D\hat \Psi e^{-L^{-2}Z_{h_L}\sum_{\o}\mu_{\bt,\o}\hat\Psi^+_{\o}\hat\Psi^-_{\o}+V^{(h_L)}(\sqrt{Z_{h_L}}\Psi,J)}.\nonumber\eea %end{equation}
          In order to study the thermodynamic limit for correlations, it is important to characterize how $E^{(h_L)}$, $S^{(h_L)}(J)$ and $V^{(h_L)}(\Psi,J)$ 
          depend on the system size $L$ and on the boundary conditions $\bt$. For illustrative purposes, we start by considering the case $A=J=0$, in which case the generating function reduces to the partition function. 
          As shown in Appendix \ref{app:grigliabis}, $Z_{\bt}:= e^{\mathcal W_L^{(\bt)}(0,0)}$ can be rewritten as
%As shown in Appendix \ref{anotherapp}, $e^{L^2 E^{(h_L)}}$ can be written as  $$e^{L^2 E^{(h_L)}} = 
%Letting  $Z_{\bt}:=e^{\mathcal W_L^{(\bt)}(0,0)}$, we get 
\begin{eqnarray} 
Z_{\bt}&=&\Big[\prod_{k\in \mathcal P'(\bt)}\mu_0(k)\Big]e^{L^2 \Delta(\lambda)} (1+s_\bt(\l))\times\label{yaa}\\
&\times& \frac1{Z_{h_L}^2}\int D\hat \Psi e^{-L^{-2}Z_{h_L}\sum_{\o}\mu_{\bt,\o}\hat\Psi^+_{\o}\hat\Psi^-_{\o}+V^{(h_L)}(\sqrt{Z_{h_L}}\Psi,0)},
%\int \mathcal D\hat \Psi e^{-L^{-2}\sum_{\o=\pm}\mu_0(k_\bt^\o)\hat\Psi^+_{k_\bt^\o}\hat\Psi^-_{k_\bt^\o}+V^{(h_L)}(\Psi,0)}
\nonumber\eea
where: $\Delta$ is analytic in $\l$, independent of $L$ and of the boundary conditions; $s_\bt(\l)$ depends on $L,\bt$ and is of order $O(\lambda)$, uniformly in $L,\bt$;
\be V^{(h_L)}(\Psi,0)= L^{-3}\sum_\o u_{\bt,\o}(\l)\hat\Psi^+_{\o}\hat\Psi^-_{\o}+L^{-6}v_\bt(\l)\hat \Psi^+_{+}\hat \Psi^-_{+}\hat \Psi^+_{-}\hat \Psi^-_{-},\label{eq:692}\ee
         with $u_{\bt,\o}(\l),v_\bt(\l)$ of order $O(\lambda)$, uniformly in $L,\bt$. From now on, for lightness of notation, we drop the argument $\l$ in $u_{\bt,\o}(\l),v_\bt(\l),s_\bt(\l)$. The integration of $\Psi$ is elementary, and gives (recall \eqref{crollo})
\begin{eqnarray} 
  Z_{\bt}=\Big[\prod_{k\in \mathcal P'(\bt)}\mu_0(k)\Big]e^{L^2 \Delta(\lambda)} (1+s_\bt)\label{yae}\Big[\prod_{\o=\pm}(-\mu_{\bt,\o}+\frac{u_{\bt,\o}}L)+\frac{v_\bt}{L^2}\Big],%\nonumber
\end{eqnarray}
or, equivalently, 
\begin{equation} 
Z_{\bt}=e^{L^2 \Delta(\lambda)} (1+s_\bt)\, \big( Z^0_{\bt}+\tilde Z^0_\bt\, L^{-2} \s_\bt\big)\label{yea}\ee
%&+& \tilde Z^0_{\bt}\cdot \big(-L^{-1}\sum_{\o=\pm}u_{\bt,\o}\,\m_0(k_\bt^{-\o})+L^{-2}(u_{\bt,+}u_{\bt,-}+v_\bt)\big)\Big],\nonumber\end{eqnarray}
where we defined 
\bea && \tilde Z^0_{\bt}=\prod_{k\in \mathcal P'(\bt)}\mu_0(k),\qquad Z^0_{\bt}=\m_{\bt,+}\,\m_{\bt,-}\,\tilde Z^0_\bt \label{tildeZ0},\\
%[note that the first product is over all the momenta in $\mathcal P(\bt)$, while the second is over the momenta in $\mathcal P'(\bt)=\mathcal P(\bt)\setminus\{k_\bt^+,k_\bt^-\}$], 
&& \s_\bt=-L\sum_{\o=\pm}u_{\bt,\o}\,\m_{\bt,-\o}+u_{\bt,+}u_{\bt,-}+v_\bt.\eea
We now let $\bt^0$ be the boundary conditions for which $k_{\boldsymbol\th}^\o$ is at the largest distance from $\bar p^\o$;  if $L$ is large enough,
\be |\m_{\bt^0,\o}|\ge (1/2)\, |\m_{\bt,\o}|,\quad \forall\bt\in\{0,1\}^2\label{gr1}\ee 
and 
\be c_-^{-1}/L\le |\m_{\bt^0,\o}|\le c_-/L,\ee
for a suitable $L$-independent constant $c_-$. Moreover, 
\be c_+^{-1}\le |{\tilde Z^0_\bt}/{\tilde Z^0_{\bt'}}|\le c_+,\label{griglia}\ee
for a suitable $L$-independent constant $c_+$, for all choices of $\bt,\bt'$ (see Appendix \ref{app:griglia}). We now multiply and divide the term $\tilde Z^0_\bt\, L^{-2} \s_\bt$ in 
\eqref{yea} by $Z^0_{\bt^0}$ and rewrite it as 
\be \tilde Z^0_\bt\, L^{-2} \s_\bt= Z^0_{\bt^0}\frac{\tilde Z^0_\bt}{\tilde Z^0_{\bt^0}}\frac{\s_{\bt}}{L^2\mu_{\bt^0,+}\,\mu_{\bt^0,-}}=: Z^0_{\bt^0}\s_{\bt,\bt^0}. \ee
%  \be Z_{\bt}=e^{L^2 \Delta(\lambda)} (1+s_\bt(\l))\cdot \Big[Z^0_{\bt}+Z^0_{\bt^0}\,\s_{\bt,\bt^0}(\l)\Big].\label{Zbt}\ee
By using 
\eqref{gr1}--\eqref{griglia}, we immediately conclude that $\s_{\bt,\bt^0}=O(\l)$, uniformly in $L,\bt$. If we now take the appropriate linear combination of $Z_\bt$, we obtain the partition function
of the interacting dimer model that, in light of the previous considerations, can be written as
\be\label{zlll} Z_L =\frac12\sum_{\bt} c_\bt Z_\bt=\frac{e^{L^2 \Delta(\lambda)} }2\sum_{\bt} (1+s_\bt) \Big[c_\bt Z^0_{\bt}+Z^0_{\bt^0}c_\bt\s_{\bt,\bt^0}\Big].\ee
We now let  $Q^0_L=\frac12\sum_{\bt}c_{\bt}Z^0_{\bt}$; we recall that the constants $c_\bt$ are either $1$ or $-1$, depending on $\bt$ and on the parity of $L/2$, see the definition after \eqref{eq:2}. A simple computation shows that
\begin{eqnarray}
  \label{eq:infradito}
c_{\bt}Z^0_{\bt}=|Z^0_{\bt}| \quad \text{for all} \quad \bt  ,
\end{eqnarray}
 see Appendix \ref{app:infradito}. Therefore, $Q^0_L=\frac12\sum_{\bt}|Z^0_{\bt}|$, so that 
\be \frac12\max_\bt|Z^0_{\bt} |\le Q^0_L\le 2\max_{\bt}|Z^0_{\bt}|.\ee
If we use these inequalities in \eqref{zlll}, we get  
\begin{equation} Z_L=e^{L^2\D(\l)}Q^0_L (1+r_L(\l))\label{eq:gfp},\end{equation}
where the error term $r_L(\l)$ is of order $O(\l)$, uniformly in $L$. 

\medskip

Let us now adapt the previous discussion to $e^{\mathcal W_L^{(\bt)}(A,0)}$, in the presence of the external field $A$. In this case, the analog of \eqref{yae}  
%(the proof being based on a similar reasoning, see also Appendices \ref{dimr} and \ref{app:grigliabis}) 
is
\begin{eqnarray} 
&& e^{\mathcal W_L^{(\bt)}(A,0)}=\tilde Z^0_\bt e^{L^2 \Delta+S_L(J)+\mathcal S_\bt(J)} (1+s_\bt)\times\nonumber\\
&&\times \frac1{Z_{h_L}^2}\int  D\hat \Psi e^{-L^{-2}Z_{h_L}\sum_{\o}\mu_{\bt,\o}\hat\Psi^+_{\o}\hat\Psi^-_{\o}+V^{(h_L)}(\sqrt{Z_{h_L}}\Psi,J)}.%+\mathcal B^{(h_L)}(\Psi,J)},
\label{yaaA}\end{eqnarray}
Let us define the various functions involved, and let us prove a number of properties that they satisfy.
In the first line, $\tilde Z^0_\bt$ was defined in \eqref{tildeZ0}, $\D=\D(\l)$, $s_\bt=s_\bt(\l)$ are the same as in \eqref{yaa}.  Moreover, $S_L(J)+\mathcal S_\bt (J)$
is a rewriting of $S^{(h_L)}(J)\equiv \sum_{h_L\le h<0}\tilde S^{(h)}(J)$ (recall that $\tilde S^{(h)}$ is the single scale contribution to the generating function, see item (1) in Sect.\ref{sectree} and Eq.\eqref{eq:6.63e}):
\begin{equation}\label{recallgenfun} S_L(J)+\mathcal S_\bt (J)=\sum_{h_L\le h<0}\tilde S^{(h)}(J).\end{equation} 
In this equation, $S_L(J)$ is independent of boundary conditions
and corresponds to the `bulk' contribution, i.e., the dominant one in the thermodynamic limit; more precisely,  
$S_L(J)$ is {\it defined} as
\begin{eqnarray}
S_L(J)&=&\sum_{m\ge 1}\sum_{\ul r\in\{1,\ldots,4\}^m} \sum_{\ul y\in\L^m}\,J_{y_1,r_1}\cdots J_{y_m,r_m} \times \label{definSL}\\
&\times &\sum_{h_L\le h\le 0} \sum_{n_2,\ldots,n_m\in\mathbb Z^2}W^{(h),\infty}_{0,m;\ul r}(y_1,y_2+n_2L,\ldots,y_m+n_mL), \nonumber\eea
where $W^{(h),\infty}_{0,m;\ul r}(\ul y)$ is the infinite volume limit of the kernel $W^{(h)}_{0,m;\ul r}(\ul y)$ of the single-scale contribution to the generating function, 
$\tilde S^{(h)}(J)$, defined after \eqref{eq:6.65}, 
which satisfies the bound \eqref{bouW} with $n=0$; in particular, recalling \eqref{eq:6.84} and \eqref{eq:6.86}, which imply that 
$\max_{h'\ge h}\frac{|Y_{h',\cdot}|}{|Z_{h'}|}\le C2^{-h [\eta_1-\eta]_+}\le C 2^{-hC|\l|}$, we find
\be \|W^{(h),\infty}_{0,m}\|_{\k,h}\le C^m 2^{h(2-m)}2^{-h C|\l|m},\label{weightL1n0} \ee
for some $C,\k>0$. 
We recall that \eqref{bouW} (and, therefore, \eqref{weightL1n0}) is a consequence of the {GN tree expansion for $W^{(h),\infty}_{0,m;\ul r}$, which reads
\be W^{(h),\infty}_{0,m;\ul r}(\ul y)=\sum_{N_n\ge 0,\, N_s \ge 1}\sum_{\t\in \mathcal T^{(h)}_{N_n,N_s}} 
\sum_{
{\bf P}\in\PP_\t}\sum_{T\in{\bf T}}\sum_{ \xx_{v_0}} \mathds 1(\xx_{v_0}=\ul y) W_{\t,{\bf P},{T}}(\xx_{v_0})\;,\label{cor.6.131}\ee
and of the weighted $L_1$ bound on $W_{\t,{\bf P},{T}}(\xx_{v_0})$ stated in 
Proposition \ref{prop:an}. A generalization of this weighted $L_1$ bound leads to the following pointwise estimate, valid under the same assumptions as
Proposition \ref{prop:an}, for $\xx_{v_0}=\ul y$:
\bea && \big|W_{\t,{\bf P},T}(\xx_{v_0})\big|\le C^{N_s}
\,(C\e)^{\max\{N_n,c|I_{v_0}^\psi|\}}\,
2^{h(2-m)}\,\Biggl[\prod_{v\ {\rm s.e.p.}}\sup_{r,\ul\o}
\Big|\frac{Y_{h_v-1,r,\ul\o}}{Z_{h_v-1}}\Big|\Biggr]\times\nonumber\\
&&\qquad \times \Big[\prod_{\substack{v\,{\rm not}\\ {\rm e.p.}}}\hskip-.02truecm\frac{C^{\sum_{i=1}^{s_v}|P_{v_i}|-|P_v|}}{s_v!} \;2^{ \frac\e2|P^\psi_v|}2^{2-\frac12|P_v^\psi|-|P_v^J|-z(P_v)}
\Big]\times\label{cor.eq}\\
&&\qquad \times  \Big[\prod_{v\in V_{nt}(\tau^*)} 2^{2(s^*_v-1)h_v}e^{-c\sqrt{2^{h_v} \delta_v}}\Big]
\;,\nonumber\eea
see \cite[Sect.7.2]{GMT17a} and in particular \cite[(7.12)]{GMT17a} and following discussion for a derivation of this estimate. 
For the definition of the `pruned tree' $\tau^*=\tau^*(\tau)$, we refer to the discussion after \cite[(7.5)]{GMT17a}; we recall that 
$V_{nt}(\tau^*)$ is the set of non trivial vertices of $\tau^*$, that 
$s^*_v$ is the number of vertices immediately following $v$ on $\tau^*$, and $\delta_v$ is the tree distance of the set $\cup_{f\in P_v^J}\{x(f)\}$. 
See also \cite[(3.11)]{BFM09} and \cite[(4.5)]{GGM} for analogous formulas and proofs in the context of the Thirring model and of non-planar 2D Ising model, respectively. By proceeding as in \cite[Sect.IV.B]{GGM}, see also \cite[Sect.3.1]{BFM09}, by plugging \eqref{cor.eq} in \eqref{cor.6.131} and summing 
over $N_n,N_s, \tau, {\bf P}, T$, we get the analogue of 
\cite[(4.31)]{GGM}, with the appropriate modifications\footnote{{Compared with \cite[(4.31)]{GGM}, in \eqref{Linftyn0.0} there are a few significant differences: 
(1) the sum over scales is over negative integers, rather than over scales smaller than $N+2$, due to the different choice of the ultraviolet scale, i.e., of the lattice spacing, equal to $1=2^0$ in this 
paper, and to $a\sim 2^{-N}$ in \cite{GGM}; (2) there is no analogue of the `short memory factor' $2^{-\theta N}$, due to the fact that here the model 
has a quartic marginal coupling, $\lambda_h$, flowing to a non trivial fixed point $\lambda_{-\infty}$, contrary to \cite{GGM}, where 
the quartic coupling is irrelevant, and the infrared theory tends exponentially fast, like $\sim 2^{\theta h}$, to a trivial fixed point; this fact also explains the reason why in the definition of 
$\alpha_v$ here there is no $\theta$, contrary to \cite[(4.32)]{GGM}; (3) there is a damping factor $2^{(h-h_0^*)(1-\epsilon)}$ associated with the 
branch of $\tau^{**}$ from $v_0^*$ to $v_0$; in \cite{GGM} its analogue did not appear explicitly, because it was summed over $h\le h_0^*$, 
see the comment right before \cite[(4.24)]{GGM}; (4) in the definition of $\alpha_v$, the constant in front of $s_v^{*,1}$ is $\epsilon'<\epsilon$ rather than $\epsilon$, 
see \cite[(4.32)]{GGM}, the reason being the presence of the factors $\frac{|Y_{h_v,\cdot}|}{|Z_{h_v}|}\le C 2^{-h_vC|\l|}$ associated with the special endpoints, which 
have the effect of changing $\epsilon s_v^{*,1}$ in the definition of $\alpha_v$ into $(\epsilon-C|\lambda|)s_v^{*,1}\ge \epsilon' s_v^{*,1}$.}}:
\begin{equation}\label{Linftyn0.0}\big| W^{(h),\infty}_{0,m;\ul r}(\ul y)\big|\le C^m \sum_{\substack{h< h_0^*<0\\ \tau^{**}\in {\mathcal T}^{**}_{0,h_0^*;m}}}
2^{(h-h_0^*)(1-\epsilon)}
\Big[
\prod_{v\in V_{nt}(\tau^{**})}2^{\alpha_v h_v}e^{-c\sqrt{2^{h_v}\delta_v}}\Big],\end{equation}
where $\alpha_v=(1-\epsilon)s_v^*+\epsilon' s_v^{*,1}$ for $v=v_0^*$, and $\alpha_v=(1-\epsilon)(s_v^*-1)+\epsilon' s_v^{*,1}$ for $v>v_0^*$; here 
$\epsilon,\epsilon'$ are two suitable, small, constants, such that $\epsilon>\epsilon'>0$. Once again, we refer to \cite[Sect.IV.B]{GGM} for the definitions, in particular 
of ${\mathcal T}^{**}_{0,h_0^*;m}$. By summing over $h_0^*$ and $\tau^{**}$, using in particular \cite[(4.14)-(4.15)]{GGM} and the fact that $\delta_v\ge (s_v^*-1)\delta$, with 
$\delta$ the minimal pairwise distance among the points in $\ul y$, see the lines after \cite[(4.15)]{GGM} and the following equation, we get the analogue of 
\cite[(4.33)]{GGM}, namely 
\begin{equation}\label{Linftyn0}\big| W^{(h),\infty}_{0,m;\ul r}(\ul y)\big|\le C^m \frac{2^{h(1-\epsilon)}e^{-\frac{c}2\sqrt{2^{h}\delta}}}{(\delta+1)^{(1-\epsilon)(m-1)}},
\end{equation}
where we also used the fact that $\sum_{v\in V_{nt}(\tau^{**})}\alpha_v=(1-\epsilon)m+\epsilon'm>(1-\epsilon)m$.}

The {\it finite size correction} $\mathcal S_\bt(J)$ is, by definition, 
the difference between the right side of \eqref{recallgenfun} and $S_L(J)$, see \eqref{definSL}, which we write as
\begin{equation}
\mathcal S_\bt(J) =\sum_{m\ge 1}\sum_{\ul r\in\{1,\ldots,4\}^m}\sum_{\ul y\in\L^m}J_{y_1,r_1}\cdots J_{y_m,r_m} w^{\bt,L}_{m;\ul r}(\ul y).\label{6.99}\end{equation}
Its kernels, $w^{\bt,L}_{m;\ul r}$, can be bounded via the same strategy discussed in App.\ref{dimr} for the bounds on the finite size corrections to the effective 
potential; the result is that $w^{\bt,L}_{m;\ul r}$ admits the same dimensional bounds \eqref{weightL1n0}, \eqref{Linftyn0} as $W^{(h_L),\infty}_{0,m;\ul r}$, 
in the case that the scale index $h$ equals $h_L$; that is, recalling that $2^{h_L}\propto L^{-1}$,
\be  \|w^{\bt,L}_{m}\|_{\kappa,h_L}\le C^m L^{m-2}L^{C|\l|m},\quad { \big|w^{\bt,L}_{m;\ul r}(\ul y)\big|\le C^mL^{-(1-\epsilon)} 
(\delta+1)^{-(1-\epsilon)(m-1)}},\label{6.100}\ee 
for some $C,\epsilon>0$ (in the second inequality we neglected the factor $e^{-\kappa_0\sqrt{2^{h} d(\ul y)}}$, because, for $h=h_L$, it is bounded from above 
and below by an $O(1)$ constant).

Let us go back to \eqref{yaaA}:
in the second line,  
$V^{(h_L)}(\Psi,J)$ is the effective potential on scale $h_L$, which admits the following representation:
\bea  &&   V^{(h_L)}(\Psi,J)= V^{(h_L)}(\Psi,0)+
\sum_{m\ge 1} \sum_{\substack{y_1,\ldots, y_m\\ r_1,\ldots, r_m}}\Big(\prod_{j=1}^m
J_{y_j,r_j}\Big)\times\label{reprrealspace}\\
&&\qquad  \times\ \Big[\sum_{\omega_1,\omega_2=\pm} Q^{\bt,L}_{m;(\o_1,\o_2),\ul r}(\ul y)\hat\Psi^+_{{\o_1}}\hat\Psi^-_{{\o_2}}+
R^{\bt, L}_{m;\ul r}(\ul y)\hat \Psi^+_{+}\hat \Psi^-_{+}\hat \Psi^+_{-}\hat \Psi^-_{-}\Big],\nonumber\eea
where $V^{(h_L)}(\Psi,0)$ is the same as \eqref{eq:692}, and the kernels in the second line satisfy the following (recall that $\ul y=(y_1,\ldots, y_m)$
and $\ul r=(r_1,\ldots, r_m)$):
\begin{equation} \|Q^{\bt,L}_{m}\|_{\kappa,h_L}\le  C^m L^{m-5} L^{C|\lambda| m},\quad  \|R^{\bt,L}_{m}\|_{\kappa,h_L}\le C^m |\lambda|\,L^{m-8}L^{C|\lambda| m},
\label{weightL1bound:YZ} \end{equation}
\begin{eqnarray} &&
{\big|Q^{\bt,L}_{m;(\omega_1,\omega_2),\ul r}(\ul y)\big|\le C^m L^{-(4-\epsilon)} 
(\delta+1)^{-(1-\epsilon)(m-1)}},\label{pointwisebound:YZuno}\\
&&  {\big|R^{\bt,L}_{m;\ul r}(\ul y)\big|\le
C^m L^{-(7-\epsilon)} 
(\delta+1)^{-(1-\epsilon)(m-1)}}, \label{pointwisebound:YZ}\end{eqnarray}
uniformly in $\bt$. In order to prove this representation and bounds, 
one starts from the general  representation of the effective 
potential, \eqref{eq:6.16ren}, in the case that $h=h_L$, 
focusing on the terms $m\ge 1$, $n\ge 2$; next, one uses the Fourier representation for the Grassmann field, see \eqref{eq:fourier} and \eqref{eq:sh}; 
recalling that at scale $h_L$ the Grassmann field 
contains only four modes, denoted by $\hat\Psi^\pm_{\o}$ with $\o=\pm$ (see \eqref{eq:apc}), one finds that the sum over $n$ 
is limited to $n=2,4$; therefore, one obtains  \eqref{reprrealspace}, with 
\begin{equation} Q^{\bt,L}_{m;\ul\o,\ul r}(\ul y)= L^{-4}\sum_{\ul x,\ul i,\ul q}\Big(\prod_{j=1}^2 e^{i\sigma_j k^{\omega_j}_\bt x_j} D_{i_j}^{q_j}(k^{\omega_j}_\bt-\bar p^{\omega_j})\Big)
W^{(h_L)}_{2,m, \ul i,\ul q;\ul\o,\ul r}(\ul x,\ul y),\label{W2mdef}\end{equation}
\begin{equation} R^{\bt,L}_{m;\ul r}(\ul y)= L^{-8}\sum_{\ul x,\ul\o,\ul i,\ul q}\alpha_{\ul\o}\Big(\prod_{j=1}^4 e^{i\sigma_j k^{\omega_j}_\bt x_j} D_{i_j}^{q_j}(k^{\omega_j}_\bt-\bar p^{\omega_j})\Big)
W^{(h_L)}_{4,m, \ul i,\ul q;\ul\o,\ul r}(\ul x,\ul y),\label{W4mdef}
\end{equation}
where: $\sigma_j=(-1)^{j-1}$, $D_i^q(k)$ is the Fourier symbol of $\hat\partial_i^q$, which satisfies $|D_i^q(k)|\le C |k|^q$
(the reason why the argument of $D_{i_j}^{q_j}$ is shifted by $\bar p^\o$ is that it originates from the action of $\hat\partial_i^q$
on $L^{-2}e^{\pm i (k^\o_\bt-\bar p^\o)x}\hat \Psi^\pm_\o$, see \eqref{eq:sh}), and,
in the second line, $\alpha_{\ul\o}$ is a symmetry factor that takes values in $\{0,-1,+1\}$ (it can be easily computed using the anticommutation properties of the Grassmann variables) and whose precise value is inessential for the
following bounds. 
{By using the bounds \eqref{bouW} on the weighted $L_1$ norm of $W^{(h_L)}_{n,m,\ul i,\ul q;\ul \o,\ul r}$, 
%as well as the corresponding pointwise bounds, 
we obtain \eqref{weightL1bound:YZ}. Similarly, using the corresponding pointwise bound\footnote{{More precisely, the norm in the left side of 
\eqref{corpointwise} is a mixed norm, pointwise in $\ul y$ and $L_1$ in $\ul x$. In order to obtain \eqref{corpointwise}, we proceed as follows: we start from 
the appropriate analogue of \eqref{cor.eq}, which has a factor $2^{h(2-m-n/2)}$ instead of $2^{h(2-m)}$, and where $h=h_L$; next, we keep the factor 
$2^{h_L(-n/2)}$ on a side and manipulate the rest of the expression as discussed after \eqref{cor.eq}, thus obtaining the right side of \eqref{corpointwise}.}}, 
analogous to \eqref{Linftyn0}, namely 
\begin{eqnarray} &&
\sum_{\ul x,\ul i,\ul q}\Big(\prod_{j=1}^n \big| D_{i_j}^{q_j}(k^{\omega_j}_\bt-\bar p^{\omega_j})\big|\Big)
\big|W^{(h_L)}_{n,m, \ul i,\ul q;\ul\o,\ul r}(\ul x,\ul y)\big|\le \label{corpointwise}\\
&&\hskip3.truecm  \le C^m 2^{h_L(1-\epsilon-n/2)} (\delta+1)^{-(1-\epsilon)(m-1)},\nonumber\end{eqnarray}
we obtain \eqref{pointwisebound:YZuno}-\eqref{pointwisebound:YZ}.}

We are now in the position of computing the integral  in \eqref{yaaA}: by doing so, we get 
\begin{eqnarray}  && e^{\mathcal W_L^{(\bt)}(A,0)}=e^{L^2 \Delta+S_L(J)+\mathcal S_\bt(J)} (1+s_\bt)
\Big\{\, Z^0_{\bt}+\label{yeaAA}\\
&& +\, \tilde Z^0_{\bt}\,L^{-2}\cdot \Big[\s_{\bt}+\sum_{m\ge 1} \sum_{\ul y,\ul r}\Big(\prod_{j=1}^m J_{y_j,r_j}\Big) 
V^{\bt,L}_{m; \ul r}(\ul y)\nonumber\\
&& +\sum_{m,m'\ge 1} \sum_{\ul y,\ul y', \ul r, \ul r'}\Big(\prod_{j=1}^m J_{y_j,r_j}\Big) \Big(\prod_{j=1}^{m'} J_{y_j',r_j'}\Big) 
W^{\bt,L}_{m,m';\ul r,\ul r'}(\ul y, \ul y')\Big]\Big\},\nonumber\end{eqnarray}
where: 
\begin{eqnarray}  && V^{\bt,L}_{m; \ul r}(\ul y):=L^6 R^{\bt,L}_{m;\ul r}(\ul y)+
  \sum_\o L^{3} Q^{\bt,L}_{m;(\o,\o), \ul r}(\ul y)(-L\mu_{\bt,-\o}+u_{\bt,-\o}),\nonumber\\
&& W^{\bt,L}_{m,m';\ul r,\ul r'}(\ul y, \ul y'):= L^6 \sum_{\omega=\pm}\omega Q^{\bt,L}_{m;(+,\omega), \ul r}(\ul y)\, Q^{\bt,L}_{m';(-,-\omega), \ul r'}(\ul y').\nonumber\end{eqnarray}
We rewrite \eqref{yeaAA} in condensed notation as
\be e^{\mathcal W_L^{(\bt)}(A,0)}=e^{L^2 \Delta+S_L(J)+\mathcal S_\bt(J)} (1+s_\bt)
\Big\{Z^0_{\bt}+\frac{\tilde Z^0_{\bt}}{L^{2}}
\sum_{m\ge 0}\sum_{\ul r,\ul y} \big[\prod_{i=1}^m J_{y_i,r_i}\big]\tilde w^{\bt ,L}_{m;\ul r}(\ul y)\Big\},\nonumber\ee
where the kernels $\tilde w^{\bt ,L}_{m;\ul r}(\ul y)$ are translationally invariant, and the summand with $m=0$ should be interpreted as being equal to $\s_\bt$.
Using the definition of $\tilde w^{\bt ,L}_{m;\ul r}(\ul y)$ implicit in this rewriting, and the pointwise bounds on $Q^{\bt,L}_{m;\ul\o,\ul r}$, $R^{\bt,L}_{m;\ul r}$ 
discussed above, we find 
\begin{equation}\label{realspacetildewL}{\big|\tilde w^{\bt ,L}_{m;\ul r}(\ul y)\big|\le C^mL^{-(1-\epsilon)} 
(\delta+1)^{-(1-\epsilon)(m-1)}},\quad m\ge 1.\end{equation}
Finally, we take the appropriate linear combination of $e^{\mathcal W_L^{(\bt)}(A,0)}$ in order to obtain the generating function of the interacting dimer model: 
\be e^{\mathcal W_L(A,0)}=\frac12\sum_{\bt} c_\bt e^{\mathcal W_L^{(\bt)}(A,0)}=Z_L e^{S_L(J)+\tilde{\mathcal S}_L(J)},
\ee
where 
\be \tilde{\mathcal S}_L(J)=\log (1+F_L(J)),\label{log1+F}\ee
with (recall from \eqref{eq:gfp} that $Z_L=e^{L^2\Delta}Q_L^0(1+r_L)$)
\bea && F_L(J)= \frac{1}{2Q_L^0(1+r_L)}\sum_\bt c_\bt  (1+s_\bt)\Big\{\frac{\tilde Z^0_{\bt}}{L^{2}}
\sum_{m\ge 1}\sum_{\ul r,\ul y} \big[\prod_{i=1}^m J_{y_i,r_i}\big]\tilde w^{\bt ,L}_{m;\ul r}(\ul y)\nonumber\\
&&\qquad +\, (e^{\mathcal S_\bt(J)}-1) \Big(Z^0_{\bt}+\frac{\tilde Z^0_{\bt}}{L^{2}}
\sum_{m\ge 0}\sum_{\ul r,\ul y} \big[\prod_{i=1}^m J_{y_i,r_i}\big]\tilde w^{\bt ,L}_{m;\ul r}(\ul y)\Big)\Big\}.\eea
By using the representation \eqref{6.99} for $\mathcal S_\bt(J)$, the bounds \eqref{6.100} on its kernels, and the bounds \eqref{realspacetildewL} on 
$\tilde w^{\bt ,L}_{m;\ul r}(\ul y)$, the reader can check\footnote{Thanks to \eqref{6.99}  and  \eqref{6.100}, $\mathcal S_\bt(J)$ is small {for $J$ small and} $L\to\infty$, so that one can {Taylor expand}
    $e^{\mathcal S_\bt(J)}-1=\sum_{n\ge 1}(\mathcal S_\bt(J))^n/n!$. { By rearranging this series in a form analogous \eqref{6.99} one finds that the corresponding kernels satisfy qualitatively the same bounds as \eqref{realspacetildewL}. 
    Similarly, by using the Taylor series of $\log(1+x)$, one finds that $F_L(J)$ admits an analogous representation, with kernels satisfying analogous bounds as well.}\label{fotonotaroughly}} that $F_L(J)$ admits a representation 
analogous to \eqref{6.99} as well, with $w^{\bt,L}_{m;\ul r}$ replaced by a modified kernel $W^L_{m;\ul r}$, satisfying the same pointwise estimates as 
the second of \eqref{6.100} and \eqref{realspacetildewL}. Finally, using \eqref{log1+F}, one finds (see footnote \ref{fotonotaroughly}) that also $\tilde{\mathcal S}_L(J)$ admits an analogous representation, with associated kernel satisfying the same pointwise estimates as 
the second of \eqref{6.100} and \eqref{realspacetildewL}.

\medskip

Thanks to these estimates, and to the explicit form of $S_L(J)$, we conclude, as desired, that the thermodynamic limit of the correlations of the interacting dimer model exist and are
given by (we let $e_i$ be the edge of type $r_i$ with black vertex $x_i$, and we assume the $m$ edges $e_1,\ldots e_m$ to be all different from each other):
\be \mathbb E_\l(\mathds 1_{e_1};\cdots; \mathds 1_{e_m})=m!\,\sum_{h\le 0}W^{(h),\infty}_{0,m;(r_1,\ldots,r_m)}(x_1,\ldots ,x_m).\label{6.105}\ee
Note that the sum in the right side is absolutely convergent, thanks to the pointwise bounds \eqref{Linftyn0}. Note also that, since $\{\cup_{e\in\Lambda}\mathds 1_e\}_{\Lambda\subset E_{\mathbb Z^2}}$, with $E_{\mathbb Z^2}$ the edge set of $\mathbb Z^2$ and $\Lambda$ a finite subset of it, is a basis for the local functions of dimer configurations, the existence of the 
thermodynamic limit for all the multipoint dimer correlations implies the existence of the thermodynamic limit of the average of any local function 
of the dimer configuration; this proves Eq.\eqref{eq:exlim} of Theorem \ref{th:1}. 

A similar discussion can be repeated for mixed dimer/Grassmann field correlations, but we will not belabor further details here. 

\subsection{Asymptotic behavior of the dimer correlation functions}\label{sec:asympt}

In order to complete the proof of our main theorems,
we are left with proving that the large distance behaviour of
the (thermodynamic limit of the) interacting propagator, vertex function and dimer-dimer 
correlation can be expressed in terms of linear combinations
of the appropriate correlations of the reference model, as stated in
Section \ref{sec:proveth}. We limit ourselves to the discussion 
of the asymptotic behaviour of the two-point dimer-dimer correlation,
$\mathbb E_\l(\mathds 1_{e_1};\mathds 1_{e_2})\equiv G^{(0,2)}_{r_1,r_2}(x_1,x_2)$, 
i.e., to the proof of \eqref{hh110}, and we leave the analogous discussion for the
propagator and vertex function (leading to \eqref{h10ab}, \eqref{h10a}) to the reader. 

We use a strategy analogous to the one of \cite[Section 7.1 and 7.2]{GMT17a} and we refer the reader to those sections for further details. The starting point is \eqref{6.105} with $m=2$, which we write as
\be G^{{(0,2)}}_{r_1,r_2}(x_1,x_2)=2 \sum_{h\le 0}W^{(h),\infty}_{0,2;(r_1,r_2)}(x_1,x_2).\ee
This is the analogue of \cite[Eq. (7.4)]{GMT17a} with $q=m=2$. The multiscale construction implies, of course, that $W^{(h),\infty}_{0,2;(r_1,r_2)}(x_1,x_2)$ can be written as a sum over trees with root at 
scale $h$ and  two external $J$ fields, that is 
\be G^{{(0,2)}}_{r_1,r_2}(x_1,x_2)=\sum_{h\le 0}\sum_{N\ge 0}\sum_{n=1}^2 \sum_{\t\in \mathcal T^{(h)}_{N,n}}\sum_{\substack{{\bf P}\in\mathcal P_\t:\\ |P_{v_0}|=|P_{v_0}^J|=2}}
S_{\t,{\bf P},(r_1,r_2)}(x_1,x_2);\label{6.107}\ee
this is the analogue of \cite[Eq. (7.5)]{GMT17a}. We now decompose \eqref{6.107} as in \cite[Eq. (7.7)]{GMT17a}, namely, 
\be G^{{(0,2)}}_{r_1,r_2}(x_1,x_2)=\mathcal S^{(1)}_{r_1,r_2}(x_1,x_2)+\mathcal S^{(2)}_{r_1,r_2}(x_1,x_2)+\mathcal S^{(3)}_{r_1,r_2}(x_1,x_2), \label{6.108}\ee
where: $\mathcal S^{(1)}_{r_1,r_2}$ (resp. $\mathcal S^{(2)}_{r_1,r_2}$) is the sum \eqref{6.107} restricted to trees whose normal endpoints are all of type $\l$, whose special endpoints are both density endpoints (resp. 
mass endpoints), 
see the definition after \eqref{YZ}, and whose value is computed by replacing all the propagators by relativistic ones; $\mathcal S^{(3)}_{r_1,r_2}$ is the remainder, which is given by a sum over trees, which  
either have at least one endpoint of type $\n,a,b,\mathcal R V^{(-1)}$, or have at least one `remainder propagator' $g^{(h)}_\o-g^{(h)}_{R,\o}$. 

Not surprisingly, the easiest term to bound in \eqref{6.108} is the third one: by  a proof analogous to the one leading to \eqref{6.77},
we find that  $\mathcal S^{(3)}_{r_1,r_2}$ can be bounded as  % (take e.g. $\th=1/2$) 
\be |\mathcal S^{(3)}_{r_1,r_2}(x_1,x_2)|\le C\sum_{h\le 0} 2^{h(2-2C|\l|)}2^{\th h} e^{-\k\sqrt{2^h|x_1-x_2|}}\le \frac{C'}{|x_1-x_2|^{2+\th-C|\l|}}.\label{eq:rem100}\ee
Note that, for $\l$ small enough, at large distances the r.h.s. of \eqref{eq:rem100} is negligible w.r.t.
both $S^{(1,1)}_{R,\o,\o}(x,y)$ and $S^{(2,2)}_{R,\o,-\o}(x,y)$ (recall the estimates
\eqref{googleee} and \eqref{instangraaam}) and therefore $\mathcal S^{(3)}_{r_1,r_2}(x_1,x_2)$ can be absorbed in the error term $R_{r_1,r_2}(x,y)$ in
\eqref{hh110}.
A couple of  comments about how the bound \eqref{eq:rem100} is obtained will be useful (see \cite[Sec. 7.1]{GMT17a} for more details on similar estimates). The factor 
$2^{h(2-2C|\l|)}e^{-\k\sqrt{2^h|x_1-x_2|}}$ is the `dimensional bound' on   trees with root on scale $h$ and external fields $J_{x_1,r_1},J_{x_2,r_2}$. The factor $2^{\th h}$
is the `dimensional gain' arising from the fact that 
all the trees contributing to $\mathcal S^{(3)}_{r_1,r_2}$ have at
least one endpoint of type $\n,a,b,\mathcal R V^{(-1)}$ or one remainder
propagator. In fact, recall that the value of an endpoint 
of type $\n,a,b,\mathcal R V^{(-1)}$, if located at scale $h\le k\le 0$, is 
of the order $O(2^{\th k}\l)$; by the short memory property, we get a factor $2^{\th'(h-k)},\th<\th'<1$ and a sum over $k\ge h$ finally produces the factor 
$2^{\th h}$ in the right side of \eqref{eq:rem100}. The contributions with one remainder propagator on scale $k\ge h$ are treated analogously. 

Let us now consider $\mathcal S^{(1)}_{r_1,r_2}(x_1,x_2)$ and $\mathcal S^{(2)}_{r_1,r_2}(x_1,x_2)$. First of all, note that they can be naturally 
rewritten as 
\bea && \mathcal S^{(1)}_{r_1,r_2}(x_1,x_2)=\sum_{\o=\pm}\mathcal S^{(1)}_{r_1,r_2;\o,\o}(x_1,x_2),\\
&& \mathcal S^{(2)}_{r_1,r_2}(x_1,x_2)=\sum_{\o=\pm}\mathcal S^{(2)}_{r_1,r_2;\o,-\o}(x_1,x_2),\eea
where $\mathcal S^{(1)}_{r_1,r_2;\o,\o}$  is the sum over the trees whose special endpoints % with coordinate label $x_1$ (resp. $x_2$) has quasi-particle label $\o_1$ (resp. $\o_2$). 
have  labels $(\o,\o)$; similarly, $\mathcal S^{(2)}_{r_1,r_2;\o,-\o}$ is the sum over the trees whose special endpoint with coordinate label $x_1$ (resp. $x_2$) 
has  label $(\o,-\o)$ (resp. $(-\o,\o)$). 
In the tree expansion for $\mathcal S^{(j)}_{r_1,r_2;\o_1,\o_2}$, we further decompose the tree values in a dominant plus a subdominant part,
the dominant part being obtained via the following replacements: replace all the values $\l_h$ of the endpoints of type $\l$ by $\l_{-\infty}=\l_{-\infty}(\l)$; replace all the values $z_h$ of the rescaling factors 
by $z_{-\infty}=z_{-\infty}(\l)$; replace all the values $Y_{h,r,(\o,\o)}/Z_{h-1}$ (resp. $Y_{h,r,(\o,-\o)}/Z_{h-1}$) of the density (resp. mass) endpoints, by 
\begin{eqnarray}
  \label{eq:ang1}
2^{-\eta}B_{-\infty,r,\o}/{A_{-\infty}}  
\end{eqnarray}
and
\begin{eqnarray}
  \label{eq:ang2}
2^{(\h-\h_1)h}\,2^{-\eta}C_{-\infty,r,\o}/{A_{-\infty}}  
\end{eqnarray}
respectively, that is their asymptotic value as $h\to-\infty$ (recall Eqs. \eqref{eq:6.84}
and \eqref{eq:6.86}). The decomposition of the tree values into dominant and subdominant
parts induces a similar decomposition of $\mathcal S^{(j)}_{r_1,r_2;\o_1,\o_2}$:
$$\mathcal S^{(j)}_{r_1,r_2;\o_1,\o_2}(x_1,x_2)=\mathcal S^{(j),d}_{r_1,r_2;\o_1,\o_2}(x_1,x_2)+
\mathcal S^{(j),s}_{r_1,r_2;\o_1,\o_2}(x_1,x_2),\quad j=1,2,$$
with obvious notation. By using the fact that $\l_h,z_h,A_h,B_{h,r,\o},C_{h,r,\o}$
all converge exponentially fast to their limiting values as $h\to-\infty$, we find that 
$\mathcal S^{(j),s}_{r_1,r_2\;\o_1,\o_2}(x_1,x_2)$ can 
be bounded in the same way as \eqref{eq:rem100} and can be absorbed in the error term $R_{r_1,r_2}(x,y)$ in \eqref{hh110}.

We are left with the dominant parts, $\mathcal S^{(j),d}_{r_1,r_2;\o_1,\o_2}$, $j=1,2$, and to prove \eqref{hh110} we need to connect them to the correlation functions of the continuum model  of Section \ref{sec:IR}. {As discussed in Sections \ref{remmm} and \ref{bareZ}, we fix}
the coupling constant $\l_\infty$ {and the bare wave function renormalization $Z$} of the continuum model in such a way that 
$\l_{-\infty;R}(\l_\infty)=\l_{-\infty}(\l)$, so that the critical exponents 
$\h(\l),\h_1(\l)$ are the same as for the dimer model, {and that the flow of the wave function renormalization of the reference model satisfies 
\be Z_{h;R}=A_{-\infty} 2^{-\h h}(1+O(\l^2\,2^{\th h})),\label{eq:6.84bis}\ee
with the {\it same} $A_{-\infty}$ as in \eqref{eq:6.84}. That is, $Z_{h;R}=Z_h(1+O(\l^2 2^{\th h})$.}
The form of the special endpoints is different for the dimer and the continuum model, simply because  the external fields $J$ of the continuum model have the structure $J^{(j)}_{x,\o}$ instead of $J_{x,r}$.
In fact, when the multi-scale construction is applied to the continuum model, the value of a special endpoint of type $j=1$ (density) or $j=2$ (mass)  is of the form (to be compared with
\eqref{YZ})
\[
\frac{Y^{(j)}_{h;R}}{Z_{h-1;R}}F^R_{Y,j,\o}(\sqrt{Z_{h-1;R}}\psi^{(\le h)}), \quad F^R_{Y,j,\o}(\psi)=\int_\L dx J^{(j)}_{x,\o}\rho^{(j)}_{x,\o}, 
\]
with $\rho^{(j)}_{x,\o}$ as in \eqref{rhodef}. As $h\to-\infty$ one has, in analogy with \eqref{eq:6.86},
\bea Y^{(1)}_{h;R}&=&\tilde B_{-\infty}2^{-\h h} (1+O(\l\,2^{\th h})),\label{eq:6.86bis}\\
Y^{(2)}_{h;R}&=&\tilde C_{-\infty}2^{-\h_1 h}(1+O(\l\,2^{\th h})),\nonumber\eea
for suitable analytic functions  $\tilde B_{-\infty}$, $\tilde C_{-\infty}$ of $\lambda$, which equal $1$ for $\l=0$% , such that 
% $\tilde B_{-\infty}=\tilde B_{-\infty}(\l_{\infty})=1+O(\l_{\infty})$, 
% $\tilde C_{-\infty}=\tilde C_{-\infty}(\l_{\infty})=1+O(\l_{\infty})$
.
Now call $\mathcal S^{(j),d;R}_{\o_1,\o_2}$ the analog of
$\mathcal S^{(j),d}_{r_1,r_2;\o_1,\o_2}$ for the continuum model. The two functions  differ only because the values associated with the special endpoints differ: in the dimer model, 
these are given as in \eqref{eq:ang1} (if $j=1$) or \eqref{eq:ang2} (if $j=2$); in the reference model, one needs to replace
%$A_{-\infty}\to \tilde A_{-\infty}, 
$B_{-\infty,r,\o}\to \tilde B_{-\infty}$ and $C_{-\infty,r,\o}\to \tilde C_{-\infty}$. In conclusion, 
\bea && \mathcal S^{(1),d}_{r_1,r_2;\o,\o}(x_1,x_2)=\hat K_{\o,r_1}\hat K_{\o,r_2}\mathcal S^{(1),d;R}_{\o,\o}(x_1,x_2),\\
&& \mathcal S^{(2),d}_{r_1,r_2;\o,-\o}(x_1,x_2)=e^{i(\bar p^\o-\bar p^{-\o})(x_1-x_2)}\hat H_{-\o,r_1}\hat H_{\o,r_2}\mathcal S^{(2),d;R}_{\o,-\o}(x_1,x_2),\nonumber\eea
with 
$$\hat K_{\o,r}={\frac{B_{-\infty,r,\o}}{\tilde B_{-\infty}}}, \quad \hat H_{\o,r}={\frac{C_{-\infty,r,-\o}}{\tilde C_{-\infty}}}.$$
The oscillating prefactor $e^{i(\bar p^\o-\bar p^{-\o})(x_1-x_2)}$
appears because it is included in the  definition of $F_{Y;r,\ul \o}$. Finally, $\mathcal S^{(j),d;R}_{\o,\o'}(x_1,x_2)$ equals
$S^{(j,j)}_{R,\o,\o'}(x_1,x_2)$ (cf. \eqref{eq:5.6op}, \eqref{eq:5.6}), up to subdominant
corrections that can be absorbed in the error term $R_{r_1,r_2}(x_1,x_2)$ and we obtain \eqref{hh110}, as wished.
 A similar discussion leads to \eqref{h10ab},
\eqref{h10a}, and we leave the details to the reader. This concludes the proofs of Theorems \ref{th:1} and \ref{th:2}. %\qed

\appendix
\section{Symmetries}\label{symm}

% 1. The interacting dimer model has only one relevant symmetry, which is
% $$\psi^\pm_x\to (-1)^x\psi^\pm_x, \qquad c\to c^*,$$
% where $c$ is a generic constant appearing in the action. This readily implies that the $\l_h$ are real. 

% \medskip

The propagator $g^{(h)}_{R,\o}(x,y)$ in \eqref{gRo} satisfies three symmetries, which are the real-space counterparts of the following:
\bea && \bar D_{-\o}(k)= -\bar D^*_\o(k),\nonumber\\
&& \bar D_\o(A^{-1}\s_1 Ak)=i\o \bar D^*_\o(k),\label{eq:nclfp}\\
&& \bar D_\o(A^{-1}\s_3 Ak)=\bar D^*_\o(k).\nonumber\eea
This means that the quadratic action associated to
the Grassmann integration with propagator $g^{(h)}_{R}$ is invariant under three transformations: for instance, the one associated to the first of \eqref{eq:nclfp} is $\hat \f^\pm_{k,\o}\mapsto i \hat \f^\pm_{k,-\o}$ and at the same time any constant appearing in the action is replaced by its complex conjugate. Then, one sees inductively that the effective potentials one obtains by setting  $\{\nu_{h',\o},a_{h',\o},b_{h',\o}\}_{h<h'\le -1}$ to zero (as is done in the definition of relativistic kernels) satisfy the same three symmetries. As a consequence,  $\hat W^{(h),R}_{2,0;(\o,\o)}(k)$ inherits the symmetries analogous to \eqref{eq:nclfp}, that are \eqref{eq:symm}. Note that biliner terms such as $\nu_{h',\o}F_{\nu;\o}(\f),a_{h',\o}F_{a;\o}(\f),b_{h',\o}F_{b;\o}(\f)$ in \eqref{lcwhV} would  break the above-mentioned symmetries, unless the coefficients $\nu_{h',\o},$ etc. are zero. On the other hand, terms such as $\l_{h'} F_{\l}(\f)$ are  invariant because we know by induction that  $\l_h',h'>h$ is real.

Next, let us show that \eqref{eq:symm} implies \eqref{eq:symm1}. Write for lightness of notation $\z_\o:=\bar D_\o(k),\z_+^*=-\z_-$. % Let us denote by $\z=\bar\a_+k_1+\bar\b_+k_2=\bar D_+(k)$ and $\z^*=-\bar D_-(k)$ its complex conjugate.
Since $k\cdot\partial_k\hat W^{(h),R}_{2;(+,+)}(0)$ is linear in $k$, we can write it as $f(\z_+,\z_-)= c \z_+ +c'\z_-$ for some complex constants $c,c'$. Note that the transformation $k\mapsto A^{-1}\sigma_3 A k$ (resp. $k\mapsto A^{-1}\sigma_1 A k$) implies $(\z_+,\z_-)\mapsto (-i \z_-,i \z_+)$ (resp. $(\z_+,\z_-)\mapsto (- \z_-,-\z_+)$).
% Let us also denote by $f(\z,\z^*)$ be the same 
% as $k\cdot\partial_k\hat W^{(h),1}_{2;(+,+)}(0)$, written in terms of the two independent complex coordinates $\z,\z^*$.
By linearizing the second and third equation in \eqref{eq:symm} we get:
\be f(-i\z_-,i\z_+)=i[f(\z_+,\z_-)]^*, \qquad f(-\z_-,-\z_+)=[f(\z_+,\z_-)]^*, \ee
which readily imply that $c'=0,c\in\mathbb R$. This is the desired formula for $\o=+$. By using the first of \eqref{eq:symm}, we get the desired formula for $\o=-$. 

\section{Finite size corrections and bounds on   $\mathcal R V^{(h)}$} \label{dimr}

The bounds on the kernels of the effective potential arising in the multi-scale procedure, such as Proposition \ref{prop:an}, as well as the
reason why the action of $\mathcal R$, responsible for the factors $2^{-z(P_v)}$ in \eqref{5.62},
makes  the renormalized perturbation theory convergent,
have been discussed 
several times in the literature in similar models, see e.g. \cite[Section 6.1.4]{GMT17a}. 
% In this section we discuss the dimensional gains associated with the
% action of $\mathcal R$, in particular as far as the finite size
% corrections are concerned. The general reason why $\mathcal R$
% produces dimensional gains sufficient for making the renormalized
% perturbation theory convergent (and in particular the factors
% $2^{-z(P_v)}$ in \eqref{5.62}) has been explained many times before,
% see e.g. \cite[Section 6.1.4]{GMT17a}, and will not be repeated
% here.
In particular, finite-size details  have  been discussed in \cite{BM-XYZ}, but the definition of the $\mathcal L,\mathcal R$ operators 
given there is different from the one proposed in this paper: in \cite{BM-XYZ} the action of $\mathcal L$ on the kernels of $V^{(h)}$ explicitly depends on the size $L$ of the box, see \cite[eq.(2.74)]{BM-XYZ}, 
while in the present case it only depends on the $L\to\infty$ limit of the kernels, see \eqref{eq:6.21}. This new definition simplifies some technical aspects of the multi-scale construction: for instance, the flow of the running coupling constants is   independent of $L$ in the present work.
The goal of this appendix is to discuss the modifications induced by the new definition of $\mathcal L,\mathcal R$ on the proof  of the bounds on the kernels of 
$\mathcal R V^{(h)}$. Familiarity with  
\cite[Sec. 6]{GMT17a} is assumed.

For illustrative purposes, we restrict our attention to the part of $\mathcal R V^{(h)}$, denoted $\mathcal R V_4^{(h)}$, that is quartic in the Grassmann fields, has no derivative terms $\hat\partial \f$, and is independent of $J$. 
A similar discussion applies to the terms quadratic in the Grassmann fields (either independent of or linear in $J$), 
but we shall not belabor the details here. For the quartic term (using the same notation for the kernel as in \eqref{eq:6.16ren}), 
we have from \eqref{eq:6.21}:
\begin{eqnarray} && 
\mathcal R V_4^{(h)}(\f)=\sum_{\substack{x_1,\ldots,x_4\in\L\\ \o_1,\ldots,\o_4}} \f^+_{x_1,\o_1}\f^-_{x_2,\o_2}\f^+_{x_3,\o_3}\f^-_{x_4,\o_4} \Big[W^{(h)}_{4,0,\ul0;\ul\o}(x_1,x_2,x_3,x_4)\nonumber\\
&&-{\bf 1}_{x_1=x_2=x_3=x_4}\sum_{x_2',x_3',x_4'\in\mathbb Z^2}W^{(h),\infty}_{4,0,\ul0;\ul\o}(x_1,x_2',x_3',x_4')\Big].\label{R4V}\end{eqnarray}
The kernel $W^{(h)}_{4,0,\ul0;\ul\o}$ is given by a tree expansion: 
\be W^{(h)}_{4,0,\ul0;\ul\o}(x_1,x_2,x_3,x_4)=\sum_{N\ge 1}\sum_{\t\in \mathcal T^{(h)}_{N,0}}\sum_{\substack{
{\bf P}\in\PP_\t \\ \xx_{v_0}}}^*\sum_{T\in{\bf T}}W_{\t,{\bf P},{T}}(\xx_{v_0}), \label{RV4tree}\ee
where the $*$ indicates the constraint that the field and coordinate labels associated with the external fields must match with the prescribed values of $\o_1,\ldots,\o_4$, $x_1,\ldots, x_4$.
Among the various contributions to the right side of \eqref{RV4tree}, there are those from trees such that $|P_v|>4$,  for which 
the action of $\mathcal R$ on all its vertices $v>v_0$ (i.e. for all vertices of $v$ that are descendants of $v_0$ along the tree), is trivial; we let $\bar {\mathcal T}^{(h)}_{N,0}$ be the family of these trees, and 
$\bar W^{(h)}_{4,0,\ul0;\ul\o}$ be the analogue of \eqref{RV4tree}, with the sum over $\t$ in the right side restricted to $\bar {\mathcal T}^{(h)}_{N,0}$; in terms of these
kernels, we let 
\begin{eqnarray} && 
\mathcal R\bar V_4^{(h)}(\f)=\sum_{\substack{x_1,\ldots,x_4\in\L\\ \o_1,\ldots,\o_4}}\Big[ \f^+_{x_1,\o_1}\f^-_{x_2,\o_2}\f^+_{x_3,\o_3}\f^-_{x_4,\o_4} \bar W^{(h)}_{4,0,\ul0;\ul\o}(x_1,x_2,x_3,x_4)\nonumber\\
&&-\d_{x_1,x_2}\d_{x_1,x_3}\d_{x_1,x_4}\sum_{x_2',x_3',x_4'\in\mathbb Z^2}\bar W^{(h),\infty}_{4,0,\ul0;\ul\o}(x_1,x_2',x_3',x_4')\Big].\label{R4barV}\end{eqnarray}
%$\mathcal R \bar V_4^{(h)}(\psi)$ be the analogue of \eqref{R4V}, with $\bar W^{(h)}_{4,0,\ul0;\ul\o}$ (resp. $\bar W^{(h),\infty}_{4,0,\ul0;\ul\o}$) replacing 
%$W^{(h)}_{4,0,\ul0;\ul\o}$ (resp. $W^{(h),\infty}_{4,0,\ul0;\ul\o}$). 
In this appendix, we limit our discussion to $\mathcal R \bar V_4^{(h)}(\psi)$, the `complementary term', $\mathcal R (V_4^{(h)}(\psi)- \bar V_4^{(h)}(\psi))$, 
being treatable similarly\footnote{The minor technical complication arising in the `complementary' case is that, if we restrict our attention to one of the trees contributing to $\mathcal R (V_4^{(h)}(\psi)- \bar V_4^{(h)}(\psi))$, the action of $\mathcal R$ on 
  a vertex $v_1$, if non trivial, can interfere with the one on a vertex $v<v_1$ preceding it. Such interference does not cause any conceptual extra difficulty, but it complicates the explicit form of
  the corresponding tree values  must be expressed in an inductive form, rather than by a formula as explicit as \eqref{g5.50s}. 
  For a discussion of these issues, see e.g. \cite[Sections 3.3 and 3.4]{BM-XYZ}.}. A convenient fact is that, if $\t\in \bar{\mathcal T}^{(h)}_{N,0}$, then 
$W_{\t,{\bf P},{T}}(\xx_{v_0})$ has the following explicit expression: 
\bea && W_{\t,{\bf P},T}(\xx_{v_0})=\Big[\prod_{v\ {\rm not}\ {\rm e.p.}}(1+z_{h_v})^{-|P_v^\psi|/2}\Big]
\Big[\prod_{v\ {\rm e.p.}} K^{(h_v)}_{v}(\xx_{v})\Big]\times\nonumber\\
&&\ \times
\Big\{\prod_{v\ {\rm not}\ {\rm e.p.}}\frac1{s_v!} \int
dP_{T_v}({\bf t}_v) \det (M^{h_v,T_v}({\bf t}_v))
\Big[\prod_{\ell\in T_v} g^{(h_v)}_\ell\Big]\Big\}\;,\label{g5.50s}\eea
where the notations are analogous to \cite[eq.(6.63)]{GMT17a}, to which we refer  for details (in particular, $M^{h_v,T_v}({\bf t}_v)$ is a matrix whose elements are propagators on scale $h_v$, like the one defined in 
\cite[Lemma 3]{GMT17a}). The infinite volume limit of $\bar W^{(h)}_{4,0,\ul0;\ul\o}$, denoted by $\bar W^{(h),\infty}_{4,0,\ul0;\ul\o}$, admits the same explicit expression as $\bar W^{(h)}_{4,0,\ul0;\ul\o}$, modulo the following 
changes: the sum over the coordinates in ${\bf x}_{v_0}$ in (the analogue of) \eqref{RV4tree} runs over $\mathbb Z^2$, rather than over $\L$; 
all the propagators appearing in \eqref{g5.50s} (both those in the elements of $M^{h_v,T_v}$ and those in the last product) should be replaced by their infinite volume limits. 

We recall that, if $\t\in \bar {\mathcal T}^{(h)}_{N,0}$ and the RCC satisfy \eqref{ggg0}, by using \eqref{g5.50s}, the Gram-Hadamard bound on $\det (M^{h_v,T_v}({\bf t}_v))$ (see \cite[Eq.(6.60)]{GMT17a}) and  the dimensional bound \eqref{eq:bbpprrh} on the propagators, we find
\be \|W_{\t,{\bf P},T}\|_{\kappa,h}\le (C\e)^{\max\{N,c|I_{v_0}^\psi|\}}\prod_{\substack{v\,{\rm not}\\ {\rm e.p.}}}\hskip-.02truecm\frac{C^{\sum_{i=1}^{s_v}|P_{v_i}|-|P_v|}}{s_v!} 2^{2-\frac{1-\e}2|P_v^\psi|}\;,
\label{goodb}\ee
which is the analogue of \eqref{5.62} (the factors $z(P_v)$ are absent because  $\mathcal R$ acts on none of the   the vertices  $v>v_0$; this bound on ``non-renormalized trees'' has been discussed in several 
previous papers, see e.g. \cite[Section 6]{GMreview}). 
After summation over $\t,{\bf P},T$, this leads to the bound $\|\bar W^{(h)}_{4,0,\ul0;\ul\o}\|_{\kappa,h}\le C\e$, uniformly in $L$: in particular, the bound applies 
to the kernel of $\mathcal R\bar V_4^{(h)}$, simply because it applies separately to $\bar W^{(h)}_{4,0,\ul0;\ul\o}$ and to $\bar W^{(h),\infty}_{4,0,\ul0;\ul\o}$.

On the other hand, by suitably taking into account cancellations between the two terms in the right side of \eqref{R4barV}, one can find an improved bound on
$\mathcal R\bar V_4^{(h)}$, which we now discuss. 

\medskip

{\bf 1.} Let us first consider the terms in the right side of \eqref{R4barV} such that either the argument of $\bar W^{(h)}_{4,0,\ul0;\ul\o}$, $(x_1,x_2,x_3,x_4)$, or the argument of 
$\bar W^{(h),\infty}_{4,0,\ul0;\ul\o}$, $(x_1,x_2',$ $x_3',x_4')$, have tree-distance (i.e. length of the shortest tree on $\L$ including the four points) larger than $L/4$ (this is the first `finite size correction' that we intend to discuss in this appendix). 
Recall that each of the trees contributing to these kernels comes with a 
a product of propagators `along the spanning tree', see the factor $\prod_{v\,not\,e.p.}\prod_{\ell\in T_v} g^{(h_v)}_\ell$ in the right side of \eqref{g5.50s}. Therefore, 
by using the stretched exponential decay of the propagators 
in \eqref{eq:bbpprrh}, we find that each of these contributions can be bounded by the right side of \eqref{goodb} times an additional, exponentially small, 
factor $e^{-(\kappa/4)\sqrt{2^hL}}=e^{-(const.)2^{(h-h_L)/2}}$. {This factor is smaller than any power of the inverse of the exponent, in particular, smaller than 
(const.)$2^{h_L-h}$, which is enough to renormalize the quartic kernels under consideration.}

\medskip

{\bf 2.} After having estimated the terms in the previous item, we are left with the terms with tree distance smaller than $L/4$, which can be rewritten as 
\begin{eqnarray} &&
\sum_{\substack{x_1\in\L\\ \o_1,\ldots,\o_4}}\ \sum_{\substack{x_2,x_3,x_4:\\ d(x_1,\ldots,x_4)<L/4}}\hskip-.3truecm
\Big[\f^+_{x_1,\o_1}\f^-_{x_2,\o_2}\f^+_{x_3,\o_3}\f^-_{x_4,\o_4}\bar W^{(h)}_{4,0,\ul0;\ul\o}(x_1,x_2,x_3,x_4)\nonumber\\
&&\qquad-
\f^+_{x_1,\o_1}\f^-_{x_1,\o_2}\f^+_{x_1,\o_3}\f^-_{x_1,\o_4}\bar W^{(h),\infty}_{4,0,\ul0;\ul\o}(x_1,x_2,x_3,x_4)\Big].\label{R4Va}\end{eqnarray}
In the first line we rewrite 
\be \label{barw}\bar W^{(h)}_{4,0,\ul0;\ul\o}(x_1,x_2,x_3,x_4)= \bar w^{(h)}_{4,0,\ul0;\ul\o}(x_1,x_2,x_3,x_4)+\bar W^{(h),\infty}_{4,0,\ul0;\ul\o}(x_1,x_2,x_3,x_4),\ee
% with \be \label{barw}\bar w^{(h)}_{4,0,\ul0;\ul\o}(x_1,x_2,x_3,x_4)=\bar W^{(h)}_{4,0,\ul0;\ul\o}(x_1,x_2,x_3,x_4)-\bar W^{(h),\infty}_{4,0,\ul0;\ul\o}(x_1,x_2,x_3,x_4),\ee 
so that 
\begin{eqnarray} && \hskip-1.truecm\eqref{R4Va}=
\sum_{\substack{x_1\in\L\\ \o_1,\ldots,\o_4}}\ \sum_{\substack{x_2,x_3,x_4:\\ d(x_1,\ldots,x_4)<L/4}}\hskip-.3truecm
\Big[\f^+_{x_1,\o_1}\f^-_{x_2,\o_2}\f^+_{x_3,\o_3}\f^-_{x_4,\o_4}\bar w^{(h)}_{4,0;\ul\o}(x_1,x_2,x_3,x_4)\nonumber\\
&&\qquad +\big(\prod_{i=1}^4\f^{\e_i}_{x_i,\o_i}-\prod_{i=1}^4\f^{\e_i}_{x_1,\o_i}\big)
\bar W^{(h),\infty}_{4,0,\ul0;\ul\o}(x_1,x_2,x_3,x_4)\Big].\label{R4Vb}\end{eqnarray}
In the two products in the second line, $\e_i:=(-1)^{i-1}$; notice that the Grassmann variables in the first product are computed at $x_i$, while in the second product they are computed at $x_1$.

The term in the second line is the `usual', infinite volume, renormalized term, which can be treated as discussed in, e.g., \cite[Section 6.1.4]{GMT17a}; we refer to that section for a discussion of why these terms have the `usual' dimensional gains leading to the factors $2^{-z(P_v)}$ in \eqref{5.62}. 
The term in the first line is, instead, 
the second `finite size correction' that we intend to discuss in this appendix. By using the representation of 
$\bar W^{(h)}_{4,0,\ul0;\ul\o}$ in terms of a tree expansion, we find that $\bar w^{(h)}_{4,0,\ul0;\ul\o}$ itself can be written as a sum over trees. Each tree comes 
with a difference between a sum over $\xx_{v_0}$ (within $\L$) of the tree value in \eqref{g5.50s} and a sum over $\xx_{v_0}$ (extended to the whole $\mathbb Z^2$) of the infinite volume limit of 
\eqref{g5.50s}. We further split this difference in two parts: the first corresponds to the case where both the sums over $\xx_{v_0}$ involve at least one coordinate at a distance larger than $L/3$ from 
$(x_1,x_2,x_3,x_4)$; by proceeding as in item {\bf 1}, we find that this first part has a bound that is better than \eqref{goodb} by a factor  $e^{-(const.)\sqrt{2^hL}}$, as desired. 
The second part corresponds to the case where we sum the difference between the tree value in \eqref{g5.50s} and its infinite volume counterpart over 
coordinates $\xx_{v_0}$ that are all closer than $L/3$ to $(x_1,x_2,x_3,x_4)$. 
In the finite volume expression of the tree value, \eqref{g5.50s}, we replace every finite volume propagator $g^{(h)}_\o(x,y)$ appearing either in the matrices $M^{h_v,T_v}({\bf t}_v)$ 
or in the products over spanning trees $\prod_{\ell\in T_v} g^{(h_v)}_\ell$ by the following infinite linear combination of infinite volume propagators, namely  (``Poisson summation formula'', see e.g. \cite[Eq. (A.8)]{GMT17a}):
$$g^{(h)}_\o(x,y)=\sum_{n\in\mathbb Z^2}(-1)^{n\cdot \bt}g^{(h),\infty}_\o(x+nL,y)\equiv g^{(h),\infty}_\o(x,y)+\d g^{(h)}_\o(x, y)$$
where $g^{(h)}_\o$ is as in \eqref{eq:ghjy6}, while $g^{(h),\infty}_\o$ is the same
expression where $1/L^2$ times the sum over $k\in\mathcal P'_\o(\bt)$ is replaced by
$(2\pi)^{-2}\int_{[-\pi,\pi]^2}dk$.
By using this decomposition, the difference between the tree value in \eqref{g5.50s} and its infinite volume counterpart can be re-expressed as a sum of terms, each of which involves at least one 
`remainder propagator' $\d g^{(h)}_\o(x,y)$. Note that, by construction, any pair of sites $x,y$ involved in the expression under consideration is closer than $L/3$: therefore, using \eqref{eq:bbpprrh}, 
\be \label{eq:bbppddh}|\d g^{(h)}_\o(x,y)|\le C 2^h e^{-\kappa \sqrt{2^h L}}. \ee
Putting things together, we find that also this second part has a bound that is better than \eqref{goodb} by a factor  $e^{-(const.)\sqrt{2^hL}}$, as desired. 

\section{Finite size corrections to the partition function}
\label{app:grigliabis} 
In this section, we prove \eqref{yaa}, which is equivalent to the fact that
\be  E^{(h_L)}-E^{(0)}=\D(\l)+L^{-2}\log(1+s_\bt(\l))-2L^{-2}\log Z_{h_L},\label{eq:b1}\ee
with $\D(\l)$ independent of $L,\bt$ and such that $|\D(\l)|\le C|\l|$, $|s_\bt(\l)|\le C|\l|$ uniformly in $L,\bt$ and $Z_{h_L}$ as in \eqref{eq:6.84}, and that \eqref{eq:692} holds, 
% \be \bullet\ V^{(h_L)}(\Psi,0)=L^{-3}\sum_{\o=\pm}u_{\bt,\o}(\l)\hat\Psi^+_{\o}\hat\Psi^-_{\o}+L^{-6}v_\bt(\l)\hat \Psi^+_{+}\hat \Psi^-_{+}\hat \Psi^+_{-}\hat \Psi^-_{-},\label{eq:b2}\ee
with $|u_{\bt,\o}(\l)|, |v_\bt(\l)|\le C|\l|$, uniformly in $L,\bt$. The analogous estimates on the generating function, stated after \eqref{yaaA}, can be derived in a similar way, and are left to the reader.

We start by proving \eqref{eq:692}. One starts from the general  representation  of the effective potential, i.e. \eqref{eq:6.16mo} with the index $(-1)$ replaced by $(h_L)$ and $J\equiv 0$, so that $m=0$. On the other hand, the field $\Psi$ contains only the four modes $\hat\Psi^\pm_{\o}$, so that the sum is limited to $n=2,4$. Moreover, due to the Krokecker delta $\d_{\ul \o}(\ul k, 0)$,
$V^{(h_L)}$ reduces to the simple form
\be V^{(h_L)}(\Psi,0)= L^{-2}\sum_\o \tilde u_{2}\hat\Psi^+_{\o}\hat\Psi^-_{\o}+L^{-6}\tilde u_4\hat \Psi^+_{+}\hat \Psi^-_{+}\hat \Psi^+_{-}\hat \Psi^-_{-},\label{eq:692bis}\ee
for some  constants $u_2,u_4$ depending on $\l,L,\bt$. Using the dimensional estimates (see \eqref{bouW}), it is easy to deduce that     \be | u_n|\le C^n|\l| 2^{h_L(2-n/2)} \ee uniformly in $\bt$,
which implies the desired estimates on $u_{\bt,\o},v_\bt$, because $h_L\sim -\log_2 L$.
                                                           
% and the fact that $\Psi_x^\pm$ has only two modes in momentum space.
% In fact, the general representation of the effective potential implies 
%    On the other hand, recalling that 
%    $\Psi_x^\pm=\sum_{\o}e^{\pm i\bar p^\o x}\Psi^\pm_{x,\o}$ is a Grassmann, anticommuting, variable with 
%    only two non-zero modes in momentum space, $\Psi_x^\pm=L^{-2}\sum_{\o=\pm}\hat \Psi^\pm_{\o} e^{\pm i k^\o_\bt x}$, we 
%    find that all the terms with $n\ge 6$ in the last line of \eqref{eq:6.16moapp} vanish and, therefore, 
%    $V^{(h_L)}(\Psi,0)$ can be written as in \eqref{eq:692}. Moreover, using the bounds on the running coupling constants and on the kernels of the bdirrelevant part,
%    recalling that $2^{h_L}$ is of the order $L^{-1}$ and that $k^\o_\bt$ is at a distance $O(1/L)$ from $\bar p^\o$, 
%    we immediately find that $u_{\bt,\o}(\l)$ and $v_{\bt}(\l)$ are upper bounded by $C|\l|$, uniformly in $\bt$ and $L$, as desired. 

% \medskip

Let us now prove \eqref{eq:b1}. From the multiscale computation of the effective potential, it follows that 
\be E^{(h_L)}-E^{(0)}=\sum_{h_L< h<0}(t_h+\tilde E_h), \label{eq:b6}\ee
where $t_h$ was defined in \eqref{eq:th}, and $\tilde E_h$ is the sum of the vacuum diagrams with smallest scale label equal to $h$, namely 
\be \tilde E_h=L^{-2}\sum_{n\ge 1}\frac1{n!} 
\mathcal E^T_h(\underbrace{\widehat V^{(h)}(\sqrt{Z_{h-1}}\psi',0);\cdots;
\widehat V^{(h)}(\sqrt{Z_{h-1}}\psi',0)}_{n\ {\rm times}},\label{eq:b7}\ee
which can be represented as a sum over trees, see \eqref{eq:6.63e}-\eqref{f5.49ab}. Let us start by discussing the contribution from $t_h$; using the definition \eqref{eq:th}, we rewrite
\bea t_h&=&L^{-2}\sum_\o\sum_{k\in\mathcal P_\o(\bt)}\log\Big(1+\frac{z_h\bar \chi_h(k)\bar D_\o(k)}{\bar D_\o(k)+r_\o(k)/Z_h}\Big)\label{eq:thappb}\\
&-&L^{-2}\sum_\o\log\Big(1+\frac{z_h\bar \chi_h(k^\o_\bt-\bar p^\o)\bar D_\o(k^\o_\bt-\bar p^\o)}{\bar D_\o(k^\o_\bt-\bar p^\o)+r_\o(k^\o_\bt-\bar p^\o)/Z_h}\Big).\nonumber\eea
% Note that the first sum runs over the momenta in $\mathcal P_\o(\bt)$ (rather than over $\mathcal P_\o'(\bt)$, see \eqref{eq:th}), while the second 
% over the two momenta in $\mathcal P_\o(\bt)\setminus \mathcal P_\o'(\bt)$.
Using Poisson summation formula (see e.g. \cite[Eq. (A.8)]{GMT17a}), the first sum in the right side can be rewritten as 
\be \sum_\o\sum_{m\in \mathbb Z^2}(-1)^{\bt\cdot m}\int_{\mathbb R^2}\frac{dk}{(2\pi)^2}\log\Big(1+\frac{z_h\bar \chi_h(k)\bar D_\o(k)}{\bar D_\o(k)+r_\o(k)/Z_h}\Big) e^{iL(k+\bar p^\o)\cdot m}.\ee
The term with $m=0$, which we denote by $t_{0,h}$, is $L,\bt$ independent and satisfies 
\be \label{app:td} |t_{0,h}|\le C|\l|2^{2h}.\ee
To see this, observe that the area of the support of $\bar\chi_h$ is $O(2^{2h})$ and 
recall that  $r_\o(k)=O(k^2)$, that $|z_h|\le C|\l|$ uniformly in $h$ and that $Z_h=O(2^{-\h h})$ (see \eqref{eq:6.84}), with $\eta(\l)$ that tends to zero for $\l\to0$.
The sum of the terms with $m\neq 0$, which we denote by $t_{1,h}$, is bounded from above as
\be \label{app:ts} |t_{1,h}|\le C|\l|2^{2h}e^{-c\sqrt{L2^h}},\ee
the stretched-exponential decay coming from the fact that  the integrand  is a function in the Gevrey class of order $2$, by assumption on $\bar \chi_h$.
Finally, recalling that $\bar\chi_h(k^\o_\bt-\bar p^\o)=1$ for all $h>h_L$ and that  $1+z_h=Z_{h-1}/Z_h$, 
we find that, if $h>h_L$, the sum in the second line of \eqref{eq:thappb} can be rewritten as
\be  \label{app:ts1} -2L^{-2}\log(Z_{h-1}/Z_h)+t_{2,h},\quad |t_{2,h}|\le CL^{-2}|\l|2^{h(1-|\eta|)}.\ee  
% if $h=h_L$, we denote the sum in the second line of \eqref{eq:thappb} by $t_{2,h_L}$, which is bounded by 
% \be  \label{app:ts2} |t_{2,h_L}|\le C|\l|2^{2h_L}.\ee  
Putting things together we write:
\be \sum_{h_L<h<0}t_h=-2L^{-2}\log Z_{h_L}+\sum_{h<0}t_{0,h}+\Big[\sum_{h_L< h<0} (t_{1,h}+t_{2,h})-\sum_{h\le h_L}t_{0,h}\Big].\ee
The second term in the right side contributes to $\D(\l)$: it is $L,\bt$ independent and, thanks to \eqref{app:td}, it is bounded by $C|\l|$. The term in brackets 
contributes to $L^{-2}\log(1+s_\bt(\l))$: thanks to \eqref{app:td} and \eqref{app:ts1}, it is bounded by $CL^{-2}|\l|$, as we wanted. 

\medskip

We are left with the sum over scales of $\tilde E_h$, see \eqref{eq:b6}-\eqref{eq:b7}. As mentioned after \eqref{eq:b7}, $\tilde E_h$ can be written as a sum over trees,
\be \tilde E_h= \sum_{N\ge 1}\sum_{\t\in \mathcal T_{N,0}^{(h)}} E(\t),\ee
where $E(\t)$, $\t\in\mathcal T^{(h)}_{N,0}$, is bounded as in \eqref{5.62}, with $|P_{v_0}^\psi|=|P_{v_0}^J|=|{\bf q}|=0$, namely 
\be |E(\t)|\le (C|\l|)^{\max\{1,cN\}}2^{2h}\prod_{\substack{v\,{\rm not}\\ {\rm e.p.}}}\hskip-.02truecm\frac{C^{\sum_{i=1}^{s_v}|P_{v_i}|-|P_v|}}{s_v!} \;2^{C|\l||P^\psi_v|}2^{2-\frac12|P_v^\psi|-|P_v^J|-z(P_v)}.
\label{app:bound}\ee
We now rewrite $\tilde E_h$ as a sum of two terms: the first,  which we denote by $\tilde E_{0,h}$,  is the sum over trees of the thermodynamic limit of the tree values (where sums over lattice points in $\L$ are replaced by sums on $\mathbb Z^2$ and single-scale propagators $g^{(h')}_\o$ are replaced by their infinite-volume counterparts $g^{(h'),\infty}$). The second is the finite-size remainder,
which we denote by $\tilde E_{1,h}$. By construction, $\tilde E_{0,h}$ is $L,\bt$ independent, and it is bounded by the sum over trees of the right side of \eqref{app:bound}, which gives
\be |\tilde E_{0,h}|\le C|\l|2^{2h}\label{eq:e0}.\ee
The finite size remainder admits an improved dimensional bound of the form 
\be |\tilde E_{1,h}|\le C|\l|2^{2h}e^{-c\sqrt{L2^h}},\label{eq:e1}\ee
which can be proved via discussion analogous to the one after \eqref{R4Vb} on the bound on the finite size contribution to the local quartic kernel $w_{4,0;\ul\o}^{(h)}=W^{(h)}_{4,0,\ul0;\ul\o}-W^{(h),\infty}_{4,0,\ul0;\ul\o}$; details are left to the reader.
By using the decomposition $\tilde E_h=\tilde E_{0,h}+\tilde E_{1,h}$, we rewrite
\be \sum_{h_L\le h<0}\tilde E_h=\sum_{h<0}\tilde E_{0,h}+\Big[\sum_{h_L\le h<0}\tilde E_{1,h}-\sum_{h<h_L} E_{0,h}\Big].\ee
The first term in the right side contributes to $\D(\l)$: it is $L,\bt$ independent and, thanks to \eqref{eq:e0}, it is bounded by $C|\l|$. The term in brackets 
contributes to $L^{-2}\log(1+s_\bt(\l))$: thanks to \eqref{eq:e1}, it is bounded by $CL^{-2}|\l|$, as desired. This concludes the proof of \eqref{eq:b1}, with the desired bounds on $\D(\l)$, $s_\bt(\l)$. 

  \section{Two technical results on the non-interacting model}
\subsection{Proof of \eqref{griglia}}
\label{app:griglia}

It is sufficient to prove the claim when
$\bt-\bt'$ equals either $(1,0)$ or $(0,1)$ and,
for definiteness, assume we are in the former case. Also, without loss
of generality, assume that
$|k^\pm_{\bt}- \bar p^\pm|\le |k^\pm_{\bt'}-
\bar p^\pm|$. From the definition \eqref{eq:Dtt} of $\mathcal P(\bt)$ we see that
\begin{eqnarray}
  \label{eq:fatica2}
\frac\pi {2L}\le |k^\pm_{\bt'}-\bar p^\pm|\le \frac{\sqrt 2\pi}L ,  
\end{eqnarray}
while $|k^\pm_{\bt}- \bar p^\pm|$ can be much smaller, possibly zero.
Write 
\begin{eqnarray}
\label{ftb}
  \frac{\tilde Z^0_\bt}{\tilde Z^0_{\bt'}}= \frac{\mu_0(k^+_{\bt'})\mu_0(k^-_{\bt'})}{\mu_0(k^+_{\bt})\mu_0(k^-_{\bt})} %\exp\left\{
  e^{
\sum_{k\in \mathcal P({\bt})}\left(\log \mu_0(k)-\frac12\log \mu_0(k^{\leftarrow})-\frac12\log \mu_0(k^{\rightarrow})\right)}
 % \right\},
\end{eqnarray}
with $k^\rightarrow=k+(\pi/L,0)\in \mathcal P(\bt')$ and
$k^\leftarrow=k-(\pi/L,0)\in \mathcal P(\bt')$. Decompose
$\mathcal P(\bt)$ as the disjoint union $A\cup B$, with $A$
containing the values of $k$ at distance at most, say, $10/L$ from either
$ \bar p^+$ or $ \bar p^-$, and $B$ all the others. The cardinality of $A$ is uniformly
bounded as a function of $L$.

Note that for all $k\in A$, $|\mu_0(k^{\leftarrow})|$ and
$|\mu_0(k^{\rightarrow})|$ are upper and lower bounded by positive
constants times $1/L$, because $\mu_0$ vanishes linearly at $\bar p^\pm$ and
the values of $k^{\rightarrow},k^{\leftarrow}$ are at distance of order
$1/L$ from $\bar p^\pm$ (cf. \eqref{eq:fatica2}).  The same holds for
$|\mu_0(k)|, k\in A$, except possibly for $k=k^\pm_{\bt}$.
One has then
\begin{equation}
  c_1\le \left|\frac{\mu_0(k^+_{\bt'})\mu_0(k^-_{\bt'})}{\mu_0(k^+_{\bt})\mu_0(k^-_{\bt})} %\exp\left\{
    e^{
\sum_{k\in A}\left(\log \mu_0(k)-\frac12\log \mu_0(k^{\leftarrow})-\frac12\log \mu_0(k^{\rightarrow})\right)}
%\right\}
\right|\le c_2.
\end{equation}
It remains to prove that the sum in \eqref{ftb}, with $k$ restricted to $B$, is upper and lower bounded (in absolute value) by $L$-independent positive constants.
Write
\begin{eqnarray}
  \log \mu_0(k)-\frac12\log \mu_0(k^{\leftarrow})-\frac12\log \mu_0(k^{\rightarrow})\\=-\frac{\pi^2}{L^2}\partial^2_{k_1}\log \mu_0(k)-\frac {\pi^3}{6 L^3}\partial^3_{k_1}\log \mu_0(k)|_{k=k'}
\end{eqnarray}
where $k'$ is a point
in the segment joining $k^\leftarrow$ and $k^\rightarrow$.  Since
$\mu_0(\cdot)$ vanishes linearly at $\bar p^\pm$,
  \[|\partial^3_{k_1}\log \mu_0(k')|=O((\min(|k-\bar p^+|,|k-\bar p^-|)^{-3}).\]
  Here it is important that $k\in B$, since this means that
  $\partial^3_{k_1}\log \mu_0(k')$, computed in the unknown point $k'$, can be safely
  replaced by the derivative computed at $k$.  Therefore,
\begin{eqnarray}
\frac1{L^3}  \sum_{k\in B}\partial^3_{k_1}\log \mu_0(k')=O(1).
\end{eqnarray}
The sum of the term involving $\partial^2_{k_1}\log \mu_0(k)$ requires more care since at first sight it diverges like $\log L$.
However, write
\begin{eqnarray}
\frac{\pi^2}{L^2}\partial^2_{k_1}\log \mu_0(k)=\frac14\int_{Q_k}\partial^2_{q_1}\log \mu_0(q)dq+O(L^{-3}|\partial^3_{k_1}\log \mu_0(k)|),
\end{eqnarray}
with $Q_k$ the square of side $2\pi/L$ centered at $k$. Therefore,
the sum in \eqref{ftb}, with $k$ restricted to $B$,
plus the integral
\begin{eqnarray}
  \label{eq:intdiv}
\frac14  \int_{[-\pi,\pi]^2\setminus(N^+\cup N^-)}dk\; \partial^2_{k_1}\log \mu_0(k),
\end{eqnarray}
 with $N^\pm$ the neighborhood of radius $10/L$ around $\bar p^\pm$, 
 is upper and lower bounded in absolute value by  positive constants.
 
 The integral \eqref{eq:intdiv} has a finite limit as
 $L\to\infty$. Indeed, since (cf. \eqref{eq:6.4}-\eqref{eq:6.5})
 $\mu_0( \bar p^\o+k')=\bar \alpha_\o k'_1+\bar \beta_\o k'_2+O(|k'|^2)$, the
 possibly singular part of the integral is proportional to
\begin{eqnarray}
  \int\frac{dk}{(\bar \alpha_\o k_1+\bar\beta_\o k_2)^2}{\bf 1}_{\{(10/L)\le |k|\le 1\}}.
\end{eqnarray}
 
We make the change of variables $q_1=\o (\bar \alpha^1 k_1+\bar \beta^1 k_2),q_2=(\bar \alpha^2 k_1+\bar \beta^2 k_2)$,
where $\bar \alpha^j,\bar \beta^j$
were defined in \eqref{eq:a1b1}. The Jacobian matrix $A_\o$ has  non-zero determinant (this is because, as observed in
Remark \ref{rem:platonico}, the ratio $\alpha_\omega/\beta_\omega$ is not real so that the same holds for
$\bar \alpha_\o/\bar\beta_\o$ if $\lambda$ is small enough). Then, the integral becomes
\begin{eqnarray}
  \det(A_\o)  \int \frac{dq}{(q_1+i q_2)^2}{\bf 1}_{\{(10/L)\le |A_\o q|\le 1\}}\\=\det(A_\o)  \int \frac{dq}{(q_1+i q_2)^2}{\bf 1}_{\{(10/L)\le | q|\le 1\}}
  +O(1)=O(1).
\end{eqnarray}
In the first equality we used the fact that the symmetric difference
between the balls of radius $10/L$  for $q$ and for $A_\o q$ has area of order
$L^{-2}$, while the integrand is $O(L^2)$ there; in the second step, we
noted that the integral is zero, using the symmetry
$(q_1,q_2)\leftrightarrow (q_2,-q_1)$.

\subsection{Proof of \eqref{eq:infradito}}
\label{app:infradito}
Recall that the values of $c_\bt$ are given in \eqref{eq:ctt}.
Further, note that if  $k\in \mathcal P(\bt)$, then also $(\pi,\pi)-k\in\mathcal P(\bt)$; if these two momenta are distinct, then they contribute $\mu_0(k)\mu_0((\pi,\pi)-k)=|\mu_0(k)|^2\ge0$ to the product $Z^0_\bt$. Here, we used the symmetry \eqref{eq:simmmu0}.
Also, unless
\begin{eqnarray}
  \label{eq:4k}
k=(\epsilon_1\pi/2,\epsilon_2\pi/2), \quad \epsilon_1=\pm1, \epsilon_2=\pm1,
\end{eqnarray}
  one has that $(\pi,\pi)-k\ne k\mod (2\pi,2\pi)$. To determine the sign
  of $Z^0_\bt$, it is therefore sufficient to determine whether the
  momenta \eqref{eq:4k} belong to $\mathcal P(\bt)$.  The four momenta \eqref{eq:4k} belong to
  $\mathcal P((0,0))$ if $L=0 \mod 4$ and to $\mathcal P((1,1))$ if
  $L=2\mod 4$.  Also, note that
\begin{multline}
  \label{eq:prod4}
\prod_{\epsilon_1=\pm1}\prod_{\epsilon_2=\pm1}  \mu_0(\epsilon_1\pi/2,\epsilon_2\pi/2)
=
\prod_{\epsilon_1=\pm1}\prod_{\epsilon_2=\pm1}  \mu(\epsilon_1\pi/2,\epsilon_2\pi/2)\\
  =(t_1-t_2+t_3+1)(t_1-t_2-t_3-1)(t_1+t_2-t_3+1)(t_1+t_2+t_3-1).
\end{multline}
To get the first equality, observe first that $p^\o$ cannot equal any
of the four momenta \eqref{eq:4k}, otherwise one would have
$p^+=p^-\mod (2\pi,2\pi)$, which is excluded by Assumption
\ref{ass:liquid} on the edge weights.  The same is true for
$\bar p^\o$ provided $\lambda$ sufficiently small, as
$\bar p^\o=p^\o+O(\lambda)$. Then, the first equality  in \eqref{eq:prod4} follows by
assuming that the support of the cut-off function $\bar \chi(\cdot)$
in \eqref{mu0} is sufficiently small (this can be guaranteed by choosing the constant $c_0$, that enters the definition of
$\bar\chi(\cdot)$, to be small enough). Finally, the last product in
\eqref{eq:prod4} is strictly negative, as follows from Remark \ref{rem:liquid}.  Wrapping up, one
has that
\begin{eqnarray}
  \sign(Z_\bt^0)=\left\{
  \begin{array}{lll}
    +1&\text{if}& \bt=(0,1) \text{ or } \bt=(1,0)\\
    (-1)^{{\bf 1}_{L=0\!\!\!\!\mod 4}} &\text{if}& \bt=(0,0) \\
    (-1)^{{\bf 1}_{L=0\!\!\!\!\mod 2}} &\text{if}& \bt=(1,1)
  \end{array}
                                                   \right..
\end{eqnarray}
In other words, $\sign(Z_\bt^0)= c_\bt$ and the claim follows.

\bigskip

{\bf Acknowledgements} We would like to thank Ron Peled and Jean-Marie St\'ephan for fruitful discussions on the 6-vertex model and the corresponding scaling relations. {We gratefully thank Rafael Greenblatt for carefully reading the manuscript and for several constructive suggestions on how to improve it.} 
This work has been supported by the European Research Council (ERC) under the European Union's Horizon 2020 research and innovation programme (ERC CoG UniCoSM, grant agreement n.724939). F.T.  was  partially  supported  by
the  CNRS  PICS  grant  151933, by ANR-15-CE40-0020-03 Grant LSD, 
ANR-18-CE40-0033 Grant DIMERS and  by Labex MiLyon (ANR-10-LABX-0070). This work was started during a long-term stay of A.G. at Univ. Lyon-1, co-funded by Amidex and CNRS, which are gratefully acknowledged.

\end{document}